\documentclass[12pt,a4paper]{article}
\pdfoutput=1 

\usepackage[utf8]{inputenc}
\usepackage[english]{babel}

\usepackage{cite}
\usepackage{graphicx}
\usepackage{tabularx,subcaption,booktabs}
\usepackage{array}
\usepackage{longtable}
\usepackage{caption}
\usepackage{comment}
\usepackage[affil-it]{authblk}
\usepackage{multicol}
\newcolumntype{Y}{>{\centering\arraybackslash}X}  

\usepackage{amsmath, amsthm, amsfonts, amssymb}
\usepackage{bbold}
\usepackage{mathrsfs}
\usepackage{mathbbol}
\usepackage{extpfeil}
\usepackage{bbm}
\usepackage{physics}
\usepackage{braket}
\usepackage{phaistos}
\usepackage{commath}
\def\ket#1{{|{#1}\rangle}}

\usepackage{float, adjustbox}
\usepackage{subcaption}
\usepackage{lmodern}
\usepackage{overpic}
\usepackage{footmisc}

\usepackage{tikz}
\usetikzlibrary{decorations.pathreplacing, shapes.symbols, fit, calc, positioning, matrix, backgrounds}
\usepackage{esvect}

\usepackage{calc} 
\tikzset{stdnode/.style={draw,  inner sep=0.7pt,  minimum size=8pt,  line width=0.6pt}}
\definecolor{hepblue}{RGB}{0, 114, 178}       
\definecolor{hepblue-contrary}{RGB}{255, 141, 77}

\definecolor{hepteal}{RGB}{0, 158, 115}
\definecolor{hepteal-contrary}{RGB}{255, 97, 140}

\definecolor{vermillion}{RGB}{213, 94, 0}
\definecolor{vermillion-contrary}{RGB}{42, 161, 255}

\tikzset{
  Cnode/.style={stdnode,  circle,  fill=hepblue,   text=white}, 
  Inode/.style={stdnode,  circle,  fill=hepteal,  text=white}, 
  Jnode/.style={stdnode,  circle,  fill=vermillion,  text=white}, 
  Knode/.style={stdnode,  circle,  fill=white,    text=black}, 
  CnodeBW/.style={stdnode,  circle,  fill=white}, 
  InodeBW/.style={stdnode,  regular polygon,  regular polygon sides=4,  fill=white}, 
  JnodeBW/.style={stdnode,  diamond,  fill=white}, 
  KnodeBW/.style={stdnode,  circle,  fill=white}
}

\tikzset{
  Anode/.style={stdnode,  circle,  fill=hepblue-contrary,   text=white}, 
  Bnode/.style={stdnode,  circle,  fill=hepteal-contrary,  text=white}, 
  CnodeABC/.style={stdnode,  circle,  fill=vermillion-contrary,  text=white}, 
  AnodeBW/.style={stdnode,  circle,  fill=white}, 
  BnodeBW/.style={stdnode,  regular polygon,  regular polygon sides=4,  fill=white}, 
  CnodeABCBW/.style={stdnode,  diamond,  fill=white}, 
}

\usetikzlibrary{shapes.geometric,shapes.symbols}

  \tikzset{
    BorderStar/.style={
      draw,
      thick,
      star,
      star points=5,
      minimum size=3pt,  
      inner sep=0.1pt     
    },
    BorderTriangle/.style={
      draw,
      thick,
      regular polygon,
      regular polygon sides=3,
      minimum size=1pt, 
      inner sep=0pt
    },
    BorderSquare/.style={
      draw,
      thick,
      regular polygon,
      regular polygon sides=4,
      minimum size=3pt,
      inner sep=0.1pt
    },
  }

\usepackage[toc,page]{appendix}

\usepackage{enumitem}
\setlistdepth{9}
\renewlist{enumerate}{enumerate}{9}
\setlist[enumerate,1]{label=\arabic*.}
\setlist[enumerate,2]{label=\alph*.}
\setlist[enumerate,3]{label=\roman*.}
\setlist[enumerate,4]{label=\Alph*.}
\setlist[enumerate,5]{label=\Roman*.}
\setlist[enumerate,6]{label=$\bullet$}
\setlist[enumerate,7]{label=$\square$}
\setlist[enumerate,8]{label=$\blacksquare$}
\setlist[enumerate,9]{label=$\circ$ }

\usepackage{xcolor}
\usepackage[
  colorlinks=true,
  urlcolor=blue,
  anchorcolor=blue,
  citecolor=blue,
  filecolor=blue,
  linkcolor=blue,
  menucolor=blue,
  pagecolor=blue,
  linktocpage=true,
  pdfproducer=medialab,
  pdfa=true
]{hyperref}

\definecolor{color1}{HTML}{A60303}
\definecolor{color2}{HTML}{0A4D92}

\usepackage{jheppub}
\usepackage{environ}
\NewEnviron{myequation}{%
  \begin{equation*}
    \scalebox{1.7}{$\BODY$}
  \end{equation*}
}

\newtheorem{proposition}{Proposition}
\newtheorem{conjecture}{Conjecture}

\newtheorem{observation}{Observation}

\newcommand{\col}[1]{\text{Col}\left(#1 \right)}

\newcommand{\spanv}[1]{\langle #1 \rangle}

\newcommand{\stab}[1]{\text{Stab}\left(#1\right)}

\newcommand{\ra}[1]{\text{rank}\left(#1\right)}
\newcommand{\rAz}[1]{\text{rank}_{\mathbb{Z}_2}\left(#1\right)}
\newcommand{\dimension}[1]{\text{dim}\left\{#1\right\}}
\newcommand{\Hil}{\mathcal{H}}

\newcommand{\bleed}{0pt}

\newcommand{\shadecell}[2]{%
  \fill[gray!60]
    ($ (m-#1-#2.north west)+(-\bleed,\bleed)$)
    rectangle
    ($ (m-#1-#2.south east)+(\bleed,-\bleed)$);}

\newcommand{\outlineblock}[4]{%
  \draw[line width=1.5pt]
    (m-#1-#3.north west) rectangle (m-#2-#4.south east);}
    
\newcommand{\colborder}[1]{%
  \pgfmathtruncatemacro{\lastrow}{\pgfmatrixcurrentrow}%
  \draw[
    line cap=butt,
    shorten <=0.5\pgflinewidth, shorten >=0.5\pgflinewidth,
    preaction={draw=white, line width=1.5pt},
    line width= 0.7pt                              
  ]
    (m-1-#1.north west) rectangle (m-\lastrow-#1.south east);%
}

\newcommand{\rowborder}[1]{%
  \pgfmathtruncatemacro{\lastcol}{\pgfmatrixcurrentcolumn}%
  \draw[
    line cap=butt,
    shorten <=0.5\pgflinewidth, shorten >=0.5\pgflinewidth,
    preaction={draw=white, line width=1.5pt},
    line width=0.7pt
  ]
    (m-#1-1.north west) rectangle (m-#1-\lastcol.south east);%
}

\newcommand{\rowborderSecond}[1]{%
  \draw[line width=.7pt]
    (m-#1-1.north west) rectangle (m-#1-5.south east);}

\newcommand{\colborderSecond}[1]{%
  \draw[line width=.7pt]
    (m-1-#1.north west) rectangle (m-4-#1.south east);}

\newcommand{\outlineblockSecond}[4]{%
  \draw[line width=1.5pt]
    (m-#1-#3.north west) rectangle (m-#2-#4.south east);}

\usepackage[most]{tcolorbox}

\newlength{\graphcardheight}
\newtcbox{\graphpairbox}{%
  enhanced,
  colback=white,
  colframe=black!60,
  arc=3pt,
  boxrule=0.4pt,
  left=0pt,right=0pt,
  top=6pt,bottom=0pt,
  height=\graphcardheight,
  valign=center,
}

\newcommand{\panelcaption}[2]{\footnotesize (#1) #2.}
    
\title{Monogamy of Mutual Information in Graph States}

\author[a]{Jesus Fuentes,}
\author[a]{Cynthia Keeler,}
\author[b]{William Munizzi,}
\author[c,d]{and Jason Pollack}

\affiliation[a]{Department of Physics, Arizona State University,\\ Tempe, AZ 85281, USA}
\affiliation[b]{Division of Physical Sciences, College of Letters and Science, University of California, Los Angeles}
\affiliation[c]{Department of Electrical Engineering and Computer Science, Syracuse University, NY 13210, USA}
\affiliation[d]{Institute for Quantum \& Information Sciences, Syracuse University, NY 13210, USA}

\emailAdd{jesusfuentesun@gmail.com}
\emailAdd{keelerc@asu.edu}
\emailAdd{wmunizzi17@ucla.edu}
\emailAdd{jasonpollack@gmail.com}

\abstract{
The monogamy of mutual information (MMI) is a quantum entropy inequality that enforces the non-positivity of tripartite information. We investigate the failure of MMI in graph states as a forbidden-subgraph phenomenon, conjecturing that every MMI-violating graph state is local-Clifford equivalent to one whose graph contains a four-star subgraph. We construct a family of star-like graphs whose states fail a specific class of MMI instances, and extend this analysis to general star topologies. Deriving adjacency matrix constraints that fix the MMI evaluation for these instances and interpreting them physically, we prove the forbidden-subgraph conjecture for this family of graphs. Finally, through an exhaustive search over graph representatives for all $8$-qubit stabilizer entropy vectors, we establish that MMI failure is not reducible to the cases within our scope.
}

\begin{document}
\maketitle


\section{Introduction}

An important insight from quantum information theory is that entanglement is fully characterized by not just a single numerical value, but also by \textit{how} it is distributed among the constituent parties of a quantum state. In multipartite systems, this pattern of entanglement precisely constrains how information can be stored, transmitted, and processed using the state~\cite{WalterGrossEisert2016}. Entanglement monotones~\cite{VidalWerner2002,Horodecki2009}, such as entanglement entropy, quantify and bound the overall strength of bipartite correlations; however, the entanglement structure of a state is further restricted by sets of fundamental inequalities which dictate how quantum information can be shared across the different subsystems~\cite{Coffman2000,Linden2002}. Among these, the best known are the subadditivity~\cite{LiebRuskai1973} and strong subadditivity~\cite{Lieb1973a,Lieb1973b} inequalities, which are satisfied by all quantum states. Beyond these universal quantum inequalities, additional constraints exist~\cite{Linden2013,ZhangYeung1997,Lashkari:2014kda,ZhangYeung1998,Matus2007,DoughertyFreilingZeger2007,LindenWinter2005,Hern_ndez_Cuenca_2024} which are only satisfied by particular classes of states, reflecting the specialized features of those states.

One particularly intriguing constraint, satisfied only by a limited subset of quantum states, is the monogamy of mutual information (MMI) 
inequality~\cite{Hayden:2013}
\begin{equation}
    \label{eq:MMI-Mutual-Information-Form}
    I(A:BC) \geq I(A:B) + I(A:C),
\end{equation}
which imposes strong restrictions on the entanglement patterns that states can exhibit. A prominent example of states that obey MMI is the class of holographic states, the set of quantum states which are dual to classical holographic systems via the AdS/CFT correspondence~\cite{Maldacena:1997re}. In a system of entangled qubits, MMI dictates how information can be distributed among the qubits, while maintaining maximal bipartite entanglement between specific pairs. Consequently, understanding the conditions under which MMI is violated has emerged as an active area of contemporary research~\cite{Bao_2015, Bao:2020mqq, Keeler:2022ajf}, with implications ranging from insights into spacetime reconstruction~\cite{Hayden_2013} in quantum gravity to enhanced protocols for distributed computation~\cite{Hein_2004,HeinDurEisertBriegel2006} in quantum networks. 

Generic quantum states, and even states with relatively ``simple'' entanglement patterns, fail MMI. In fact, MMI failure occurs at small system sizes ($n=4$) even among stabilizer states~\cite{gottesman1998heisenbergrepresentationquantumcomputers}, the class of quantum states well known to be efficiently simulable~\cite{Aaronson_2004} on a classical computer. A canonical example is the $4$-qubit $GHZ$ state, a stabilizer state whose entanglement entropy is the same regardless of the Hilbert space partition.

The full set of entanglement structures achievable by stabilizer states can be realized on a strict subset~\cite{VanDenNestDehaeneDeMoor2004}, the set of graph states. In the graph state formalism~\cite{Hein_2004,BriegelRaussendorf2001,RaussendorfBriegel2001,HeinDurEisertBriegel2006,fattal2004entanglementstabilizerformalism,Lovitz2022}, each state is represented as a simple undirected graph, where vertices correspond to qubits and edges encode the action of $CZ$ gates. This representation provides a natural framework for analysis of MMI violations in stabilizer states, using the graph-theoretic encoding of entanglement in graph states.

In this paper, we translate the evaluation of MMI for stabilizer states into an analysis of connectivity in graph states. We identify the algebraic conditions on graph adjacency matrices that determine the entanglement structures for the corresponding states, and consequently determine whether each entanglement structure satisfies or fails MMI. We discover a family of graphs that are guaranteed to fail MMI, and characterize the physical features of the states they represent. Furthermore, we specify conditions that give rise to MMI violation, highlighting minimal structural features responsible for such failures and describing how multipartite correlations depend on the graph topology. By understanding when and why violations of MMI occur, we gain deeper insight into the entanglement structures for different state classes, as well as into the physical properties which determine their potential utility in quantum information science. 

The remainder of this paper is organized as follows. In Section~\ref{Review}, we review the stabilizer formalism and define the mapping between stabilizer entanglement entropies and graph state connectivity. In Section~\ref{MMIFailueTableaux} we demonstrate how MMI can be reformulated as a comparison between graph adjacency matrices, and define MMI violation as a forbidden-subgraph problem. Section~\ref{ConditionsForViolation} specifies precise conditions under which MMI is violated in stabilizer states, using our graph-theoretic formalism. In Section~\ref{Discussion} we discuss the impact of this work, and propose future follow-ups to this manuscript. 

A comprehensive set of appendices accompanies this work, providing an extensive catalog of stabilizer state entanglement structures. These appendices include graph representations for all stabilizer state entropy vectors up through $n=8$ qubits, detailed MMI evaluations for each entropy vector, and physical interpretations corresponding to each entanglement structure. We further provide the count and percentages of stabilizer states which satisfy, saturate, and fail MMI at each qubit number, offering a self-contained reference for future studies of entanglement in stabilizer and graph states. The appendices also include full derivations, proofs, and added details for the calculations in this work. In particular, Tables~\ref{tab:MMI-Satisfying-EightQ-Graphs}–\ref{tab:MMI-Failing-EightQ-Graphs-Missing-Nontrivial-Int} provide graph representatives for every unique entropy vector at $n=8$, along with hyperlinks to a website that gives direct access to the entropy vector of the corresponding state.

\section{Review}\label{Review}

\subsection{Entropy Vectors}

Given a factorizable Hilbert space $\Hil$, and $n$-party pure state $\ket{\psi} \in \Hil$, let $I$ denote an $\ell$-party subsystem of $\ket{\psi}$. The entanglement entropy $S_I$, between $I$ and its $(n-\ell)$-party complement $\bar{I}$ is computed as
\begin{equation}\label{SubsystemEntropyDefinition}
   S_{I} = -\Tr \left(\rho_{I} \log \left(\rho_{I}\right) \right),
\end{equation}
with $\rho_{I}$ the reduced density matrix of subsystem $I$. Since we will consider qubit systems in this paper, it is natural for information to be measured in bits, and entropies are therefore computed using $\log_2$.

For a fixed division of the Hilbert space into $n$ parties, there are $2^{n} - 1$ computable entanglement entropies of a state $\ket{\psi}$. The ordered list of these entropies forms the entropy vector \cite{Bao_2015} for $\ket{\psi}$, and contains a complete description of the state's bipartite entanglement structure. For example, a $3$-party state $\ket{\psi}$ has an entropy vector
\begin{equation}
    \Vec{S}\left( \ket{\psi} \right) = \left(S_A,\ S_B,\ S_C,\ S_{AB},\ S_{AC},\ S_{BC},\ S_{ABC} \right).
\end{equation}
For a pure state $\ket{\psi}$, the fact that $S_{ABC}=0$ and $S_I = S_{\bar{I}}$ gives a symmetry to the entropy vector $\Vec{S}\left( \ket{\psi} \right)$. This symmetry allows $\Vec{S}\left( \ket{\psi} \right)$ to be specified using only $2^{n-1} - 1$ components, for example
\begin{equation}
    \Vec{S}\left( \ket{\psi} \right) = \left(S_A,\ S_B,\ S_C\right).
\end{equation}
Since we are considering pure states in this paper, we often use this reduced entropy vector notation for compactness.

\subsection{Entropy Inequalities and Monogamy of Mutual Information}

A homogeneous entropy inequality is a linear constraint of the form
\begin{equation}
    \sum_{\varnothing \neq A \in [n]} b_A S_A \geq 0 \quad b_A \in \mathbb{Z},
\end{equation}
that restricts the entanglement structure for a class of quantum states. The closure of the set of all entropy vectors of the class of states obeying some inequalities defines a convex, polyhedral subspace of the entropy vector space, known as the entropy cone~\cite{Linden2013,Bao_2015,Munizzi:2023ihc}. Accordingly, all entropy vectors realized by states in a particular class must reside within the entropy cone for that class. One such inequality, obeyed by all quantum states dual to classical holographic systems~\cite{Hayden:2013}, is the monogamy of mutual information (MMI). For an $n$-party system with pairwise-disjoint subsystems $I,\ J,$ and $K$, MMI requires
\begin{equation}
    \label{eq:MMI}
    S_{IJ} + S_{IK} + S_{JK} \geq S_I + S_J + S_K + S_{IJK}.
\end{equation}
We will refer to a single MMI instance, specified by pairwise disjoint tripartite subsystems $I,\ J,\ K$, as $MMI_{IJK}$. When discussing MMI in general, which constitutes all inequalities of the form of Eq.\ \eqref{eq:MMI}, we will omit this notation and simply refer to ``MMI.'' In a purely quantum-information context, MMI defines the non-positivity of tripartite information and constrains the correlations that can be shared among three parties while retaining maximal entanglement between any two. The MMI inequalities in Eq.\ \eqref{eq:MMI}, ranging over all possible choices of disjoint subsystems $I,\, J,\,K$, constitute a collection of inequalities obeyed by the holographic entropy cone~\cite{Bao_2015}.

\subsection{Stabilizer States and the Tableau Formalism}

The Pauli group on $n$ qubits, denoted $\mathcal{P}_n$, is the set of all $n$-fold tensor products of Pauli operators (including global phases), explicitly
\begin{equation}
    \label{eq:n-qubit Pauli Group}
    \mathcal{P}_n = \big\{  i^{j} \sigma_{1}\otimes\sigma_{2}\otimes \cdots \otimes \sigma_{n} \mid j \in \{ 0,1,2,3 \}, \; \sigma_{k} \in \{ \mathbb{1}, \sigma_x, \sigma_y, \sigma_z \} \big\},
\end{equation}
where each Pauli operator admits a $2 \times 2$ matrix  representation as
\begin{equation}
\label{eq:Pauli Matrices}
   \mathbb{1} = \begin{bmatrix}
        1 & 0 \\
        0 & 1
    \end{bmatrix}, \quad
   \sigma_x = \begin{bmatrix}
        0 & 1 \\
        1 & 0
    \end{bmatrix}, \quad 
    \sigma_y = \begin{bmatrix}
        0 & -i \\
        i & 0
    \end{bmatrix}, \quad
    \sigma_z = \begin{bmatrix}
        1 & 0 \\
        0 & -1
    \end{bmatrix}.
\end{equation}

Given $\mathcal{P}_n$, the $n$-qubit stabilizer states~\cite{Aaronson_2004} are then defined as the set of states invariant under a $2^n$-element subgroup of $\mathcal{P}_n$. Alternatively, the $n$-qubit stabilizer states can be defined as the orbit~\cite{Keeler:2023xcx} of $\ket{0}^{\otimes n}$ under the $n$-qubit Clifford group $\mathcal{C}_n$. The Clifford group is the normalizer of $\mathcal{P}_n$ in $U(2^{n})$,
\begin{equation}
    \mathcal{C}_n \equiv \{ U \in U(2^{n}) \mid U\mathcal{P}_nU^{\dagger}\mathcal{P}_n\}.
\end{equation}
The group $\mathcal{C}_n$ is generated by the Hadamard $H$, phase $P$, and Controlled-NOT%
\footnote{We use a convention for $CNOT_{i,j}$ where $i$ denotes the control qubit and $j$ the target qubit.} %
$CNOT$ quantum gates, each representable by the following matrices
\begin{equation}
    \label{eq:MatrixRepresentationsCliffordGen}
    H = \frac{1}{\sqrt{2}} \begin{bmatrix}
    1 & 1 \\
    1 & -1
    \end{bmatrix},
    \quad
    P = \begin{bmatrix}
    1 & 0 \\
    0 & i
    \end{bmatrix},
    \quad 
    CNOT_{1,2} = \begin{bmatrix}
    1 & 0 & 0 & 0 \\
    0 & 1 & 0 & 0 \\
    0 & 0 & 0 & 1 \\
    0 & 0 & 1 & 0
    \end{bmatrix}.
\end{equation}

Representing a generic qubit state requires an exponential number of classical bits to account for the $2^n$ Hilbert space dimension. Famously, however, each $n$-qubit stabilizer state can be described using $\mathcal{O}\left(n^2\right)$ classical bits, and Clifford circuits acting on stabilizer states are therefore efficiently classically simulable~\cite{gottesman1998heisenbergrepresentationquantumcomputers,Aaronson_2004}. This substantial reduction in the overhead is exhibited by the tableau formalism, an efficient representation for stabilizer states and Clifford transformations. Let $\mathfrak{G} = \{ g_1, \cdots, g_n \}$ be the generating set for the stabilizer subgroup $S < \mathcal{P}_n$ of a state $\ket{\psi}$. We represent $\mathfrak{G}$, and thereby the state $\ket{\psi}$, as a pair of binary square matrices $X,\ Z\in\mathbb{Z}_2^{n\times n}$ with entries defined
\begin{equation}
X_{ij} =
\begin{cases}
1, & (g_{i})_{j}\in\{\sigma_x,\sigma_y\}\\
0, & \text{otherwise}
\end{cases}
\quad
Z_{ij} =
\begin{cases}
1, & (g_{i})_{j}\in\{\sigma_y,\sigma_z\}\\
0, & \text{otherwise}
\end{cases}
\end{equation}
where $(g_{i})_{j}$ indicates the $j^{th}$ operator in the Pauli string $g_i$ and $\mathbb{Z}_2=\{0,1\}$. Placing the matrices $X$ and $Z$ side-by-side gives the $n \times 2n$ tableau representation for $\ket{\psi}$:
\begin{equation}
    T_{\psi} = \begin{bmatrix}
        X \mid Z
    \end{bmatrix}.
\end{equation}
To uniquely specify a state, we also need to assign a sign $\pm1$ to each row in $T_{\psi}$. However, since this specification does not impact the entanglement entropy of the state, we will omit it.

The evolution of $\ket{\psi}$ under Clifford operators can likewise be efficiently simulated through simple transformations on $T_{\psi}$. For each $n$-qubit Clifford operator, the associated action on $T_{\psi}$ is described by a map
\begin{equation}
    T_\psi \longmapsto\ T^{\prime}_{\psi} = \begin{bmatrix}
        X^{\prime} \mid Z^{\prime}
    \end{bmatrix},
\end{equation}
implemented through bitwise XOR ($\oplus$) operations on the affected columns. For the Hadamard, phase, and $CNOT$ gates the respective update rules~\cite{gottesman1998heisenbergrepresentationquantumcomputers, Aaronson_2004} on each row $i\in n$ in $T_\psi$ are 
\begin{alignat}{2}
H_a: \quad 
  &X^{\prime}_{ia} = Z_{ia}, 
  \quad& Z^{\prime}_{ia} &= X_{ia};
  \label{eq:tabH}\\[4pt]
S_a: \quad & X^{\prime}_{ia} = X_{ia}, 
  & Z^{\prime}_{ia} &= Z_{ia}\!\oplus X_{ia};
  \label{eq:tabS}\\[4pt]
CNOT_{a,b}: \quad 
  &X^{\prime}_{ib} = X_{ib}\!\oplus X_{ia}, 
  \quad& Z^{\prime}_{ia} &= Z_{ia}\!\oplus Z_{ib}.
  \label{eq:tabCNOT}
\end{alignat}
All other tableau entries (except possibly those specifying the signs of rows) remain unchanged. Eqs.\ \eqref{eq:tabH}--\eqref{eq:tabCNOT} are applied sequentially each time the corresponding gate acts on $\ket{\psi}$.

The entanglement entropies of a state $\ket{\psi}$ can likewise be computed from the tableau $T_{\psi}$.
For an $n$-party state $\ket{\psi}$, let $A \subseteq \{n\}$ denote the $A$ subsystem and $P_A$ the projection operator that deletes every column of $X$ and $Z$ in $T_G$ indexed by $\overline A$. The entanglement entropy of $A$ with $\overline{A}$ is then calculated~\cite{fattal2004entanglementstabilizerformalism} as 
\begin{equation}\label{eq:EETableau}
    S_A(\ket{\psi}) = R_A -|A|,
\end{equation}
where $R_A \equiv \rAz{P_A(T_G)}$. Eq.\ \eqref{eq:EETableau} therefore offers an efficient technique for evaluating stabilizer entanglement entropies using the tableau formalism. 

\subsection{Graph States}

A graph state is a quantum state representable by a simple, undirected graph, where vertices indicate qubits and edges indicate $CZ$ action. Since the $CZ$ gate acting on the state $\ket{+}$ generates maximal entanglement, the edges of a graph state additionally encode the entanglement structure of the state. Given a graph $G = (V, E)$ with $n$ vertices, the corresponding graph state $\ket{G}$ is obtained by acting on the all-plus state $\ket{+}^{\otimes n}$ with a controlled-$Z$ gate for each edge in $G$, explicitly
\begin{equation}\label{GraphStateDefinition}
    \ket{G} = \prod_{(i,j) \in E}CZ_{i,j} \ket{+}^{\otimes n}.
\end{equation}
In terms of Clifford generators, $CZ_{i,j} = H_jCNOT_{i,j}H_j,$. The $CZ$ gate is symmetric, making the distinction between control and target qubits purely conventional.

A local Clifford ($LC$) operation is a Clifford operator that acts as a non-identity only on a single qubit. On a graph state $\ket{G}$, $LC$ operations can be realized~\cite{VanDenNestDehaeneDeMoor2004} as a local complementation on the associated graph $G$. Accordingly, we use $LC$ to refer interchangeably to local Clifford equivalence and local complementation equivalence in graph states. For a graph $G$ and vertex $a$, the local complementation $\tau_a$ acts on the open 
neighborhood of $a$ (i.e.\ the neighborhood of $a$ with $a$ itself excluded) in $G$, denoted $N(a)$. The transformation $\tau_a$ replaces the subgraph of $G$, induced\footnote{Recall that the subgraph of a graph $G=(V,E)$ induced by the vertex set $V^\prime\subseteq V$ is the graph with vertices $V$ and edges connecting two vertices in $V^\prime$ if and only if those vertices are connected by an edge in $G$.} by $N(a)$, with its graph complement. The graphical effect of local complementation is illustrated by Figure \ref{fig:local-complementation}. 
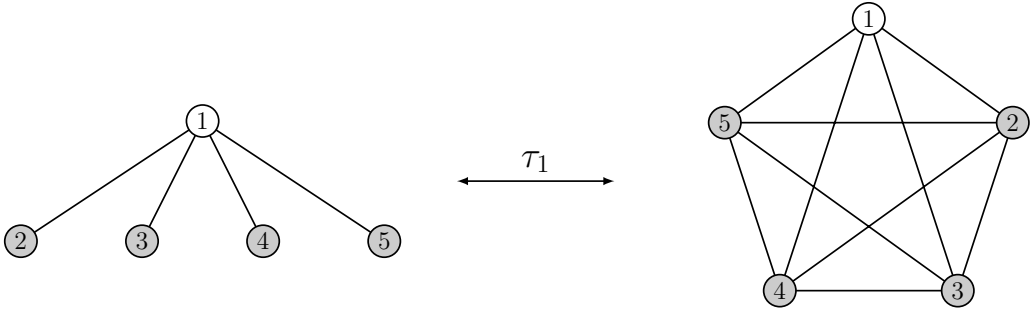
\begin{figure}[H]
  \centering
  \resizebox{0.9\linewidth }{!}{%
    \begin{tikzpicture}[>=stealth, thick, node/.style={circle, draw, fill=white, inner sep=2pt}]
      \path (0,0) coordinate (L1) (-3,-2) coordinate (L2) (-1,-2) coordinate (L3) (1,-2) coordinate (L4) (3,-2) coordinate (L5);
      \draw (L1)--(L2) (L1)--(L3) (L1)--(L4) (L1)--(L5);
      \node[node] at (L1) {1};
      \node[node, fill=black!20] at (L2) {2};
      \node[node, fill=black!20] at (L3) {3};
      \node[node, fill=black!20] at (L4) {4};
      \node[node, fill=black!20] at (L5) {5};
      \draw[<->, line width=0.8pt, >={latex[scale=2.0]}, draw=black]
  (4.2,-1) -- (6.8,-1) node[midway, above, font=\Large, text=black] {$\tau_1$};
      \begin{scope}[shift={(11,-0.8)}]
        \def\R{2.5}
        \path (90:\R) coordinate (R1) (18:\R) coordinate (R2) (-54:\R) coordinate (R3) (-126:\R) coordinate (R4) (162:\R) coordinate (R5);
        \draw (R1)--(R2)--(R3)--(R4)--(R5)--cycle;
        \draw (R1)--(R3)--(R5)--(R2)--(R4)--cycle;
        \node[node] at (R1) {1};
        \node[node, fill=black!20] at (R2) {2};
        \node[node, fill=black!20] at (R3) {3};
        \node[node, fill=black!20] at (R4) {4};
        \node[node, fill=black!20] at (R5) {5};
      \end{scope}
    \end{tikzpicture}
  }
  \caption{Left: star graph on five vertices $K_5$. Right: the result of locally complementing $K_5$ on vertex $1$, a complete graph on five vertices. The  neighborhood of vertex $1$, $N(1)$, is shaded in both cases. Notice that locally complementing vertex $1$ of the complete graph recovers $K_5$.}
  \label{fig:local-complementation}
\end{figure}
The effect of $\tau_a$ on the adjacency matrix for $G$ is
\begin{equation}
\label{eq:Local-Complementation}
\tau_a \left(\Gamma_G \right)=\Gamma_G+\Gamma_{N(a)} \mod{2},
\end{equation}
where $\Gamma_{N(a)}$ is the adjacency block of $\Gamma_G$ corresponding to $N(a)$. 

Each stabilizer state can be mapped to a graph state under $LC$ operations~\cite{VanDenNestDehaeneDeMoor2004}. Moreover, since $LC$ operations are single-qubit unitaries, they preserve all binary entanglement entropies and thus the entropy vector. For this reason stabilizer and graph states share the same set of entropy vectors, despite graph states comprising a proper subset of the stabilizer states, thereby motivating our exploration of the stabilizer entropy cone using graph states.

Entanglement entropies in graph states can be computed directly from the adjacency matrix~\cite{Hein_2004,burchardt2025foliagepartitioneasytocomputelcinvariant}. For a graph $G$ with vertex set $V = \{n\}$ and adjacency matrix $\Gamma_G$, let $\{A, \overline{A}\}$ be any bipartition of $V$. The entanglement entropy of $A$, in $\ket{G}$, is calculated from $\Gamma_G$ as
\begin{equation}\label{eq:EEadj}
         S_A(\ket{G}) = \rAz{\Gamma_{G}^{A \, \overline{A}}}.
\end{equation}
The matrix $\Gamma_{G}^{A , \overline{A}}$ in Eq.\ \eqref{eq:EEadj} indicates the local information submatrix of $A$, extracted from $\Gamma_G$ by only considering edges connecting a vertex in $A$ with a vertex in $\overline{A}$. Explicitly, we define $\Gamma_{G}^{A, \overline{A}}$ as
\begin{equation}\label{eq:connectionsMatrix}
       \Gamma_{G}^{A \, \overline{A}} = \left[ (\Gamma_G)_{ij} \right]_{i \in A, \, j \in \overline{A}}.
\end{equation}
While the local information submatrix $\Gamma_{G}^{A , \overline{A}}$ contains the information needed to compute the entanglement entropy for any subsystem of $A' \subseteq A$, $S_{A'}$ is independent of the matrix square blocks corresponding to the self adjacencies of $A'$ and $\overline{A'}$. 

Going forward we omit the subscript on $\Gamma_{G}^{A , \overline{A}}$, as context will unambiguously identify the associated graph. We likewise omit the complement subsystem in the superscript for conciseness. 

\section{MMI Failure in Tableaux and Graphs}\label{MMIFailueTableaux}

In this section, we show that the monogamy of mutual information (MMI) inequalities can be reformulated as a comparison between the ranks of tableau submatrices. We further demonstrate how the entangling action of a $CNOT$ gate alters the row and column spaces of these submatrices, thereby changing their ranks and leading to MMI violation. Extending this construction, we move to a graph-theoretic formulation in which MMI violation exists as a forbidden-subgraph problem. Using graph state adjacency submatrices, which encode the entanglement connectivity between distinct subsystems, we qualify the conditions required for MMI violation.

\subsection{Tableaux and MMI}\label{TableauxAndMMI}

The monogamy of mutual information (MMI) inequalities, given in Eq.\ \eqref{eq:MMI}, are typically expressed in terms of subsystem entanglement entropies $S$ \eqref{SubsystemEntropyDefinition}. However, when a state $\ket{\psi}$ admits a tableau representation, MMI can be reformulated directly in terms of the tableau data. For a stabilizer state $\ket{\psi}$ with associated tableau $T_\psi$, we can use Eq.\ \eqref{eq:EETableau} to rewrite MMI using the ranks of the corresponding tableau projections as
\begin{equation}\label{eq:MMI-Tableau-Ranks}
    R_{IJ} + R_{IK} + R_{JK} \geq R_I + R_J + R_K + R_{IJK}.
\end{equation}
The expression of MMI in Eq.\ \eqref{eq:MMI-Tableau-Ranks} offers a useful algebraic format to better analyze stabilizer entanglement structure. In particular, rewriting MMI in terms of submatrix ranks allows us to analyze how a sequence of quantum gates, via transformations on the tableau, induces a violation of MMI by directly altering the linear dependencies of tableau submatrices.

As a preliminary example, we use Eq.\ \eqref{eq:MMI-Tableau-Ranks} to analyze how a single Clifford gate induces the MMI failure exhibited by the four-qubit $GHZ$, the lowest qubit number stabilizer state to violate an MMI \eqref{eq:MMI} instance. The four-qubit $GHZ$ can be prepared using Clifford generators as
\begin{equation}
    \ket{GHZ}_4 = CNOT_{3,4}CNOT_{3,2}CNOT_{3,1}H_3 \ket{0}^{\otimes 4} = \frac{1}{\sqrt{2}}\left(\ket{0000}+\ket{1111}\right),
\end{equation}
and has a stabilizer group
\begin{equation}
    \stab{\ket{GHZ}_4} = \langle \text{XXXX}, \text{ZZII}, \text{IZZI}, \text{IIZZ} \rangle.
\end{equation}
To perform our analysis, we consider the state 
\begin{equation}
  \ket{\phi} = CNOT_{3,4} \ket{GHZ}_4 = \frac{1}{\sqrt{2}}\left(\ket{0000}+\ket{1110}\right),
\end{equation}
with stabilizer group 
\begin{equation}
    \stab{\ket{\phi}} = \langle \text{XXXI}, \text{ZZII}, \text{IZZI}, \text{IIIZ} \rangle.
\end{equation} 
Acting with a single $CNOT_{3,4}$ gate transforms $\ket{\phi}$ back into $\ket{GHZ}_4$. The tableaux $T_{\ket{GHZ}_4}$ and $T_{\ket{\phi}}$ are presented in Table~\ref{tab:MMI_eval_summary}. 

The state $\ket{\phi}$ has the entropy vector
\begin{equation}
    \vec{S}\left(\ket{\phi}\right) = \left(1,1,1,1,1,1,0 \right),
\end{equation}
which saturates every $4$-qubit instance of MMI \eqref{eq:MMI}. In contrast, $\ket{GHZ}_4$ has the entropy vector
\begin{equation}
    \vec{S}\left(\ket{GHZ}_4\right) = \left(1,1,1,1,1,1,1 \right),
\end{equation}
which violates multiple instances of MMI \eqref{eq:MMI}, as shown in Table \ref{tab:MMI_eval_summary}.
\begin{table}[H]
    \centering
    \small
    \begin{tabular}{l*{5}{c}}
        \toprule
        \textbf{State} & \textbf{(1,2,3)} & \textbf{(1,2,4)} & \textbf{(1,3,4)} & \textbf{(2,3,4)} & (1,2,34) \\
        \midrule
        $\ket{\phi}$ (Before)  & ST & ST & ST & ST & ST \\
        $\ket{GHZ}_4$ (After) & \textbf{F} & \textbf{F} & \textbf{F} & \textbf{F} & ST \\
        \midrule[\heavyrulewidth] 
        \textbf{State} & (1,3,24) & (1,4,23) & (2,3,14) & (2,4,13) & (3,4,12) \\
    \midrule
        $\ket{\phi}$ (Before)  & ST & ST & ST & ST & ST \\
        $\ket{GHZ}_4$ (After) & ST & ST & ST & ST & ST \\
        \bottomrule
    \end{tabular}
        \caption{Starting from state $\ket{\phi}$, which saturates all $4$-qubit MMI instances, a single $CNOT_{3,4}$ transforms $\ket{\phi}$ into $\ket{GHZ}_4$, which fails MMI. Columns label the MMI instance, and show whether each state satisfies (S), saturates (ST), or fails (F) that instance.}
    \label{tab:MMI_eval_summary}
\end{table}

We inspect the rank vector $\Vec{R}$, the ordered set of tableau submatrix ranks, to understand how the application of $CNOT_{3,4}$ induces an MMI violation. Specifically, we focus our analysis on the $3$-party MMI instance that involves subsystems $\{1,2,3\}$, $MMI_{\{1,2,3\}}$. Table~\ref{tab:rank_comparison} shows each component of $\Vec{R}_{\ket{\phi}}$ and $\Vec{R}_{\ket{GHZ}_4}$ that appears in $MMI_{\{1,2,3\}}$, and illustrates how the action of $CNOT_{3,4}$ selectively increases component $R_{123}$ from $3$ to $4$, leading to violation.
\begin{table}[H]
    \centering
    \small
    \begin{tabular}{lccccccc}
        \toprule
        \textbf{State} & $R_1$ & $R_2$ & $R_3$ & $R_{12}$ & $R_{13}$ & $R_{23}$ & $R_{123}$ \\
        \midrule
        $\ket{\phi}$ (Before)  & 2 & 2 & 2 & 3 & 3 & 3 & 3 \\
        $\ket{GHZ}_4$ (After) & 2 & 2 & 2 & 3 & 3 & 3 & \textbf{4} \\
        \bottomrule
    \end{tabular}
        \caption{Rank vectors $\Vec{R}$ for $\ket{\phi}$ and $\ket{GHZ}_4$ (before and after $CNOT_{3,4}$ is applied). The action of $CNOT_{3,4}$ increments the $R_{123}$ component of $\Vec{R}_{\ket{\phi}}$ by $1$, causing a violation of $MMI_{\{1,2,3\}}$, the MMI on the $3$-party instance involving subsystems $\{1,\ 2,\ 3\}$.}
    \label{tab:rank_comparison}
\end{table}

The rank increase shown in Table~\ref{tab:rank_comparison} is generated by a change to the column space of the $\ket{\phi}$ tableau. Table~\ref{tab:projected_tableaux} gives the tableaux for $\ket{\phi}$ and $\ket{GHZ}_4$, each projected onto the subsystem $\{123\}$ and its complement $\{4\}$. The objects $T_{123}$ and $T_4$ represent the tableau submatrices for subsystems $\{123\}$ and $\{4\}$ respectively.
\begin{table}[H]
    \centering
    \begin{tabular}{@{}ccc@{}}
        \toprule
        \textbf{State} & \textbf{$T_{123}$} & \textbf{$T_{4}$} \\
        \midrule
        \vspace{0.4em}
        $\ket{\phi}$ (Before) &
        $\left[
        \begin{array}{ccc@{\hspace{4pt}}|@{\hspace{4pt}}ccc}
         1 & 1 & 1 & 0 & 0 & 0 \\
         0 & 0 & 0 & 1 & 1 & 0 \\
         0 & 0 & 0 & 0 & 1 & 1 \\
         0 & 0 & 0 & 0 & 0 & 0 \\
        \end{array}
        \right]$ &
        $\left[
        \begin{array}{c@{\hspace{4pt}}|@{\hspace{4pt}}c}
         0 & 0 \\ 0 & 0 \\ 0 & 0 \\ 0 & 1 \\
        \end{array}
        \right]$ \\
        \addlinespace[6pt]
        $\ket{GHZ}_4$ (After) &
        $\left[
        \begin{array}{ccc@{\hspace{4pt}}|@{\hspace{4pt}}ccc}
         1 & 1 & 1 & 0 & 0 & 0 \\
         0 & 0 & 0 & 1 & 1 & 0 \\
         0 & 0 & 0 & 0 & 1 & 1 \\
         0 & 0 & 0& 0 & 0 & 1 \\
        \end{array}
        \right]$ &
        $\left[
        \begin{array}{c@{\hspace{4pt}}|@{\hspace{4pt}}c}
         1 & 0 \\ 0 & 0 \\ 0 & 0 \\ 0 & 1 \\
        \end{array}
        \right]$ \\ 
        \bottomrule
    \end{tabular}
        \caption{Tableaux projections $T_{123}$ and $T_{4}$ before and after $CNOT_{3,4}$ is applied to $\ket{\phi}$. The action of $CNOT_{3,4}$ increases the rank of $T_{123}$, which yields a violation of $MMI_{\{ 1,2,3 \}}$, the MMI instance involving the tripartite subsystem $\{ 1,2,3 \}$. The rank of $T_{4}$ also increases consistently with the equality of entanglement entropies for complementary subsystems, given by $T_{A} - T_{\bar{A}} = |A| - |\bar{A}|$.}
    \label{tab:projected_tableaux}
\end{table}
Table~\ref{tab:projected_tableaux} reveals that $CNOT_{3,4}$ introduces a fourth nonzero, linearly independent row in the projected tableau $T_{123}$, thereby increasing its rank from $3$ to $4$. Notably, this same $CNOT_{3,4}$ action on $\ket{\phi}$ preserves the $T_{123}$ column substructure that determines the first six entries of the rank vectors in Table~\ref{tab:rank_comparison}, with $R_{123}$ the only component affected by this gate.

Before the $CNOT_{3,4}$ gate is applied, all $X$ columns in $T_{123}$ are linearly dependent, while all $Z$ columns are pairwise linearly independent. Consequently, every length-one subsystem of $\{123\}$ is described by a tableau containing exactly one $X$ column and one $Z$ column, therefore guaranteeing a rank of $2$. Similarly, each length-two subsystem contains two $X$ and two $Z$ columns, giving a rank of $3$. Preserving the $X$ and $Z$ column dependencies, while increasing the rank of $T_{123}$, generates the $MMI_{\{ 1,2,3 \}}$ equality breaking that leads to violation. This analysis demonstrates that MMI violation in stabilizer states can be understood from the transformations of row/column spaces of tableau projections. In the next section we show how representing this construction in a graph-state formalism enables a compact and explicit qualification for MMI violation in stabilizer states.

\subsection{Entanglement Entropies From the Adjacency Matrices}

We now translate the tableau rank calculation from the previous section into a graph-theoretic formalism. Recall that graph states, while constituting a proper subset of stabilizer states, realize all possible stabilizer entropy vectors.
For an undirected simple graph $G$, which realizes the graph state $\ket{G}$, let $\Gamma_G$ indicate the adjacency matrix for $G$.  For any subsystem $\mathcal{L} \subseteq \ket{G}$, Eq.\ \eqref{eq:connectionsMatrix} defines the local information submatrix $\Gamma^{\mathcal{L}}$, constructed from $\Gamma$. The submatrix $\Gamma^\mathcal{L}$ contains all information needed to compute the entanglement entropy of $\mathcal{L}$. Moreover, $\Gamma^\mathcal{L}$ characterizes the connectivity between sets of vertices in the associated graph $G$, and simplifies the requisite MMI calculations as we will demonstrate in Sections \ref{section:single-vertex-center} and \ref{section:generalizedStarGraphs}.

We illustrate the calculation of entanglement entropy from a graph's adjacency matrix using the graph state $\ket{K_4}$, representable by the $4$-vertex star graph $K_4$ shown in Figure \ref{graph:GHZ4} and constructable as
\begin{equation}
    \ket{K_4} = CZ_{1,4}CZ_{1,3}CZ_{1,2}\ket{+}^{\otimes4}.
\end{equation}
\begin{figure}[H]
  \centering
  \begin{tikzpicture}[>=stealth, thick, 
    node/.style={circle, draw, fill=white, inner sep=2pt}]
    
    \node[node] (2) at (-2,-2) {2};
    \node[node] (3) at (0,-2) {3};
    \node[node] (4) at (2,-2) {4};
    \node[node] (1) at (0,0) {1};      

    \draw (1) -- (2);
    \draw (1) -- (3);
    \draw (1) -- (4);
    
  \end{tikzpicture}
  \caption{Graph Representation for $\ket{K_4}$, a state $LC$ equivalent to $\ket{GHZ}_4$, which violates all $3$-party instances of MMI.}
  \label{graph:GHZ4}
\end{figure}
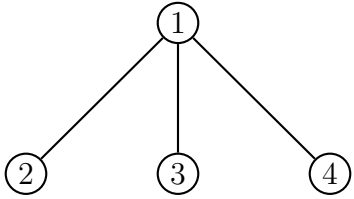
The graph state $\ket{K_4}$ is $LC$ equivalent to the $4$-qubit GHZ state $\ket{GHZ}_4$ through
\begin{equation}
    \ket{GHZ} = H_4 H_3 H_2 \ket{K_4}.
\end{equation}
Following Figure\ \ref{graph:GHZ4}, the $K_4$ adjacency matrix only contains non-zero elements in the first row and column, and is written
\begin{equation}
\Gamma =
    \begin{bmatrix}
    0 & 1 & 1 & 1 \\
    1 & 0 & 0 & 0 \\
    1 & 0 & 0 & 0 \\
    1 & 0 & 0 & 0 \\
    \end{bmatrix}.
\end{equation}

We evaluate $MMI_{\{ 1,3,4 \}}$, which corresponds to the MMI instance in Eq.\ \eqref{eq:MMI} on the disjoint subsystems $C = \{1\}, I=\{3\},$ and $ J=\{4\}$. The local information submatrix of $\mathcal{L} = \{ 1,3,4\}$ is
\begin{equation}
\Gamma^{\mathcal{L}} =
    \begin{bmatrix}
    0 & 1 & 1 & 1 \\
    1 & 0 & 0 & 0 \\
    1 & 0 & 0 & 0 \\
    \end{bmatrix}.
\end{equation}
For a subsystem $A$, involved in an MMI calculation, the local information submatrix $\Gamma^A$ is easily obtained from $\Gamma^{\mathcal{L}}$ by selecting the rows indexed by elements in $A$ and the columns indexed by elements in its complement $\overline{A}$. Figure \ref{fig:LI-submatrices-GHZ} illustrates this selection process. Following Eq.\ \eqref{eq:EEadj}, the entanglement entropy of $A$ is computed as the rank of $\Gamma^A$, over $\mathbb{Z}_2$.
\begin{center}
  \captionsetup{type=figure}
  \setlength{\tabcolsep}{0pt}
  \renewcommand{\arraystretch}{1}

  \begin{longtable}{@{}
    >{\centering\arraybackslash}m{0.3\textwidth}@{}
    >{\centering\arraybackslash}m{0.3\textwidth}@{}}
  \endfirsthead\endhead\endfoot\endlastfoot

  \begin{minipage}[c]{\linewidth}\centering
    \[
      \boldsymbol{\Gamma^{\{1\}}}\vspace{-2em}
    \]
    \begin{tikzpicture}[baseline=(current bounding box.north)]
      \matrix (m)[matrix of nodes,
                  ampersand replacement=\&, left delimiter={[}, right delimiter={]},
                  inner sep=3pt, row sep=0pt, column sep=0pt]{
        0\&1\&1\&1\\
        1\&0\&0\&0\\
        1\&0\&0\&0\\
      };
      \begin{scope}[on background layer]
        \foreach \c in {2,3,4}{\colborder{\c}}
        \rowborder{1}
        \foreach \c in {2,3,4}{\shadecell{1}{\c}}
        \outlineblock{1}{1}{2}{4}
      \end{scope}
    \end{tikzpicture}
  \end{minipage}
  &
  \begin{minipage}[c]{\linewidth}\centering
    \[
      \boldsymbol{\Gamma^{\{1,3\}}}\vspace{-2em}
    \]
    \begin{tikzpicture}[baseline=(current bounding box.north)]
      \matrix (m)[matrix of nodes,
                  ampersand replacement=\&, left delimiter={[}, right delimiter={]},
                  inner sep=3pt, row sep=0pt, column sep=0pt]{
        0\&1\&1\&1\\
        1\&0\&0\&0\\
        1\&0\&0\&0\\
      };
      \begin{scope}[on background layer]
        \foreach \c in {2,4}{\colborder{\c}}
        \rowborder{1}\rowborder{2}
        \foreach \r/\c in {1/2,1/4,2/2,2/4}{\shadecell{\r}{\c}}
        \outlineblock{1}{2}{2}{2}
        \outlineblock{1}{2}{4}{4}
      \end{scope}
    \end{tikzpicture}
  \end{minipage}\\

  \begin{minipage}[c]{\linewidth}\centering
    \[
      \boldsymbol{\Gamma^{\{3\}}}\vspace{-2em}
    \]
    \begin{tikzpicture}[baseline=(current bounding box.north)]
      \matrix (m)[matrix of nodes,
                  ampersand replacement=\&, left delimiter={[}, right delimiter={]},
                  inner sep=3pt, row sep=0pt, column sep=0pt]{
        0\&1\&1\&1\\
        1\&0\&0\&0\\
        1\&0\&0\&0\\
      };
      \begin{scope}[on background layer]
        \foreach \c in {1,2,4}{\colborder{\c}}
        \rowborder{2}
        \foreach \c in {1,2,4}{\shadecell{2}{\c}}
        \outlineblock{2}{2}{1}{2}
        \outlineblock{2}{2}{4}{4}
      \end{scope}
    \end{tikzpicture}
  \end{minipage}
  &
  \begin{minipage}[c]{\linewidth}\centering
    \[
      \boldsymbol{\Gamma^{\{1,4\}}}\vspace{-2em}
    \]
    \begin{tikzpicture}[baseline=(current bounding box.north)]
      \matrix (m)[matrix of nodes,
                  ampersand replacement=\&, left delimiter={[}, right delimiter={]},
                  inner sep=3pt, row sep=0pt, column sep=0pt]{
        0\&1\&1\&1\\
        1\&0\&0\&0\\
        1\&0\&0\&0\\
      };
      \begin{scope}[on background layer]
        \foreach \c in {2,3}{\colborder{\c}}
        \rowborder{1}\rowborder{3}
        \foreach \r/\c in {1/2,1/3,3/2,3/3}{\shadecell{\r}{\c}}
        \outlineblock{1}{1}{2}{3}
        \outlineblock{3}{3}{2}{3}
      \end{scope}
    \end{tikzpicture}
  \end{minipage}\\

  \begin{minipage}[c]{\linewidth}\centering
    \[
      \boldsymbol{\Gamma^{\{4\}}}\vspace{-2em}
    \]
    \begin{tikzpicture}[baseline=(current bounding box.north)]
      \matrix (m)[matrix of nodes,
                  ampersand replacement=\&, left delimiter={[}, right delimiter={]},
                  inner sep=3pt, row sep=0pt, column sep=0pt]{
        0\&1\&1\&1\\
        1\&0\&0\&0\\
        1\&0\&0\&0\\
      };
      \begin{scope}[on background layer]
        \foreach \c in {1,2,3}{\colborder{\c}}
        \rowborder{3}
        \foreach \c in {1,2,3}{\shadecell{3}{\c}}
        \outlineblock{3}{3}{1}{3}
      \end{scope}
    \end{tikzpicture}
  \end{minipage}
  &
  \begin{minipage}[c]{\linewidth}\centering
    \[
      \boldsymbol{\Gamma^{\{3,4\}}}\vspace{-2em}
    \]
    \begin{tikzpicture}[baseline=(current bounding box.north)]
      \matrix (m)[matrix of nodes,
                  ampersand replacement=\&, left delimiter={[}, right delimiter={]},
                  inner sep=3pt, row sep=0pt, column sep=0pt]{
        0\&1\&1\&1\\
        1\&0\&0\&0\\
        1\&0\&0\&0\\
      };
      \begin{scope}[on background layer]
        \foreach \c in {1,2}{\colborder{\c}}
        \rowborder{2}\rowborder{3}
        \foreach \r/\c in {2/1,2/2,3/1,3/2}{\shadecell{\r}{\c}}
        \outlineblock{2}{3}{1}{2}
      \end{scope}
    \end{tikzpicture}
  \end{minipage}\\

  \multicolumn{2}{>{\centering\arraybackslash}m{0.6\textwidth}}{%
    \begin{minipage}[c]{\linewidth}\centering
      \[
        \boldsymbol{\Gamma^{\{1,3,4\}}}\vspace{-2em}
      \]
      \begin{tikzpicture}[baseline=(current bounding box.north)]
        \matrix (m)[matrix of nodes,
                    ampersand replacement=\&, left delimiter={[}, right delimiter={]},
                    inner sep=3pt, row sep=0pt, column sep=0pt]{
          0\&1\&1\&1\\
          1\&0\&0\&0\\
          1\&0\&0\&0\\
        };
        \begin{scope}[on background layer]
          \colborder{2}
          \rowborder{1}\rowborder{2}\rowborder{3}
          \foreach \r in {1,2,3}{\shadecell{\r}{2}}
          \outlineblock{1}{3}{2}{2}
        \end{scope}
      \end{tikzpicture}
    \end{minipage}}\\

  \end{longtable}

\caption{For each distinct bipartition $\{A,\overline{A}\}$ of $\mathcal{L} = \{1,3,4\}$, rows indexed by vertices in $A$ and columns indexed by vertices in $\overline{A}$ are outlined; their intersection defines $\Gamma^{A}$, which is shaded. $S_A$ is independent of intra-subsystem adjacencies within the partition $\{A, \bar{A}\}$, that is, of the adjacencies within $A$ itself and within $\bar{A}$ itself. Observe that each $\Gamma^{A}$ shown has rank $1$.}
  \label{fig:LI-submatrices-GHZ}
\end{center}
Each matrix in Figure \ref{fig:LI-submatrices-GHZ} has rank $1$, and thus the entanglement entropy of each subsystem is likewise $1$. As demonstrated in Section \ref{TableauxAndMMI}, a $4$-qubit entropy vector with each component $1$ will violate all $3$-party instances of $MMI_{\{1,3,4 \}}$.

Connecting a fifth vertex to any leg of the $4$-star graph representing $\ket{K_4}$ produces a new graph state $\ket{K^{1,1,2}_{4}}$, which is not $LC$ equivalent to a $GHZ$ state. Nonetheless, the state $\ket{K^{1,1,2}_{4}}$ still violates MMI due to the $GHZ$-type entanglement present on four of its vertices. For example, let $K^{1,1,2}_{4}$ be the graph obtained by connecting a fifth vertex to vertex $4$ of the $K_4$ graph in Figure \ref{graph:GHZ4}. The state $\ket{K^{1,1,2}_{4}}$ can be constructed as $CZ_{4,5}\left( \ket{K_4}\otimes\ket{+}\right)$, and its associated graph $K^{1,1,2}_{4}$ is given in Figure \ref{graph:Psi}.
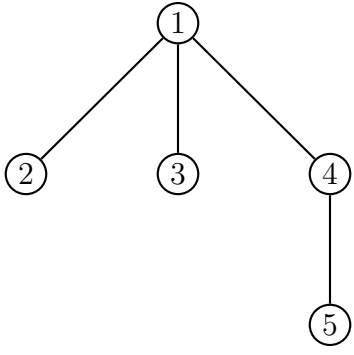
\begin{figure}[H]
  \centering
  \begin{tikzpicture}[>=stealth, thick, 
    node/.style={circle, draw, fill=white, inner sep=2pt}]
    
    \node[node] (2) at (-2,-2) {2};
    \node[node] (3) at (0,-2) {3};
    \node[node] (4) at (2,-2) {4};
    \node[node] (5) at (2,-4) {5};   
    \node[node] (1) at (0,0) {1};      

    \draw (1) -- (2);
    \draw (1) -- (3);
    \draw (1) -- (4);
    \draw (4) -- (5);
    
  \end{tikzpicture}
  \caption{Graph Representation for $\ket{K^{1,1,2}_{4}} = CZ_{4,5}\left( \ket{K_4}\otimes\ket{+}\right)$, with a four-star $K_4$ induced by vertices $\{1,2,3,4\}$. While $\ket{K^{1,1,2}_{4}}$ is not a $GHZ$ state, it nevertheless fails $MMI_{\{ 1,3,45 \}}$ due to its star topology.}
  \label{graph:Psi}
\end{figure}

Following from Figure \ref{graph:Psi}, the adjacency matrix for $K^{1,1,2}_{4}$ is given by
\begin{equation}
\Gamma =
    \begin{bmatrix}
    0 & 1 & 1 & 1 & 0 \\
    1 & 0 & 0 & 0 & 0 \\
    1 & 0 & 0 & 0 & 0 \\
    1 & 0 & 0 & 0 & 1 \\
    0 & 0 & 0 & 1 & 0 \\
    \end{bmatrix}.
\end{equation}
As before, we evaluate $MMI_{\{ 1,3,45 \}}$ using the local information submatrix for $ \\\mathcal{L} = \{1,3,4,5\}$, written
\begin{equation}
\label{eq:GHZ-adj-matrix}
\Gamma^{\mathcal{L}} =
    \begin{bmatrix}
    0 & 1 & 1 & 1 & 0 \\
    1 & 0 & 0 & 0 & 0 \\
    1 & 0 & 0 & 0 & 1 \\
    0 & 0 & 0 & 1 & 0 \\
    \end{bmatrix}.
\end{equation}

Figure \ref{fig:local-info-block-foprm} shows a breakdown of $\Gamma^{\mathcal{L}}$ into blocks, with the respective shared information elements highlighted.
\begin{figure}[H]
    \centering
    \includegraphics[width=0.4\linewidth]{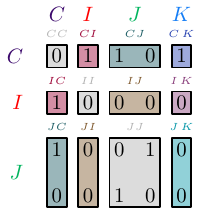}.
    \caption{Local information submatrix $\Gamma^{\mathcal{L}}$, for $\mathcal{L}=\{ 1,3,4,5\}$ in $\ket{K^{1,1,2}_{4}}$, given in block matrix form. The state $\ket{K^{1,1,2}_{4}}$ is partitioned such that $C\sqcup I\sqcup J\sqcup K = \{ 1,2,3,4,5\}$, and shared information elements are colored. Notice that $S_A$ for $A\subseteq \mathcal{L}$ is independent of the self adjacency blocks $CC$, $II$, and $JJ$, which are shaded gray.}
    \label{fig:local-info-block-foprm}
\end{figure}

Figure~\ref{fig:LI-submatrices-PSI} gives every local information submatrix needed to evaluate $MMI_{\{ 1,3,45 \}}$:
\begin{equation}
\label{eq:MMIViolationExample2}
\begin{aligned}
  \ra{\Gamma^{\{1\}}} + \ra{\Gamma^{\{3\}}} + \ra{\Gamma^{\{4,5\}}} + \ra{\Gamma^{\{1,3,4,5\}}} \\
  \overset{?}{\leq} \; \ra{\Gamma^{\{1,3\}}} + \ra{\Gamma^{\{1,4,5\}}}+ \ra{\Gamma^{\{3,4,5\}}}  ,
\end{aligned}
\end{equation}
where all ranks are taken over $\mathbb{Z}_2$, as the entanglement from adjacency matrix formula Eq.\ \eqref{eq:EEadj} requires. Inspection of Figure~\ref{fig:LI-submatrices-PSI} reveals that each term in Eq.\ \eqref{eq:MMIViolationExample2} has rank $1$, giving an entanglement entropy of $1$, thereby yielding a violation of $MMI_{\{ 1,3,45 \}}$ for $\ket{K^{1,1,2}_{4}}$.
\begin{center}
  \captionsetup{type=figure}
  \setlength{\tabcolsep}{0pt}
  \renewcommand{\arraystretch}{1}

  \begin{longtable}{@{}
    >{\centering\arraybackslash}m{0.3\textwidth}@{}
    >{\centering\arraybackslash}m{0.3\textwidth}@{}}
  \endfirsthead\endhead\endfoot\endlastfoot

  \begin{minipage}[c]{\linewidth}\centering
    \[
      \boldsymbol{\Gamma^{\{1\}}}\vspace{-2em}
    \]
    \begin{tikzpicture}[baseline=(current bounding box.north)]
      \matrix (m)[matrix of nodes,
                  ampersand replacement=\&, left delimiter={[}, right delimiter={]},
                  inner sep=3pt, row sep=1.01pt, column sep=0pt]{
        0\&1\&1\&1\&0\\
        1\&0\&0\&0\&0\\
        1\&0\&0\&0\&1\\
        0\&0\&0\&1\&0\\
      };
      \begin{scope}[on background layer]
        \foreach \c in {2,3,4,5}{\colborderSecond{\c}}
        \rowborderSecond{1}
        \foreach \c in {2,3,4,5}{\shadecell{1}{\c}}
        \outlineblockSecond{1}{1}{2}{5}
      \end{scope}
    \end{tikzpicture}
  \end{minipage}
  &
   \begin{minipage}[c]{\linewidth}\centering
    \[
      \boldsymbol{\Gamma^{\{1,3\}}}\vspace{-2em}
    \]
    \begin{tikzpicture}[baseline=(current bounding box.north)]
      \matrix (m)[matrix of nodes,
                  ampersand replacement=\&, left delimiter={[}, right delimiter={]},
                  inner sep=3pt, row sep=1.01pt, column sep=0pt]{
        0\&1\&1\&1\&0\\
        1\&0\&0\&0\&0\\
        1\&0\&0\&0\&1\\
        0\&0\&0\&1\&0\\
      };
      \begin{scope}[on background layer]
        \foreach \c in {2,4,5}{\colborderSecond{\c}}
        \rowborderSecond{1}\rowborderSecond{2}
        \foreach \r/\c in {1/2,2/2}{\shadecell{\r}{\c}}
        \foreach \r/\c in {1/4,1/5, 2/4, 2/5}{\shadecell{\r}{\c}}
        \outlineblockSecond{1}{2}{2}{2}
        \outlineblockSecond{1}{2}{4}{5}
        
      \end{scope}
    \end{tikzpicture}
  \end{minipage}\\

  \begin{minipage}[c]{\linewidth}\centering
    \[
      \boldsymbol{\Gamma^{\{3\}}}\vspace{-2em}
    \]
    \begin{tikzpicture}[baseline=(current bounding box.north)]
      \matrix (m)[matrix of nodes,
                  ampersand replacement=\&, left delimiter={[}, right delimiter={]},
                  inner sep=3pt, row sep=1.01pt, column sep=0pt]{
        0\&1\&1\&1\&0\\
        1\&0\&0\&0\&0\\
        1\&0\&0\&0\&1\\
        0\&0\&0\&1\&0\\
      };
      \begin{scope}[on background layer]
        \foreach \c in {1,2,4,5}{\colborderSecond{\c}}
        \rowborderSecond{2}
        \foreach \c in {1,2,4,5}{\shadecell{2}{\c}}
        \outlineblockSecond{2}{2}{1}{2}
        \outlineblockSecond{2}{2}{4}{5}
      \end{scope}
    \end{tikzpicture}
  \end{minipage}
  &
  \begin{minipage}[c]{\linewidth}\centering
    \[
      \boldsymbol{\Gamma^{\{1,4,5\}}}\vspace{-2em}
    \]
    \begin{tikzpicture}[baseline=(current bounding box.north)]
      \matrix (m)[matrix of nodes,
                  ampersand replacement=\&, left delimiter={[}, right delimiter={]},
                  inner sep=3pt, row sep=1.01pt, column sep=0pt]{
        0\&1\&1\&1\&0\\
        1\&0\&0\&0\&0\\
        1\&0\&0\&0\&1\\
        0\&0\&0\&1\&0\\
      };
      \begin{scope}[on background layer]
        \foreach \c in {2,3}{\colborderSecond{\c}}
        \rowborderSecond{1}\rowborderSecond{3}\rowborderSecond{4}
        \foreach \r/\c in {1/2,1/3,3/2,3/3,4/2,4/3}{\shadecell{\r}{\c}}
        \outlineblockSecond{1}{1}{2}{3}
        \outlineblockSecond{3}{4}{2}{3}
      \end{scope}
    \end{tikzpicture}
  \end{minipage}\\

  \begin{minipage}[c]{\linewidth}\centering
    \[
      \boldsymbol{\Gamma^{\{4,5\}}}\vspace{-2em}
    \]
    \begin{tikzpicture}[baseline=(current bounding box.north)]
      \matrix (m)[matrix of nodes,
                  ampersand replacement=\&, left delimiter={[}, right delimiter={]},
                  inner sep=3pt, row sep=1.01pt, column sep=0pt]{
        0\&1\&1\&1\&0\\
        1\&0\&0\&0\&0\\
        1\&0\&0\&0\&1\\
        0\&0\&0\&1\&0\\
      };
      \begin{scope}[on background layer]
        \foreach \c in {1,2,3}{\colborderSecond{\c}}
        \rowborderSecond{3}\rowborderSecond{4}
        \foreach \r/\c in {3/1,3/2,3/3,4/1,4/2,4/3}{\shadecell{\r}{\c}}
        \outlineblockSecond{3}{4}{1}{3}
      \end{scope}
    \end{tikzpicture}
  \end{minipage}
  &
  \begin{minipage}[c]{\linewidth}\centering
    \[
      \boldsymbol{\Gamma^{\{3,4,5\}}}\vspace{-2em}
    \]
    \begin{tikzpicture}[baseline=(current bounding box.north)]
      \matrix (m)[matrix of nodes,
                  ampersand replacement=\&, left delimiter={[}, right delimiter={]},
                  inner sep=3pt, row sep=1.01pt, column sep=0pt]{
        0\&1\&1\&1\&0\\
        1\&0\&0\&0\&0\\
        1\&0\&0\&0\&1\\
        0\&0\&0\&1\&0\\
      };
      \begin{scope}[on background layer]
        \foreach \c in {1,2}{\colborderSecond{\c}}
        \rowborderSecond{2}\rowborderSecond{3}\rowborderSecond{4}
        \foreach \r/\c in {2/1,2/2,3/1,3/2,4/1,4/2}{\shadecell{\r}{\c}}
        \outlineblockSecond{2}{4}{1}{2}
      \end{scope}
    \end{tikzpicture}
  \end{minipage}\\

  \multicolumn{2}{>{\centering\arraybackslash}m{0.6\textwidth}}{%
    \begin{minipage}[c]{\linewidth}\centering
      \[
        \boldsymbol{\Gamma^{\{1,3,4,5\}}}\vspace{-2em}
      \]
      \begin{tikzpicture}[baseline=(current bounding box.north)]
        \matrix (m)[matrix of nodes,
                    ampersand replacement=\&, left delimiter={[}, right delimiter={]},
                    inner sep=3pt, row sep=1.01pt, column sep=0pt]{
          0\&1\&1\&1\&0\\
          1\&0\&0\&0\&0\\
          1\&0\&0\&0\&1\\
          0\&0\&0\&1\&0\\
        };
        \begin{scope}[on background layer]
          \colborderSecond{2}
          \rowborderSecond{1}\rowborderSecond{2}\rowborderSecond{3}\rowborderSecond{4}
          \foreach \r in {1,2,3,4}{\shadecell{\r}{2}}
          \outlineblockSecond{1}{4}{2}{2}
        \end{scope}
      \end{tikzpicture}
    \end{minipage}}\\

  \end{longtable}
  \caption{For each bipartition $\{A,\overline{A}\}$ of $\mathcal{L} = \{1,3,4,5\}$, rows indexed by $A$ and columns indexed by $\overline{A}$ are outlined. The intersection of these rows and columns defines $\Gamma^{A}$, which is shaded. $S_A$ is independent of intra-subsystem adjacencies within the partition $\{A, \bar{A}\}$, that is, of the adjacencies within $A$ itself and within $\bar{A}$ itself. For state $\ket{K^{1,1,2}_{4}}$, every $\Gamma^{A}$ has rank $1$, and therefore $\ket{K^{1,1,2}_{4}}$ fails $MMI_{\{ 1,3,45 \}}$.}
  \label{fig:LI-submatrices-PSI}
\end{center}

In this section we showed how the MMI inequalities can be reformulated in terms of tableau submatrix ranks, and demonstrated how entangling operations, e.g. a $CNOT$ gate, alter submatrix rank to induce MMI violations. Tracking the row and column space transformations of tableau projections, we identified the mechanism that leads to MMI violation in graph states that are $LC$ equivalent to $4$-qubit GHZ states, and other graph states with an equivalent entanglement structure. We extended this construction to a graph-theoretic setting, where the condition for MMI violation manifests as a forbidden-subgraph structure that embodies $GHZ$-type entanglement. Evaluating MMI using the ranks of graph state adjacency matrices, we explicitly showed how MMI violation is compactly diagnosed by the local information submatrices of subsystems involved in the inequalities. In the next section we lay out necessary conditions for a family of graph states to violate particular MMI instances, and discuss the physical interpretations for these requirements. 

\section{Conditions for MMI Violation in Graph States}\label{ConditionsForViolation}

In this section, we characterize the structural conditions for violation of specific MMI families in graph states. We analyze the relationship between graph connectivity and the entanglement entropy between disjoint graph state subsystems, identifying the explicit conditions that guarantee failure of these specific MMI instances. We first consider a specific class of star-like graphs $G$, featuring a single central vertex $c$ connected to disjoint and disconnected subgraphs, and derive the necessary and sufficient conditions under which such graphs violate a family of MMI instances $MMI_{cIJ}$. We then extend this framework to consider generalized star graphs $\mathcal{G}$, possessing a multi-vertex center $C$, and algebraically relate failure of the analogous MMI family $MMI_{CIJ}$ to specific intersection and distributivity properties of adjacency matrix column spaces. We highlight the physical implications of the violation of these specific MMI instances in graph states, and present several observations of significance. 

\subsection{An MMI-Violating Graph}\label{section:single-vertex-center}

We now derive the necessary conditions for a family of graph states to violate the MMI instances $MMI_{cIJ}$, defined in Eq.\ \eqref{eq:MMI-Fail-G_m}. 

Let $G = (V, E)$ be a graph on $n$ vertices that consists of a central vertex $c$ connected to $k \geq 3$ nonempty subgraphs,
\begin{equation}\label{GraphProducts}
    G^{1},\, G^{2},\, \dots,\, G^{k},
\end{equation}
which are vertex-disjoint and mutually%
\footnote{Multiple edges may exist between the central vertex $c$ and each disjoint subgraph $G^p$.}
disconnected, and where each $G^{p} = (V^{p}, E^{p})$ for $p = 1,2,\dots,k$. For any pair of indices $p,\ q \in \{1,2,\dots,k\}$, such that $p \neq q$, the graph $G$ satisfies the following conditions:
\begin{enumerate}[label=(\arabic*)]
    \item \label{condition: disjointness} If $v \in V^{p}$, then $v \notin V^{q}$ (subgraph disjointness),
    \item \label{condition: partition} The set $\{ \{c\},V^1, V^{2}, \cdots, V^k \}$ defines a partition of $V$ (graph partitioning).
    \item \label{condition: Disconnection} For every $v \in V^{p}$ and $w \in V^{q}$, we have $(v, w) \notin E$ (inter-subgraph disconnectedness).
    \item \label{condition: Anchoring} There exists at least one $v \in V^{p}$ such that $(v, c) \in E$ (subgraph anchoring),
\end{enumerate}
When considering a particular instance of MMI we will fix a partition of the vertices $V\setminus c$ into three subsets $I,\:J$ and $K:=V\setminus \{c,I,J\}$, where $I,J,K$ are each unions of the vertex sets $V^{p}$. A schematic diagram for such a graph $G$ is presented in Figure~\ref{fig:G_m-Graph}.
\begin{figure}[H]
  \centering
  \begin{tikzpicture}[
      every node/.style={cloud, cloud puffs=10, draw, align=center},
      >=stealth
  ]

    \node[
      circle,
      draw,
      fill=white,
      minimum size=1cm,
      cloud puffs=0, 
      font=\large
    ] (center) at (0,0) {$c$};

    \node[fill=gray!20, aspect=1] (V1) at (-6,-3) {$G^1$};
    \draw (center) -- (V1);

    \node[fill=gray!20, aspect=1] (V2) at (-3,-3) {$G^2$};
    \draw (center) -- (V2);

    \node[fill=gray!20, aspect=1] (V3) at (0,-3)  {$G^3$};
    \draw (center) -- (V3);

    \node[draw=none, fill=none] (dots) at (3,-3) {\large $\cdots$};

    \node[fill=gray!20, aspect=1] (Sl) at (6,-3) {$G^k$};
    \draw (center) -- (Sl);

    \coordinate (subbot) at (V1.south);
    \begin{scope}[shift={(5,0)}]
            \node[draw, rectangle, rounded corners, minimum width=3cm, minimum height=1.3cm, fill=white] (legend) {};
            \node[circle, draw, fill=white, minimum size=0.4cm, cloud puffs=0] at (-0.9,0.3) {};
            \node[draw=none, fill=none, anchor=west, font=\small] at (-1,0.3) {Vertex};
    
            \node[cloud, cloud puffs=10, draw, fill=white, minimum size=0.5cm] at (-0.9,-0.3) {};
            \node[draw=none, fill=none, anchor=west, font=\small] at (-1.05,-0.3) {Subgraph};
    \end{scope}
    \draw[
      decorate,
      decoration={brace,mirror,amplitude=5pt}
    ]
      ([yshift=-2ex]V1.south west|-subbot) -- ([yshift=-2ex]V2.south east|-subbot)
      node[midway,yshift=-4ex,draw=none,fill=none] {$I$};
    
    \draw[
      decorate,
      decoration={brace,mirror,amplitude=5pt}
    ]
      ([yshift=-2ex]V3.south west|-subbot) -- ([yshift=-2ex]V3.south east|-subbot)
      node[midway,yshift=-4ex,draw=none,fill=none] {$J$};
    
    \draw[
      decorate,
      decoration={brace,mirror,amplitude=5pt}
    ]
      ([yshift=-2ex]dots.center|-subbot) -- ([yshift=-2ex]Sl.center|-subbot)
      node[midway,yshift=-4ex,draw=none,fill=none] {$K$};

  \end{tikzpicture}
  \caption{Generalized star graph $G$ which satisfies the conditions of subgraph disjointness, graph partitioning, subgraph anchoring, and inter-subgraph disconnectedness. The subsystems $\{c\},\ I,\ $ $J$, and $K$ give a valid partition that defines $MMI_{cIJ}$, as in Eq.\ \eqref{eq:MMI-Fail-G_m}, where any path between the subsystems $I$, $J$ and $K$ must include the central vertex $c$. Any graph state represented by a graph of type $G$ will have an entropy vector that violates $MMI_{cIJ}$.}
  \label{fig:G_m-Graph}
\end{figure}
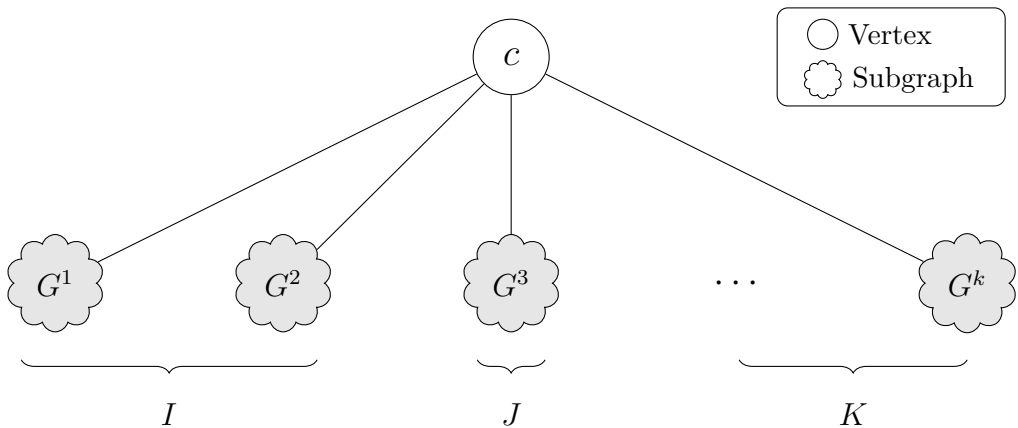

Having defined the class of graphs $G$, which satisfy the properties of subgraph disjointness, graph partitioning, subgraph anchoring, and inter-subgraph disconnectedness, we now prove that any graph state $\ket{G}$ violates%
\footnote{We emphasize that the condition $k\geq3$, in Eq.\ \eqref{GraphProducts}, is important since states $GHZ_n$, for $n \leq3$, will satisfy $MMI_{cIJ}$ as the adjacency submatrix ranks can produce entropies of $0$.} %
$MMI_{cIJ}$.
\begin{proposition}
\label{thm:MMI-violation}
Every graph state $\ket{G}$, represented by a graph $G$ which admits the properties of subgraph disjointness, graph partitioning, subgraph anchoring, and inter-subgraph disconnectedness, has an entropy vector that violates every $MMI_{cIJ}$, which is of the form
\begin{equation}\label{eq:MMI-Fail-G_m}
\begin{aligned}
    S_{cI} + S_{cJ} + S_{IJ} \geq S_{c} + S_{I} + S_{J} + S_{cIJ},
\end{aligned}
\end{equation}
where $c$ denotes the subsystem containing the central vertex, and $I, J \subset V \setminus c$ such that $I \cap J = \varnothing$, $ \overline{c\sqcup I\sqcup J} \neq \varnothing $, and any path between $I,J$ and $K$ must include $c$.
\end{proposition}

\begin{proof}
We prove the proposition by showing that each term in $MMI_{cIJ}$, as in Eq.\ \eqref{eq:MMI-Fail-G_m} has value $1$. This result follows from the fact that the adjacency submatrices, relevant to the $MMI_{cIJ}$ entries in Eq.\ \eqref{eq:MMI-Fail-G_m}, only have nonzero entries along a single row or column.

For a graph state $\ket{G}$, the selection of disjoint subsystems $c,\ I,\ J \subseteq \ket{G}$ induce the vertex partition $V = c \sqcup I \sqcup J \sqcup K$ on $G$. Let $\mathcal{L} = c \sqcup I \sqcup J$ which, after a vertex permutation, admits the local information submatrix $\Gamma^{\mathcal{L}}$ as
\begin{equation}\label{locInfSub}
\Gamma^{\mathcal{L}} =
\begin{bmatrix}
0 & c I & c J & c K \\[1mm]
\left(c I\right)^T & I   I & I   J & I   K \\[1mm]
\left(c J\right)^T & \left(I   J\right)^T & J   J & J   K
\end{bmatrix}.
\end{equation}
In this representation, each block of $\Gamma^{\mathcal{L}}$ contains adjacencies between the two subsystems which label that block, i.e. each entry of $\Gamma^{\mathcal{L}}$ is a block matrix $(\Gamma)_{u,v}$ with $u,v \in \{c,\ I,\ J,\ K\}$.

Second, the condition of inter-subgraph disconnectedness~\ref{condition: Disconnection} requires that the block structure of $\Gamma^{\mathcal{L}}$ obeys the constraints
\begin{equation}
\begin{split}
    IJ &= \textbf{0}_{IJ},\\
    IK &= \textbf{0}_{IK},\\
    JK &= \textbf{0}_{JK},\\
\end{split}
\end{equation}
where $\textbf{0}_{IJ}$ denotes the zero matrix of size $|I| \times |J|$. Omitting the dimensionality of the block-zero matrices, for conciseness and since each is fixed by its adjacent blocks, allows $\Gamma^{\mathcal{L}}$ in Eq.\ \eqref{locInfSub} to be rewritten as
\begin{equation}\label{eq:InfoSubmatrix-G_m}
\Gamma^{\mathcal{L}} =
\begin{bmatrix}
0 & c I & c J & c K \\[1mm]
\left(c I\right)^T & I   I & \boldsymbol{0} & \boldsymbol{0} \\[1mm]
\left(c J\right)^T & \boldsymbol{0} & J   J & \boldsymbol{0}
\end{bmatrix}.
\end{equation}

For each subset $A \subseteq \mathcal{L}$, the local information submatrix $\Gamma^{A}$ can be read directly from $\Gamma^{\mathcal{L}}$ using the definition of these matrices provided in Eq.\ \eqref{eq:connectionsMatrix}. The rank of each $\Gamma^{A}$ determines the entanglement entropy $S_A$, as given by Eq.\ \eqref{eq:EEadj}. Therefore all terms needed to evaluate $MMI_{cIJ}$ are computed as
\begin{equation}\label{eq:S_allGm}
\begin{array}{c}
\begin{tabular}{cc}
  \begin{tabular}{r@{\,$=\,$}l}
    $S_c$ &
    $\displaystyle
    \ra{\begin{bmatrix}
       c I & c J & c K
    \end{bmatrix}}$ \\[3ex]
    $S_I$ &
    $\ra{c I}$ \\[3ex]
    $S_J$ &
    $\ra{c J}$
  \end{tabular}
  &
  \begin{tabular}{r@{\,$=\,$}l}
    $S_{c I}$ &
    $\displaystyle
    \ra{\begin{bmatrix}
       c J & c K
    \end{bmatrix}}$ \\[3ex]
    $S_{c J}$ &
    $\displaystyle
    \ra{\begin{bmatrix}
       c I & c K
    \end{bmatrix}}$ \\[3ex]
    $S_{I J}$ &
    $\displaystyle
    \ra{\begin{bmatrix}
       c I & c J
    \end{bmatrix}}$
  \end{tabular}
\end{tabular}

\vspace{3ex}

\\
\begin{tabular}{r@{\,$=\,$}l}
  $S_{cIJ}$ &
  $\ra{c K}$
\end{tabular}
\end{array}
\end{equation}

Each entropy in Eq.\ \eqref{eq:S_allGm} corresponds to the rank of a single-row matrix, which can take value $0$ or $1$, with a value of $1$ occurring if and only if the matrix contains a non-zero entry. However, the subgraph anchoring condition~\ref{condition: Anchoring} guarantees that $cI$, $cJ$, and $cK$ each contain at least one nonzero entry. Consequently, the rank of each matrix in Eq.\ \eqref{eq:S_allGm} is $1$, and thus each entropy $S_A = 1$. Therefore any graph state $\ket{G}$ will have an entropy vector that fails $MMI_{cIJ}$.
\end{proof} 

In the above section, we derived the structural conditions under which a graph state belonging to the family $G$ will violate the $MMI_{cIJ}$ as in Eq.\ \eqref{eq:MMI-Fail-G_m}. We demonstrated that any graph $G$, composed of a single central vertex connected to mutually disjoint and disconnected subgraphs, always possesses an entropy vector that fails $MMI_{cIJ}$. For each subsystem of $\ket{G}$ involved in the $MMI_{cIJ}$ evaluation, we showed how the block structure of the adjacency submatrix yields an entanglement entropy of $1$, leading to a guaranteed violation of this MMI instance. Accordingly, we established that $MMI_{cIJ}$ violation inevitably occurs in graph states with a star topology, as in Figure~\ref{fig:G_m-Graph}, satisfying the conditions of subgraph disjointness~\ref{condition: disjointness}, graph partition~\ref{condition: partition}, subgraph disconnectedness~\ref{condition: Disconnection}, and subgraph anchoring~\ref{condition: Anchoring}. We now generalize our proof to include graphs where the central vertex $c$ is promoted to a multi-vertex subgraph.

\subsection{Characterization of MMI Family Violation in a Larger Family of Graph States}
\label{section:generalizedStarGraphs}

We now extend the class of graphs $G$, introduced in Section \ref{section:single-vertex-center}, by generalizing the single central vertex $c$ to a multi-vertex subgraph $C$, as depicted in Figure \ref{fig:calG_M-Graph}. In these generalized star graphs, denoted $\mathcal{G}$, we identify the failure conditions for the MMI instance $MMI_{CIJ}$. We show that subgraph anchoring \ref{condition: Anchoring}, is no longer sufficient to guarantee failure of $MMI_{CIJ}$, and detail additional requirements for violation.

We define $\mathcal{G} = (V, E)$ as a graph on $n$ vertices, with a central subgraph $C$ consisting of vertices $\{c_1, c_2, \dots, c_m \}$. The subgraph $C$ is connected to $k$ non-empty subgraphs, 
\begin{equation}
    \mathcal{G}^{1},\, \mathcal{G}^{2},\, \dots,\, \mathcal{G}^{k},
\end{equation}
that are both vertex-disjoint and mutually disconnected, and where each $\mathcal{G}^{p} = (V^{p},E^{p})$ for $p = 1,2,\dots,k$. As before, for any pair of indices $p,q \in \{1,2,\ldots,k\}$ such that $p\neq q$, the graph $\mathcal{G}$ satisfies the following conditions: 
\begin{enumerate}[label=(\arabic*)]
    \item \label{condition: genDisjointness} If $v \in V^{p}$, then $v \notin V^{q}$ (subgraph disjointness).
    \item \label{condition: genPartition} The set $\{ C ,V^1, V^{2}, \cdots, V^k\}$ defines a partition of $V$ (graph partitioning).
    \item \label{condition: genDisconnection} For every $v \in \mathcal{V}^{p}$ and $w \in V^{q}$, we have $(v, w) \notin E$ (inter-subgraph disconnectedness).
\end{enumerate}
These conditions do not capture the generalization of subgraph anchoring which requires a more detailed discussion as below (see Observation \ref{obs:proper-generalization-G}). A schematic representation for graphs of the form $\mathcal{G}$ is given in Figure \ref{fig:calG_M-Graph}.
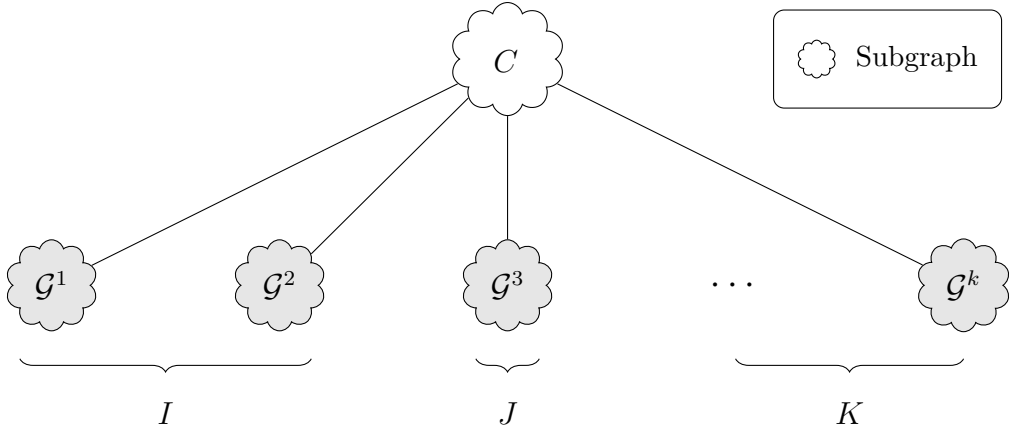
\begin{figure}[H]
  \centering
  \begin{tikzpicture}[
      every node/.style={cloud, cloud puffs=10, draw, align=center},
      >=stealth
  ]

    \node[fill=white, aspect=1] (center) at (0,0) {\hspace{0.3em}$C^{\phantom{1}}$};

    \node[fill=gray!20, aspect=1] (V1) at (-6,-3) {$\mathcal{G}^1$};
    \draw (center) -- (V1);

    \node[fill=gray!20, aspect=1] (V2) at (-3,-3) {$\mathcal{G}^2$};
    \draw (center) -- (V2);

    \node[fill=gray!20, aspect=1] (V3) at (0,-3)  {$\mathcal{G}^3$};
    \draw (center) -- (V3);

    \node[draw=none, fill=none] (dots) at (3,-3) {\large $\cdots$};

    \node[fill=gray!20, aspect=1] (Sl) at (6,-3) {$\mathcal{G}^k$};
    \draw (center) -- (Sl);

    \coordinate (subbot) at (V1.south);
    \begin{scope}[shift={(5,0)}]
        \node[draw, rectangle, rounded corners, minimum width=3cm, minimum height=1.3cm, fill=white] (legend) {};

        \node[draw=none, shape=rectangle, inner sep=0pt] at (legend.center) {%
            \tikz[baseline=-0.5ex, every node/.style={}]
              \node[cloud, cloud puffs=10, draw, fill=white, minimum size=0.5cm] {};
            \hspace{0.3em}\small Subgraph%
          };
    \end{scope}
    
    \draw[
      decorate,
      decoration={brace,mirror,amplitude=5pt}
    ]
      ([yshift=-2ex]V1.south west|-subbot) -- ([yshift=-2ex]V2.south east|-subbot)
      node[midway,yshift=-4ex,draw=none,fill=none] {$I$};
    
    \draw[
      decorate,
      decoration={brace,mirror,amplitude=5pt}
    ]
      ([yshift=-2ex]V3.south west|-subbot) -- ([yshift=-2ex]V3.south east|-subbot)
      node[midway,yshift=-4ex,draw=none,fill=none] {$J$};
    
    \draw[
      decorate,
      decoration={brace,mirror,amplitude=5pt}
    ]
      ([yshift=-2ex]dots.center|-subbot) -- ([yshift=-2ex]Sl.center|-subbot)
      node[midway,yshift=-4ex,draw=none,fill=none] {$K$};

  \end{tikzpicture}
  \caption{Graph $\mathcal{G}$, the generalization of $G$ in Figure \ref{fig:G_m-Graph}, with the central vertex $c$ promoted to the subgraph $C$. Every $\mathcal{G}$ satisfies the subgraph disjointness, graph partition, and inter-subgraph disconnectedness conditions, but not subgraph anchoring. Subsystems $C,\ I,\ $ and $J$ partition the vertices of $\mathcal{G}$, defining $MMI_{CIJ}$ in Eq.\ \eqref{eq:MMI-Fail-G_M}, and every path between $I$, $J$ and $K$ must only include vertices in $C$. Not every $\ket{\mathcal{G}}$ violates $MMI_{CIJ}$, as shown in Table \ref{tab:MMIcases}.}
  \label{fig:calG_M-Graph}
\end{figure}

Given graphs of the form $\mathcal{G}$, we now seek to determine under what conditions the state $\ket{\mathcal{G}}$ fails $MMI_{CIJ}$, 
\begin{equation}\label{eq:MMI-Fail-G_M}
\begin{aligned}
    S_{CI} + S_{CJ} + S_{IJ} \geq S_{C} + S_{I} + S_{J} + S_{CIJ}
\end{aligned},
\end{equation}
where $I$ and $J$ are defined as in Section \ref{section:single-vertex-center} with $I, J \subset V \setminus C \text{ such that } I \cap J = \varnothing \text{ and } \overline{C\sqcup I\sqcup J} \neq \varnothing$. As before, subgraph disjointness~\ref{condition: genPartition} and graph partitioning~\ref{condition: genPartition} guarantee that the choice of $C,\ I,$ and $J$ defines the partition $V = C \sqcup I \sqcup J \sqcup K$. We set $\mathcal{L} = C \sqcup I \sqcup J $, and express the local information submatrix $\Gamma^{\mathcal{L}}$ in block form
\begin{equation}\label{eq:InfoSubmatrix-G_M}
\Gamma^{\mathcal{L}} =
\begin{bmatrix}
C   C  & C  I & C  J & C  K \\[1mm]
\left(C I\right)^T & I   I & \boldsymbol{0} & \boldsymbol{0} \\[1mm]
\left(C J\right)^T & \boldsymbol{0} & J   J & \boldsymbol{0}
\end{bmatrix},
\end{equation}
immediately imposing inter-subgraph disconnectedness \ref{condition: genDisconnection}. The first row of $\Gamma^{\mathcal{L}}$ in Eq.\ \eqref{eq:InfoSubmatrix-G_M} is now a row of block matrices, instead of a row of numbers as was the case for $G$. 

We can reformulate $MMI_{CIJ}$ in terms of the dimensions of intersections and sums of block matrix column spaces, as detailed in Appendix \ref{A1:DimColEquations}. Applying this rewriting, $MMI_{CIJ}$ becomes
\begin{equation}\label{eq:MMI-Fail-ColSpaces-G_M}
\begin{split}
&\dimension{\col{C  I} \cap \col{C  K}} + \dimension{\col{C  J} \cap \col{C  K}}\\
& \qquad \qquad \qquad \qquad  \leq \dimension{\left( \col{C  I} + \col{C  J}\right) \cap \col{C  K}}.\\
\end{split}
\end{equation} 
In Eq.\ \eqref{eq:MMI-Fail-ColSpaces-G_M} $\col{M}$ denotes the column space of the matrix $M$, i.e. the linear subspace spanned by the columns of $M$. Thereby $\dimension{\col{CI} \cap \col{CK}}$ is the dimension of the subspace formed by the intersection of the column space of block matrix $CI$ and the column space of $CK$. Importantly, addition inside of the dimension evaluation in Eq.\ \eqref{eq:MMI-Fail-ColSpaces-G_M} corresponds to subspace addition, which is not simply the union of two subspaces, but instead the set of all possible sums of elements in both subspaces. 

For compactness, let $W_I = \col{C I}$, $W_J = \col{C J}$, and $W_K = \col{C K}$. Using Eqs.\ \eqref{subspaceDimension} and \eqref{subspaceInclusion} from Appendix \ref{A2:sufficient-Condition-General-Graph}, which detail the arithmetic for the dimension of sums of subspaces and the property of subspace inclusion, respectively, we derive the following inequality for column spaces:
\begin{equation}\label{eq:Gen-Intersection-Over-Sum-Inequality}
    \resizebox{0.97\linewidth}{!}{$
    \dimension{W_I \cap W_K}
    + \dimension{W_J \cap W_K}
    - \dimension{W_I \cap W_J \cap W_K}
    \leq
    \dimension{(W_I + W_J) \cap W_K}
    $}.
\end{equation}
Eq.\ \eqref{eq:Gen-Intersection-Over-Sum-Inequality} differs from the $MMI_{CIJ}$ reformulation in Eq.\ \eqref{eq:MMI-Fail-ColSpaces-G_M} only by the subtraction of term $\dimension{W_I\cap W_J \cap W_K}$. Another important property that enables our analysis of MMI is subspace distributivity. The subspaces $W_I$, $W_J$ and $W_K$ are said to be
distributive%
\footnote{If Eq.\ \eqref{eq:subspaceDistributivity} is true, then distributivity also holds for every permutation of $\{W_I,\ W_J,\ W_K\}$.} %
linear subspaces of $\mathbb{Z}_2^{|C|}$ if
\begin{equation}\label{eq:subspaceDistributivity}
    W_I \cap W_K + W_J\cap W_K = \left(W_I+W_J\right)\cap W_K.
\end{equation}

Using Eqs.\ \eqref{eq:Gen-Intersection-Over-Sum-Inequality} and \eqref{eq:subspaceDistributivity} we now classify distinct realizations of $\mathcal{G}$ according to whether $W_I,\ W_J$, and $W_K$ have a trivial intersection and are distributive. For the states $\ket{\mathcal{G}}$ belonging to each class, we deduce directly from Eq.\ \eqref{eq:Gen-Intersection-Over-Sum-Inequality} whether they strictly satisfy, saturate, or fail the specified $MMI_{CIJ}$ in Eq.\ \eqref{eq:MMI-Fail-G_M}. A proof for each case is provided in Appendix \ref{A2:sufficient-Condition-General-Graph}. Table \ref{tab:MMIcases} provides a summary of our classification.
\begin{table}[H]
\centering
\begin{tabular}{c c c c }
\toprule
\textbf{Case} & 
\textbf{Distributive} & 
\textbf{Nontrivial Intersection} & \textbf{$\boldsymbol{MMI_{CIJ}}$ Evaluation} \\
\midrule
1 & No &  No & Satisfies \\[6pt]
2 & Yes  & No & Saturates \\[6pt]
3 & Yes & Yes & Fails \\[6pt]
4 & No & Yes & Undetermined \\[6pt]
\bottomrule
\end{tabular}
\caption{Evaluation criteria for $MMI_{CIJ}$, in Eq.\ \eqref{eq:MMI-Fail-G_M}, for graph states $\ket{\mathcal{G}}$ based on the intersection $W_I\cap W_J \cap W_K$ and the distributivity of $W_I = \col{C I}$, $W_J = \col{C J}$, and $W_K = \col{C K}$ defined in Eq.\ \eqref{eq:subspaceDistributivity}.}
\label{tab:MMIcases}
\end{table}
Tables~\ref{tab:MMI-characterization-examples-1-3} and~\ref{tab:case-4-graph-examples} provide explicit $\mathcal{G}$ realizations of every case in Table~\ref{tab:MMIcases}.
\begin{table}[H]
  \centering
  \tikzset{
  stdnode/.style={circle, draw, fill=white, inner sep=1.6pt, minimum size=18pt},
  stdlines/.style={line width=0.9pt}
  }
  \begin{tabular}{c c}
  \toprule
  \textbf{Case} & $\boldsymbol{\mathcal{G}}$\textbf{ Representative} \\
  \midrule
  1 &
  \begin{minipage}{0.72\linewidth}\centering \vspace{10pt}
  \begin{tikzpicture}[stdlines]
  \node[stdnode] (i1) at (-3,0) {$i_1$};
  \node[stdnode] (c1) at (-1.5,0) {$c_1$};
  \node[stdnode] (k1) at ( 0,0) {$k_1$};
  \node[stdnode] (c2) at ( 1.5,0) {$c_2$};
  \node[stdnode] (j1) at ( 3,0) {$j_1$};
  \draw (i1)--(c1)--(k1)--(c2)--(j1);
  \end{tikzpicture}\\[2pt]
  \footnotesize Non-distributive; trivial intersection
  \end{minipage}
  \\[30pt]
  2 &
  \begin{minipage}{0.72\linewidth}\centering
  \begin{tikzpicture}[stdlines]
  \node[stdnode] (c1) at (-2, 1.6) {$c_1$};
  \node[stdnode] (c2) at ( 0, 1.6) {$c_2$};
  \node[stdnode] (c3) at ( 2, 1.6) {$c_3$};
  \node[stdnode] (i1) at (-2, 0) {$i_1$};
  \node[stdnode] (j1) at ( 0, 0) {$j_1$};
  \node[stdnode] (k1) at ( 2, 0) {$k_1$};
  \draw (c1)--(i1);
  \draw (c2)--(j1);
  \draw (c3)--(k1);
  \end{tikzpicture}\\
  \footnotesize Distributive; trivial intersection
  \end{minipage}\\[30pt]
  \\
  3 &
  \begin{minipage}{0.72\linewidth}\centering
  \begin{tikzpicture}[stdlines]
  \path[use as bounding box] (-3.2,-0.5) rectangle (3.2,2.3);
  \begin{scope}[xshift=1mm] 
  \node[stdnode] (c1) at (-1.2, 1.8) {$c_1$};
  \node[stdnode] (c2) at ( 0.0, 1.8) {$c_2$};
  \node[stdnode] (c3) at ( 1.2, 1.8) {$c_3$};
  \node[stdnode] (i2) at (-3.2, 0) {$i_1$};
  \node[stdnode] (i1) at (-2.0, 0) {$i_2$};
  \node[stdnode] (j1) at ( 0.0, 0) {$j_1$};
  \node[stdnode] (k1) at ( 2.0, 0) {$k_1$};
  \node[stdnode] (k2) at ( 3.2, 0) {$k_2$};
  \draw (c2)--(i1);
  \draw (c2)--(j1);
  \draw (c2)--(k1);
  \draw (c1)--(i2);
  \draw (c3)--(k2);
  \end{scope}
  \end{tikzpicture}\\[2pt]
  \footnotesize Distributive; non-trivial intersection
  \end{minipage} \vspace{10pt}
  \\
  \bottomrule
  \end{tabular}
  \caption{$\mathcal{G}$ instances illustrating Cases 1–3 of the $MMI_{CIJ}$ \eqref{eq:MMI-Fail-G_M} categorization for $\ket{\mathcal{G}}$ states, as listed in Table \ref{tab:MMIcases}. Each example depicts the subspace structure defined by the column spaces of $CI$, $CJ$, and $CK$, associated with the tripartite subsystem $\{C,I,J\}$. Within the partition $\{C,I,J, K\}$, every vertex in each graph belongs to the subsystem that labels it (i.e., $c_p \in C$, $i_q \in I$, and $j_m \in J$ for all $c_p, i_q, j_m \in V$). Notice that the subgraph induced by $C$ is not required to be connected.}
  \label{tab:MMI-characterization-examples-1-3}
\end{table}
\begin{table}[H]
\centering
\setlength{\tabcolsep}{6pt}
\renewcommand{\arraystretch}{1}
\tikzset{
  stdnode/.style={circle, draw, fill=white, inner sep=1.8pt, minimum size=32pt, font=\fontsize{22pt}{22pt}\selectfont},
  stdlines/.style={line width=1.5pt}
}
\resizebox{\textwidth}{!}{%
\begin{tabular}{@{}c@{\hspace{2em}}c@{}}  
\toprule
\textbf{$\boldsymbol{\mathcal{G}}$\textbf{ Representative}} & \textbf{Isomorphic Form} \\
\midrule

\multicolumn{2}{c}{\small\bfseries Saturates}\\[4pt]

\begin{minipage}{0.44\linewidth}\centering \vspace{10pt}
\resizebox{\linewidth}{!}{%
\tikzset{bignode/.style=stdnode}
\begin{tikzpicture}[stdlines,scale=1.7619,transform shape=false]
  \node[bignode] (c2) at (-1.2, 3.0) {$c_2$};
  \node[bignode] (c1) at ( 0.0, 3.0) {$c_1$};
  \node[bignode] (c3) at ( 1.2, 3.0) {$c_3$};
  \node[bignode] (i2) at (-4.2, 0.0) {$i_1$};
  \node[bignode] (i1) at (-3.0, 0.0) {$i_2$};
  \node[bignode] (j1) at (-0.6, 0.0) {$j_1$};
  \node[bignode] (j2) at ( 0.6, 0.0) {$j_2$};
  \node[bignode] (k1) at ( 3.0, 0.0) {$k_1$};
  \node[bignode] (k2) at ( 4.2, 0.0) {$k_2$};
  \draw (c1)--(i1);
  \draw (c1)--(i2);
  \draw (c1)--(j1);
  \draw (c1)--(k1);
  \draw (c2)--(k2);
  \draw (k2)--(c3);
  \draw (c3)--(j2);
\end{tikzpicture}}
\end{minipage}
&
\begin{minipage}{0.44\linewidth}\centering \vspace{10pt}
\resizebox{\linewidth}{!}{%
\begin{tikzpicture}[stdlines,scale=2.0556,transform shape=false]
  \node[stdnode] (c1) at ( 0.0,  0.8) {$c_1$};
  \node[stdnode] (i2) at (-1.4,  2) {$i_1$};
  \node[stdnode] (k1) at ( 1.4,  2) {$k_1$};
  \node[stdnode] (i1) at (-1.4, -0.5) {$i_2$};
  \node[stdnode] (j1) at ( 1.4, -0.5) {$j_1$};
  \draw (i2)--(c1)--(k1);
  \draw (i1)--(c1)--(j1);
  \node[stdnode] (j2) at (-3.6, -2) {$j_2$};
  \node[stdnode] (c3) at (-1.2, -2) {$c_3$};
  \node[stdnode] (k2) at ( 1.2, -2) {$k_2$};
  \node[stdnode] (c2) at ( 3.6, -2) {$c_2$};
  \draw (j2)--(c3)--(k2)--(c2);
\end{tikzpicture}}
\end{minipage}
\\[80pt]

\multicolumn{2}{c}{\small\bfseries Fails}\\[4pt]

\begin{minipage}{0.44\linewidth}\centering \vspace{20pt}
\resizebox{\linewidth}{!}{%
\tikzset{bignode/.style=stdnode}
\begin{tikzpicture}[stdlines,scale=1.4231,transform shape=false]
  \node[bignode] (c1) at (-2.4, 4) {$c_1$};
  \node[bignode] (c2) at (-0.8, 4) {$c_2$};
  \node[bignode] (c3) at ( 0.8, 4) {$c_3$};
  \node[bignode] (c4) at ( 2.4, 4) {$c_4$};
  \node[bignode] (i3) at (-5.2, 0.0) {$i_3$};
  \node[bignode] (i2) at (-4.0, 0.0) {$i_2$};
  \node[bignode] (i1) at (-2.8, 0.0) {$i_1$};
  \node[bignode] (j1) at (-1.2, 0.0) {$j_1$};
  \node[bignode] (j2) at ( 0.0, 0.0) {$j_2$};
  \node[bignode] (j3) at ( 1.2, 0.0) {$j_3$};
  \node[bignode] (k1) at ( 2.8, 0.0) {$k_1$};
  \node[bignode] (k2) at ( 4.0, 0.0) {$k_2$};
  \node[bignode] (k3) at ( 5.2, 0.0) {$k_3$};
  \draw (c1)--(i1);
  \draw (c3)--(i2);
  \draw (c4)--(i3);
  \draw (c2)--(j1);
  \draw (c3)--(j2);
  \draw (c4)--(j3);
  \draw (c1)--(k1);
  \draw (c2)--(k1);
  \draw (c3)--(k2);
  \draw (c4)--(k3);
\end{tikzpicture}}
\end{minipage}
&
\begin{minipage}{0.44\linewidth}\centering \vspace{20pt}
\resizebox{\linewidth}{!}{%
\begin{tikzpicture}[stdlines,scale=1.6087,transform shape=false]
  \node[stdnode] (c3) at (-3.0, 1.8) {$c_3$};
  \node[stdnode] (k2) at (-4.6, 3.0) {$k_2$};
  \node[stdnode] (i2) at (-1.4, 3.0) {$i_2$};
  \node[stdnode] (j2) at (-3.0, -0.2) {$j_2$};
  \draw (k2)--(c3)--(i2);
  \draw (c3)--(j2);
  \node[stdnode] (c4) at ( 3.0, 1.8) {$c_4$};
  \node[stdnode] (k3) at ( 1.4, 3.0) {$k_3$};
  \node[stdnode] (i3) at ( 4.6, 3.0) {$i_3$};
  \node[stdnode] (j3) at ( 3.0, -0.2) {$j_3$};
  \draw (k3)--(c4)--(i3);
  \draw (c4)--(j3);
  \node[stdnode] (i1) at (-4.0, -2) {$i_1$};
  \node[stdnode] (c1) at (-2.0, -2) {$c_1$};
  \node[stdnode] (k1) at ( 0.0, -2) {$k_1$};
  \node[stdnode] (c2) at ( 2.0, -2) {$c_2$};
  \node[stdnode] (j1) at ( 4.0, -2) {$j_1$};
  \draw (i1)--(c1)--(k1)--(c2)--(j1);
\end{tikzpicture}}
\end{minipage}
\\[80pt]

\multicolumn{2}{c}{\small\bfseries Satisfies}\\[4pt]

\begin{minipage}{0.44\linewidth}\centering \vspace{10pt}
\resizebox{\linewidth}{!}{%
\tikzset{bignode/.style=stdnode}
\begin{tikzpicture}[stdlines, baseline=(current bounding box.center)]
  \node[bignode] (c1) at (-3.0, 6.0) {$c_1$};
  \node[bignode] (c2) at (-1.8, 6.0) {$c_2$};
  \node[bignode] (c3) at (-0.6, 6.0) {$c_3$};
  \node[bignode] (c4) at ( 0.6, 6.0) {$c_4$};
  \node[bignode] (c5) at ( 1.8, 6.0) {$c_5$};
  \node[bignode] (c6) at ( 3.0, 6.0) {$c_6$};

  \node[bignode] (i1) at (-8.075, 0.0) {$i_1$};
  \node[bignode] (i2) at (-6.725, 0.0) {$i_2$};
  \node[bignode] (i3) at (-5.375, 0.0) {$i_3$};
  \node[bignode] (i4) at (-4.025, 0.0) {$i_4$};

  \node[bignode] (j1) at (-2.025, 0.0) {$j_1$};
  \node[bignode] (j2) at (-0.675, 0.0) {$j_2$};
  \node[bignode] (j3) at ( 0.675, 0.0) {$j_3$};
  \node[bignode] (j4) at ( 2.025, 0.0) {$j_4$};

  \node[bignode] (k1) at ( 4.025, 0.0) {$k_1$};
  \node[bignode] (k2) at ( 5.375, 0.0) {$k_2$};
  \node[bignode] (k3) at ( 6.725, 0.0) {$k_3$};

  \draw (c1)--(i1);  \draw (c2)--(i2);  \draw (c3)--(i3);  \draw (c5)--(i4);
  \draw (c1)--(j1);  \draw (c2)--(j2);  \draw (c4)--(j3);  \draw (c6)--(j4);
  \draw (c1)--(k1);  \draw (c2)--(k1);
  \draw (c3)--(k2);  \draw (c4)--(k2);
  \draw (c5)--(k3);  \draw (c6)--(k3);
\end{tikzpicture}}
\end{minipage}
&
\begin{minipage}{0.44\linewidth}\centering \vspace{10pt}
\resizebox{\linewidth}{!}{%
\begin{tikzpicture}[stdlines,scale=1.85,transform shape=false,baseline=(current bounding box.center)]
  \node[stdnode] (c1) at (-3.0, 1.4) {$c_1$};
  \node[stdnode] (k1) at ( 0.0, 1.4) {$k_1$};
  \node[stdnode] (c2) at ( 3.0, 1.4) {$c_2$};

  \node[stdnode] (i1) at (-4, 2.6) {$i_1$};
  \node[stdnode] (j1) at (-4, 0.2) {$j_1$};
  \node[stdnode] (j2) at ( 4, 2.6) {$j_2$};
  \node[stdnode] (i2) at ( 4, 0.2) {$i_2$};

  \draw (i1)--(c1)--(k1)--(c2)--(j2);
  \draw (j1)--(c1);
  \draw (i2)--(c2);

  \node[stdnode] (i4) at (-4.0, -1.0) {$i_4$};
  \node[stdnode] (c5) at (-2.0, -1.0) {$c_5$};
  \node[stdnode] (k3) at ( 0.0, -1.0) {$k_3$};
  \node[stdnode] (c6) at ( 2.0, -1.0) {$c_6$};
  \node[stdnode] (j4) at ( 4.0, -1.0) {$j_4$};
  \draw (i4)--(c5)--(k3)--(c6)--(j4);

  \node[stdnode] (i3) at (-4.0, -2) {$i_3$};
  \node[stdnode] (c3) at (-2.0, -2) {$c_3$};
  \node[stdnode] (k2) at ( 0.0, -2) {$k_2$};
  \node[stdnode] (c4) at ( 2.0, -2) {$c_4$};
  \node[stdnode] (j3) at ( 4.0, -2) {$j_3$};
  \draw (i3)--(c3)--(k2)--(c4)--(j3);
\end{tikzpicture}}
\end{minipage}\vspace{10pt}
\\
\bottomrule
\end{tabular}}
\caption{$\mathcal{G}$ representatives illustrating every $MMI_{CIJ}$~\eqref{eq:MMI-Fail-G_M} outcome belonging to Case~4 (undetermined) of the categorization shown in Table~\ref{tab:MMIcases}. Each representative is accompanied by an isomorphic form that makes its disconnected structure explicit. Every example exhibits non-distributive and nontrivially intersecting column spaces of $CI$, $CJ$, and $CK$ (see Table~\ref{tab:Case4Examples} in Appendix~\ref{A2:sufficient-Condition-General-Graph} for additional details), all associated with the tripartite subsystem $\{C,I,J\}$. Within the partition $\{C,I,J,K\}$, every vertex in each graph belongs to the subsystem that labels it (i.e., $c_p \in C$, $i_q \in I$, and $j_m \in J$ for all $c_p, i_q, j_m \in V$). The outcome of $MMI_{CIJ}$~\eqref{eq:MMI-Fail-G_M} for $\ket{\mathcal{G}}$ labels each $\mathcal{G}$.}
\label{tab:case-4-graph-examples}
\end{table}
The results of Table~\ref{tab:MMIcases} can be summarized by the following observations, each referring to $MMI_{CIJ}$ in Eq.\ \eqref{eq:MMI-Fail-G_M}.
\begin{observation}\label{obs:cannot-strictly-satisfy-generalized-graph}
If $\col{CI}$, $\col{CJ}$, and $\col{CK}$ are distributive, then $\ket{\mathcal{G}}$ cannot strictly satisfy $MMI_{CIJ}$~\eqref{eq:MMI-Fail-G_M}.
\end{observation}

\begin{observation}\label{obs:necessary-condition-generalized-graph}
The condition $\col{CI} \cap \col{CJ} \cap \col{CK} \neq \varnothing$ is necessary for $\ket{\mathcal{G}}$ to fail $MMI_{CIJ}$~\eqref{eq:MMI-Fail-G_M}.
\end{observation}

\begin{observation}\label{obs:sufficient-generalized-star}
If $\col{CI}$, $\col{CJ}$, and $\col{CK}$ are distributive and have a nontrivial intersection, then $\ket{\mathcal{G}}$ fails $MMI_{CIJ}$~\eqref{eq:MMI-Fail-G_M}.
\end{observation}

For consistency, we further verify that any graph of the class $\mathcal{G}$ that violates $MMI_{CIJ}$~\eqref{eq:MMI-Fail-G_M} is of type $G$ when its central subgraph $C$ consists of a single vertex, as in Section~\ref{section:single-vertex-center}. When subgraph $C$ contains only a single vertex $c$, then $W_I$, $W_J$, and $W_K$ are subspaces of the one-dimensional vector space $\mathbb{Z}_2$. By Observation~\ref{obs:necessary-condition-generalized-graph}, these subspaces must have a nontrivial intersection for $\ket{\mathcal{G}}$ in order to fail $MMI_{CIJ}$~\eqref{eq:MMI-Fail-G_M}. Since the only subspaces of $\mathbb{Z}_2$ are 
\begin{equation}
    \left\{ \begin{bmatrix}
        0
    \end{bmatrix} \right\}  \text{ and } \left\{ \begin{bmatrix}
        0
    \end{bmatrix},  \begin{bmatrix}
        1
    \end{bmatrix} \right\},
\end{equation}
the condition for $W_I$, $W_J$, and $W_K$ to share a nontrivial intersection is satisfied only when
\begin{equation}
    W_1= W_2= W_3 =\left\{ \begin{bmatrix}
        0
    \end{bmatrix}, \begin{bmatrix}
        1
    \end{bmatrix} \right\},
\end{equation}
which is precisely the subgraph anchoring condition \ref{condition: Anchoring} from Section~\ref{section:single-vertex-center}. Therefore, any graph $\mathcal{G}$ that falls into Case $3$ of Table~\ref{tab:MMIcases} reduces to a graph of the form $G$ when $|C|=1$.

More generally, the evaluation of $MMI_{CIJ}$~\eqref{eq:MMI-Fail-G_M} for $\ket{\mathcal{G}}$ trivially reduces to Case $3$ in Table~\ref{tab:MMIcases} when $W_I = W_J = W_K$, failing $MMI_{CIJ}$ by all entropies being equal (see \ref{A1:DimColEquations}). Since this particular $MMI_{CIJ}$ violation is identical to the violation of $MMI_{cIJ}~\eqref{eq:MMI-Fail-G_m}$ exhibited by states $\ket{G}$, we define the following generalization of $\ket{G}$ that is guaranteed to always fail $MMI_{CIJ}$~\eqref{eq:MMI-Fail-G_M}.
\begin{observation}
\label{obs:proper-generalization-G}
    Define $G^{\prime}$ to have the following properties:
    \textnormal{
    \begin{enumerate}[label=(\arabic*)]
    \item Is a generalized star $\mathcal{G}$ under the partition $\{C,I,J,K\}$.
    \item The $\{C,I,J,K\}$ partition obeys
    $$
    \col{CI}= \col{CJ}= \col{CK} \quad \text{ (generalized subgraph anchoring).}
    $$
\end{enumerate}}
Then, the state $\ket{G^{\prime}}$ will always fail $MMI_{CIJ}$~\eqref{eq:MMI-Fail-G_M} as a result of
\begin{equation}
    S_A = \ra{CI}= \ra{CJ}= \ra{CJ}, \text{ for every } A \subset CIJ.
\end{equation}
\end{observation}
The characteristic $MMI_{CIJ}$~\eqref{eq:MMI-Fail-G_M} failure exhibited by $\ket{G^{\prime}}$ in Observation~\ref{obs:proper-generalization-G} can be interpreted as a generalized $GHZ$-type MMI violation. Furthermore, for a state $\ket{\mathcal{G}}$, the column spaces of $CI$, $CJ$, and $CK$ are trivially distributive whenever two or more subspaces are equal, preventing the state from strictly satisfying $MMI_{CIJ}~\eqref{eq:MMI-Fail-G_M}$ by Observation~\ref{obs:cannot-strictly-satisfy-generalized-graph}.

In this section, we extended the central-vertex star graphs $G$ introduced in Section~\ref{section:single-vertex-center} to a broader class of graphs $\mathcal{G}$, in which the central vertex $c$ is replaced by a multi-vertex subgraph $C$. We derived the conditions under which the graph states $\ket{\mathcal{G}}$ violate $MMI_{CIJ}$~\eqref{eq:MMI-Fail-G_M}, demonstrating that subgraph anchoring is no longer a sufficient condition to guarantee violation. Instead we showed how the relationship between column space distributivity and intersection, in the adjacency submatrices for parties involved in $MMI_{CIJ}$~\eqref{eq:MMI-Fail-G_M}, determines satisfaction, saturation, and failure of the inequalities. We provided a classification of $MMI_{CIJ}$~\eqref{eq:MMI-Fail-G_M} evaluation for all realizations of $\mathcal{G}$, summarized in Table~\ref{tab:MMIcases}. Finally, we verified that graphs $\mathcal{G}$ reduce to the form $G$ when the central subgraph $C$ consists of a single vertex, and introduced a generalization of $G$ that is guaranteed to always fail $MMI_{CIJ}$~\eqref{eq:MMI-Fail-G_M} by containing a $GHZ$-type subgraph. In the next section we provide a physical interpretation for our $MMI_{CIJ}$~\eqref{eq:MMI-Fail-G_M} failure conditions in the context of quantum circuits.

\subsection{A Physical Interpretation of the $MMI_{CIJ}$ Violation Conditions}
\label{subsection:physical-interpretation}

In Section~\ref{section:generalizedStarGraphs} we established conditions to determine if a state whose graph representation is $LC$ equivalent to a graph of type $\mathcal{G}$ satisfies, saturates, or violates $MMI_{CIJ}$ \eqref{eq:MMI-Fail-G_M}. Our analysis relied on an algebraic reformulation of the $MMI_{CIJ}$~\eqref{eq:MMI-Fail-G_M} inequality into a relation between the column spaces of adjacency matrix blocks $CI$, $CJ$, and $CK$, for $\ket{\mathcal{G}}$ subsystems $C,\ I,\ J,$ and $K$. We now give a physical interpretation for these algebraic conditions by moving from a vector space formalism into the language of quantum circuits.

\subsubsection{From Column Spaces to Reachable Circuits}

Following Eq.\ \eqref{GraphStateDefinition}, an $n$-qubit graph state $\ket{G}$ is constructed by a sequence of $CZ_{i,j}$ gates acting on the state $\ket{+}^{\otimes n}$. Since $CZ$ gates acting on different qubit pairs commute, the effect of all $CZ_{i,j}$ used to prepare $\ket{G}$ can be piecewise analyzed by examining the subset of $CZ$ gates acting between different subsystems. For example, the edges in $\mathcal{G}$ between vertex sets $C$ and $I$ are realized in $\ket{\mathcal{G}}$ by a set of circuits, one for each qubit in $I$. For each vertex
$i \in I$, the edges connecting $i$ to the center $C$ are described by the $i$-th column of the block matrix $CI$, and correspond to a circuit $U_i$ composed of 
$CZ_{c,i}$ gates, with $c \in C$. This correspondence between adjacency matrix columns and $CZ$ gates establishes a dictionary between the vector algebraic and quantum circuit formalisms, complete with the following relations. 
\begin{enumerate}
    \item \textbf{Vector Addition:} Bitwise addition of two column vectors, i.e.\ $\mathbf{v}_i \oplus \mathbf{v}_j$, corresponds to applying their respective circuits sequentially $U_i U_j$.
    \item \textbf{Column Space:} The column space $\col{CI}$, the linear span of $CI$ columns, represents the set of all circuits built by multiplying smaller circuits $\{U_i\}_{i\in I}$. This set of ``reachable circuits'' forms a finite Abelian group.
\end{enumerate}
The vector-circuit correspondence allows column space properties, such as intersection and distributivity, to be interpreted as physical properties of the entangling circuits that prepare the physical state. Moreover, this association highlights that non-local Clifford operations corresponding to elementary column operations, such as replacing a column $\mathbf{v}_k$ with $\mathbf{v}_i \oplus \mathbf{v}_j$, preserve the column space of the associate adjacency block. In particular, any transformation that leaves the column spaces of $CI$, $CJ$, and $CK$ invariant preserves the classification of $\ket{\mathcal{G}}$ for $MMI_{CIJ}$~\eqref{eq:MMI-Fail-G_M} given in Table~\ref{tab:MMIcases}, even though such transformations generally alter the state's entropy vector (see Appendix \ref{A3:deducing-physical-interp-claims} for details).

\subsubsection{The Physical Meaning of a Nontrivial Intersection}
\label{section:Physical-Interpretation-Nontrivial-Intersection}

The necessary condition for the column spaces of $CI$, $CJ$, and $CK$ to share a non-trivial intersection in order to fail $MMI_{CIJ}$~\eqref{eq:MMI-Fail-G_M}, detailed in Observation~\ref{obs:necessary-condition-generalized-graph}, has direct physical implications in a quantum circuit framework. For block adjacency matrices $CI$, $CJ$, and $CK$, the condition $\col{CI} \cap \col{CJ} \cap \col{CK} \neq \varnothing$ implies that a vector exists which can be constructed independently from the columns of $CI$, from those of $CJ$, and from those of $CK$. In a quantum circuit framework, this criterion asserts that there must exist at least one entangling circuit, which we denote $U_{c}$, that can be constructed in three independent ways:
\begin{enumerate}
    \item By multiplying the set of circuits $\{U_i\}_{i \in \alpha}$ where $\alpha \subseteq I$,
    \item By multiplying the set of circuits $\{U_j\}_{j \in \beta}$ where $\beta \subseteq J$,
    \item By multiplying the set of circuits $\{U_k\}_{k \in \gamma}$ where $\gamma \subseteq K$.
\end{enumerate}
The graph state $\ket{K_4}$, represented by the four-star, illustrates the simplest possible construction of $U_c$. The MMI instances $\ket{K_4}$ fails are all of type $MMI_{cIJ}$ (see Table \ref{tab:MMI_eval_summary}), which is the special case of $MMI_{CIJ}$ with $|C|=1$. Thus, since $\ket{K_4}$ fails an $MMI_{CIJ}$, then, the preceding analysis ensures the existence of $U_c$ for any valid tripartite subsystem $\{ C=\{ 1\},I,J\}$ of $\ket{K_4}$. We can see this trivially from the adjacency matrix of $K_4$ \eqref{eq:GHZ-adj-matrix}. The adjacency blocks relevant to $MMI_{CIJ}$ that correspond to every tripartite subsystem $\{ C=\{ 1\},I,J\}$ of $\ket{K_4}$ are
\begin{equation}
    CI =  CJ =  CK = \begin{bmatrix} 1 \end{bmatrix},
\end{equation}
which is the exact preparation of $\ket{K_4}$: a $CZ$ between the central vertex $\{1\}$ and its leaves $\{2\}$, $\{3\}$, and $\{4\}$. In this case, we obtain the only nontrivial $CZ$ circuit between two vertices. 

Indeed, for states represented by graphs of type $G$, the existence of $U_c$ is equivalent to the subgraph anchoring condition from Section \ref{section:single-vertex-center}. For graph states whose representation is of the $\mathcal{G}$
family, the existence of $U_{c}$ acts as a similar, though weaker, constraint to the subgraph anchoring. Using $U_{c}$, we can construct a state $\ket{\mathcal{G'}}$, using only transformations on $\ket{\mathcal{G}}$ that preserve the column spaces of $CI, CJ$ and $CK$, where at least one qubit from each subsystem $i \in I,\ j \in J,\ k \in K$ shares identical connectivity with $C$. Importantly, while the state $\ket{\mathcal{G'}}$ may have a different entropy vector from $\ket{\mathcal{G}}$, its $MMI_{CIJ}$~\eqref{eq:MMI-Fail-G_M} evaluation is the same.

The construction of $U_{c}$ above offers an interesting physical perspective. Given that subsystems $I,\ J$, and $K$ are not directly connected by edges in $\mathcal{G}$, Observation~\ref{obs:necessary-condition-generalized-graph} implies that $MMI_{CIJ}$~\eqref{eq:MMI-Fail-G_M} violation requires the existence of a vertex set $\{i,\ j,\ k,\ c\}$, with $i \in I,\ j \in J,\ k \in K$, and $c \in C$, such that
\begin{equation}\label{K4Inclusion}
\mathcal{G}[\{i,j,k,c\}] \cong K_4.
\end{equation}
Otherwise stated, the subgraph of $\mathcal{G}$ induced by the vertices $\{i,\ j,\ k,\ c\}$ forms a four-star configuration. 

For graphs of type $\mathcal{G}$, the presence of a $K_4$ subgraph is required for $MMI_{CIJ}$~\eqref{eq:MMI-Fail-G_M} violation. Consequently, arbitrary $MMI_{CIJ}$ violations in $\ket{\mathcal{G}}$ must include (possibly recursive) compositions of this minimal MMI-violating subgraph (see Appendix \ref{A3:deducing-physical-interp-claims} for details). This principle holds if the core entanglement structure is preserved, i.e. the leaf vertices in each $K_4$ subgraph remain mutually disconnected within the graph. We highlight that the existence of a $K_4$ subgraph guarantees that there exists a partition under which the graph is (nontrivially) of type $\mathcal{G}$, since we can place the central vertex of $K_4$ in $C$, place each leaf of $K_4$ in $I,\ J,$ and $K$, respectively, and place all remaining vertices of the graph in $C$.


\subsubsection{The Physical Meaning of Distributivity}

Column space distributivity constrains how information is shared from the center vertex set $C$ to other subsystems. Consider a basis of $|C|$ independent and entangling operations $\mathcal{B} = \{B_1, B_2, \dots, B_{C}\}$. The basis $\mathcal{B}$ generates all entanglement arrangements accessible by concatenations of $B \in \mathcal{B}$, comprised of $CZ$ gates between the center $C$ and any target qubit. The column spaces $\col{CI}$, $\col{CJ}$, and $\col{CK}$ are then distributive if and only if every circuit reachable from the operations that built $CI$, $CJ$, and $CK$ can be generated by a subset $\mathcal{B}$ (details in Appendix \ref{A2:sufficient-Condition-General-Graph}). For example, all circuits constructable from the operators that prepared adjacencies in $CI$ may be generated by $\{B_1, B_3\}$, while the corresponding circuits constructable from the operators that prepared $CJ$ are generated by $\{B_2, B_3, B_5\}$, and so forth. 

The crucial feature underlying distributivity is that all three sets of circuits, constructable from $CI,\ CJ,$ and $CK$, are built from subsets of the \textit{same} basis. In a physical sense, this condition implies that the entanglement of $I$, $J$, and $K$ with the center $C$ is not arbitrary, but is instead compactly generated from the same shared set of operations. This implication is bidirectional, allowing the generating subsets%
\footnote{In general, the full basis remains undetermined unless the union of the column spaces of $CI$, $CJ$, or $CK$ is $\mathbb{Z}_{2}^{|C|}$. Nevertheless, we can always deduce the generating subsets in a compatible basis.} %
of $\mathcal{B}$ to be deduced given the set of entangling circuits for $CI$, $CJ$, and $CK$.

\subsection{MMI Failure for $\ket{\mathcal{G}}$ states is Not Fully Reducible to $MMI_{CIJ}$ Failure}

Our analysis thus far has focused on $MMI_{CIJ}~\eqref{eq:MMI-Fail-G_M}$. However, we have observed, and explicitly verified up through $n=8$ qubits, that the presence of a non-trivial intersection among column spaces $CI$, $CJ$, and $CK$ in $\ket{\mathcal{G}}$, while not sufficient to ensure violation of $MMI_{CIJ}~\eqref{eq:MMI-Fail-G_M}$ itself, is nevertheless sufficient to produce an $MMI$ violation in some other subsystem of the same state. To illustrate this observation, consider the graph%
\footnote{While this graph is disconnected, the missing edges do not impact $MMI_{CIJ}$~\eqref{eq:MMI-Fail-G_M} (see Eq.\ \eqref{eq:MMI-Fail-ColSpaces-G_M}). Furthermore, connecting vertices $k_1,\ k_2,$ and $k_3$ does not alter the satisfaction of $MMI_{CIJ}$~\eqref{eq:MMI-Fail-G_M}. The graph remains MMI-violating, since it becomes a graph of the form $G$.} %
shown in Figure~\ref{graph:MMI-case-4-satisfies-single} which represents a state of type $\ket{\mathcal{G}}$ with non-distributive, intersecting column spaces for $CI$, $CJ$, and $CK$. This state $\ket{\mathcal{G}}$ satisfies $MMI_{CIJ}$~\eqref{eq:MMI-Fail-G_M} on the vertex sets $C$, $I$, and $J$ shown in Figure~\ref{graph:MMI-case-4-satisfies-single}. The specific subspace constructions are given in Table~\ref{tab:Case4Examples} (a) of Appendix~\ref{A2:sufficient-Condition-General-Graph} under case ``Satisfies'', and a graph isomorphic to $\mathcal{G}$, which illustrates that Figure~\ref{graph:MMI-case-4-satisfies-single} is indeed a generalized star graph, is given in Table~\ref{tab:case-4-graph-examples}.
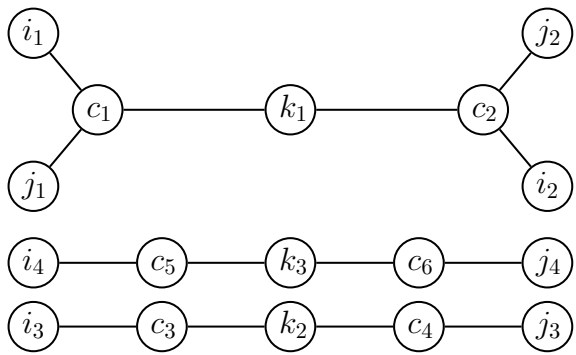
\begin{figure}[H]
\centering
\tikzset{
  stdnode/.style={
    circle, draw, fill=white,
    inner sep=1.8pt,
    minimum size=22pt,
    font=\large,
    text width=10pt,          
    text height=8pt,
    text depth=2pt,
    align=center
  },
  stdlines/.style={line width=0.9pt}
}
\resizebox{0.50\linewidth}{!}{
    \begin{tikzpicture}[stdlines, baseline=(current bounding box.center)]
      \node[stdnode] (c1) at (-3.0, 1.4) {$c_1$};
      \node[stdnode] (k1) at ( 0.0, 1.4) {$k_1$};
      \node[stdnode] (c2) at ( 3.0, 1.4) {$c_2$};
    
      \node[stdnode] (i1) at (-4, 2.6) {$i_1$};
      \node[stdnode] (j1) at (-4, 0.2) {$j_1$};
      \node[stdnode] (j2) at ( 4, 2.6) {$j_2$};
      \node[stdnode] (i2) at ( 4, 0.2) {$i_2$};
    
      \draw (i1)--(c1)--(k1)--(c2)--(j2);
      \draw (j1)--(c1);
      \draw (i2)--(c2);
    
      \node[stdnode] (i4) at (-4.0, -1.0) {$i_4$};
      \node[stdnode] (c5) at (-2.0, -1.0) {$c_5$};
      \node[stdnode] (k3) at ( 0.0, -1.0) {$k_3$};
      \node[stdnode] (c6) at ( 2.0, -1.0) {$c_6$};
      \node[stdnode] (j4) at ( 4.0, -1.0) {$j_4$};
      \draw (i4)--(c5)--(k3)--(c6)--(j4);
    
      \node[stdnode] (i3) at (-4.0, -2) {$i_3$};
      \node[stdnode] (c3) at (-2.0, -2) {$c_3$};
      \node[stdnode] (k2) at ( 0.0, -2) {$k_2$};
      \node[stdnode] (c4) at ( 2.0, -2) {$c_4$};
      \node[stdnode] (j3) at ( 4.0, -2) {$j_3$};
      \draw (i3)--(c3)--(k2)--(c4)--(j3);
    \end{tikzpicture}}
    \caption{Graph representation for a state $\ket{\mathcal{G}}$ with nontrivial intersection among non-distributive subspaces $\col{CI}$, $\col{CJ}$, and $\col{CK}$, that satisfies $MMI_{CIJ}~\eqref{eq:MMI-Fail-G_M}$. Vertex labels give subsystem membership $c_q \in C$, $i_q \in I$, $j_q \in J$, and $k_q \in K$.}
  \label{graph:MMI-case-4-satisfies-single}
\end{figure}

While the full graph in Figure~\ref{graph:MMI-case-4-satisfies-single} satisfies $MMI_{CIJ}~\eqref{eq:MMI-Fail-G_M}$, the uppermost subgraph containing vertices $\{ c_1, c_2, i_1, i_2, j_1, j_2, k_1 \}$ would fail $MMI_{CIJ}~\eqref{eq:MMI-Fail-G_M}$ on its own by virtue of Observation~\ref{obs:sufficient-generalized-star}. In fact, for every graph $\mathcal{G}$ possessing a non-trivial intersection we find \textit{some} violation of MMI on a subgraph of $\mathcal{G}$. While we have not successfully produced a generic proof of this conjecture, we have performed an exhaustive search up through $8$ qubits.%
\footnote{\label{footnote:Table14-Exception}The graph labeled (1) in Table~\ref{tab:MMI-Failing-EightQ-Graphs-Missing-Nontrivial-Int} is the only exception, as we chose an LC orbit representative with an induced four-star rather than the minimal-edge one.}
When a nontrivial intersection was identified, we provide the specific choice of $C$, $I$, and $J$ that contains the largest number of elements still satisfying this condition. All states obtained through this procedure are of type $\ket{\mathcal{G}}$ under the indicated partition. Among the $46$ MMI-satisfying graphs in Table~\ref{tab:MMI-Satisfying-EightQ-Graphs}, none exhibit a nontrivial intersection, nor do any states in their respective $LC$ orbits. Of the $136$ MMI-failing graphs for each entropy vector, Table~\ref{tab:MMI-Failing-EightQ-Graphs-Missing-Nontrivial-Int} reveals that $8$ fail MMI despite not having a nontrivial intersection for any valid choice of $C$, $I$, and $J$. Table~\ref{tab:grpah-representatives-examples} presents example graphs from Tables~\ref{tab:MMI-Satisfying-EightQ-Graphs}-~\ref{tab:MMI-Failing-EightQ-Graphs-Missing-Nontrivial-Int}; clicking on the label of each graph redirects to an applet that provides access to the entropy vector of its corresponding state.
\begin{center}
    \begin{tikzpicture}
    
        \begin{scope}
          \node[draw, rectangle, rounded corners,
                minimum width=0.33\linewidth, minimum height=1.8cm,
                fill=white, inner sep=3pt] (stylebox) at (0,0) {};
    
          \node[anchor=north, inner sep=8pt]
            at (stylebox.north) {\tiny\bfseries Vertex Shape: $\boldsymbol{\mathcal{G}}$ partition};
    
          \def\legendY{-0.1}
          \def\labelY{-0.75}
    
          \node[Knode,BorderStar,
                minimum size=7pt, inner sep=0pt]      (vs)   at (-1.65, \legendY) {};
          \node[Knode,BorderTriangle,
                minimum size=7pt, inner sep=0pt]      (vt)   at (-0.55, \legendY) {};
          \node[Knode,BorderSquare,
                minimum size=7pt, inner sep=0pt]      (vq)   at ( 0.55, \legendY) {};
          \node[Knode,
                minimum size=7pt, inner sep=0pt]      (vreg) at ( 1.65, \legendY) {};
    
          \node[anchor=south] at (-1.65, \labelY) {\scriptsize $C$};
          \node[anchor=south] at (-0.55, \labelY) {\scriptsize $I$};
          \node[anchor=south] at ( 0.55, \labelY) {\scriptsize $J$};
          \node[anchor=south] at ( 1.65, \labelY) {\scriptsize $K$};
        \end{scope}
        \begin{scope}[xshift=0.33\linewidth]
          \node[draw, rectangle, rounded corners,
                minimum width=0.3\linewidth, minimum height=1.8cm,
                fill=white, inner sep=3pt] (colorbox) at (0,0) {};
    
          \node[anchor=north, inner sep=8pt]
            at (colorbox.north) {\tiny\bfseries Vertex Color: failed $\boldsymbol{MMI_{CIJ}}$};
    
          \def\legendYc{-0.1}
          \def\labelYc{-0.75}
    
          \node[Cnode]     (vc) at (-1.65, \legendYc) {};
          \node[Inode]     (vi) at (-0.55, \legendYc) {};
          \node[Jnode]  (vj) at ( 0.55, \legendYc) {};
          \node[Knode]     (vk) at ( 1.65, \legendYc) {};
    
          \node[anchor=south] at (-1.65, \labelYc) {\scriptsize $C$};
          \node[anchor=south] at (-0.55, \labelYc) {\scriptsize $I$};
          \node[anchor=south] at ( 0.55, \labelYc) {\scriptsize $J$};
          \node[anchor=south] at ( 1.65, \labelYc) {\scriptsize $K$};
        \end{scope}
        
        \begin{scope}[xshift=0.66 \linewidth]
          \node[draw, rectangle, rounded corners,
                minimum width=0.33\linewidth, minimum height=1.8cm,
                fill=white, inner sep=3pt] (colorbox) at (0,0) {};
    
          \node[anchor=north, inner sep=8pt]
            at (colorbox.north) {\tiny\bfseries Vertex Color: failed $\boldsymbol{MMI_{STU}}$};
    
          \def\legendYc{-0.1}
          \def\labelYc{-0.75}
    
          \node[Anode]     (vc) at (-1.65, \legendYc) {};
          \node[Bnode]     (vi) at (-0.55, \legendYc) {};
          \node[CnodeABC]  (vj) at ( 0.55, \legendYc) {};
          \node[Knode]     (vk) at ( 1.65, \legendYc) {};
    
          \node[anchor=south] at (-1.65, \labelYc) {\scriptsize $S$};
          \node[anchor=south] at (-0.55, \labelYc) {\scriptsize $T$};
          \node[anchor=south] at ( 0.55, \labelYc) {\scriptsize $U$};
          \node[anchor=south] at ( 1.65, \labelYc) {\scriptsize $\overline{STU}$};
        \end{scope}
        
      \end{tikzpicture}
\end{center}
\begin{table}[H]
\centering
\begin{tabular}{@{}c@{\hspace{8pt}}c@{\hspace{8pt}}c@{}}

\begin{minipage}[t]{0.30\textwidth}\centering
  \graphpairbox{%
    \begin{minipage}[t]{0.48\linewidth}\centering
      \vspace*{2pt}%
      \resizebox{0.960\linewidth}{!}{%
        \begin{tikzpicture}[scale=1]
          \path[use as bounding box] (-0.100, -0.100) rectangle (3.600, 3.780);
          \node[anchor=south,  inner sep=0pt] at (1.750, 3.680) {\large (4004, 3766, 0)};
          \node[Knode, BorderTriangle] (v1) at (3.290, 2.530) {\scriptsize 1};
          \node[Knode, BorderSquare]  (v2) at (0.210, 2.528) {\scriptsize 2};
          \node[Knode]                (v3) at (1.753, 0.970) {\scriptsize 3};
          \node[Knode, BorderStar]    (v4) at (2.748, 2.319) {\scriptsize 4};
          \node[Knode, BorderStar]    (v5) at (0.751, 2.318) {\scriptsize 5};
          \node[Knode, BorderStar]    (v6) at (2.079, 2.051) {\scriptsize 6};
          \node[Knode, BorderStar]    (v7) at (1.421, 2.051) {\scriptsize 7};
          \node[Knode, BorderStar]    (v8) at (1.751, 1.571) {\scriptsize 8};
          \begin{scope}[on background layer]
            \draw[line width=0.9pt, line cap=round] (v1) -- (v4);
            \draw[line width=0.9pt, line cap=round] (v2) -- (v5);
            \draw[line width=0.9pt, line cap=round] (v3) -- (v8);
            \draw[line width=0.9pt, line cap=round] (v4) -- (v6);
            \draw[line width=0.9pt, line cap=round] (v5) -- (v7);
            \draw[line width=0.9pt, line cap=round] (v6) -- (v7);
            \draw[line width=0.9pt, line cap=round] (v6) -- (v8);
            \draw[line width=0.9pt, line cap=round] (v7) -- (v8);
          \end{scope}
        \end{tikzpicture}%
      }\\[3pt]
      {\small\bfseries \href{https://jesus99222.github.io/MMI-Failure-For-Graph-States/starRepsSatisfySvecs/33.html}{(33)}}%
      \vspace*{2pt}%
    \end{minipage}%
    \hspace{0.02\linewidth}%
    \begin{minipage}[t]{0.48\linewidth}\centering
      \vspace*{2pt}%
      \resizebox{0.960\linewidth}{!}{%
        \begin{tikzpicture}[scale=1]
          \path[use as bounding box] (-0.100, -0.100) rectangle (3.600, 3.780);
          \node[anchor=south,  inner sep=0pt] at (1.750, 3.680) {\large (5152, 2618, 0)};
          \node[Knode, BorderTriangle] (v1) at (0.286, 3.119) {\scriptsize 1};
          \node[Knode, BorderSquare]  (v2) at (3.211, 3.120) {\scriptsize 2};
          \node[Knode]                (v3) at (1.750, 1.479) {\scriptsize 3};
          \node[Knode, BorderStar]    (v4) at (0.210, 1.545) {\scriptsize 4};
          \node[Knode, BorderStar]    (v5) at (3.290, 1.548) {\scriptsize 5};
          \node[Knode, BorderStar]    (v6) at (1.751, 0.380) {\scriptsize 6};
          \node[Knode, BorderStar]    (v7) at (1.263, 2.763) {\scriptsize 7};
          \node[Knode, BorderStar]    (v8) at (2.235, 2.765) {\scriptsize 8};
          \begin{scope}[on background layer]
            \draw[line width=0.9pt, line cap=round] (v1) -- (v4);
            \draw[line width=0.9pt, line cap=round] (v1) -- (v7);
            \draw[line width=0.9pt, line cap=round] (v1) -- (v8);
            \draw[line width=0.9pt, line cap=round] (v2) -- (v5);
            \draw[line width=0.9pt, line cap=round] (v2) -- (v7);
            \draw[line width=0.9pt, line cap=round] (v2) -- (v8);
            \draw[line width=0.9pt, line cap=round] (v3) -- (v6);
            \draw[line width=0.9pt, line cap=round] (v3) -- (v7);
            \draw[line width=0.9pt, line cap=round] (v3) -- (v8);
            \draw[line width=0.9pt, line cap=round] (v4) -- (v6);
            \draw[line width=0.9pt, line cap=round] (v4) -- (v7);
            \draw[line width=0.9pt, line cap=round] (v5) -- (v6);
            \draw[line width=0.9pt, line cap=round] (v5) -- (v8);
          \end{scope}
        \end{tikzpicture}%
      }\\[3pt]
      {\small\bfseries \href{https://jesus99222.github.io/MMI-Failure-For-Graph-States/starRepsSatisfySvecs/41.html}{(41)}}%
      \vspace*{2pt}%
    \end{minipage}%
  }
  \panelcaption{a}{Graphs satisfying every MMI (Table~\ref{tab:MMI-Satisfying-EightQ-Graphs})}
\end{minipage}
&
\begin{minipage}[t]{0.30\textwidth}\centering
  \graphpairbox{%
    \begin{minipage}[t]{0.48\linewidth}\centering
      \vspace*{2pt}%
      \resizebox{0.960\linewidth}{!}{%
        \begin{tikzpicture}[scale=1]
          \path[use as bounding box] (-0.100, -0.100) rectangle (3.600, 3.780);
          \node[anchor=south,  inner sep=0pt] at (1.750, 3.680) {\large (4596, 3126, 48)};
          \node[Inode] (v1) at (3.290, 1.751) {\scriptsize 1};
          \node[Jnode] (v2) at (2.608, 1.751) {\scriptsize 2};
          \node[Knode] (v3) at (1.222, 0.534) {\scriptsize 3};
          \node[Knode] (v4) at (1.220, 2.966) {\scriptsize 4};
          \node[Knode] (v5) at (0.212, 1.172) {\scriptsize 5};
          \node[Cnode] (v6) at (0.210, 2.327) {\scriptsize 6};
          \node[Cnode] (v7) at (2.461, 0.883) {\scriptsize 7};
          \node[Cnode] (v8) at (2.460, 2.618) {\scriptsize 8};
          \begin{scope}[on background layer]
            \draw[line width=0.9pt, line cap=round] (v1) -- (v7);
            \draw[line width=0.9pt, line cap=round] (v1) -- (v8);
            \draw[line width=0.9pt, line cap=round] (v2) -- (v7);
            \draw[line width=0.9pt, line cap=round] (v2) -- (v8);
            \draw[line width=0.9pt, line cap=round] (v3) -- (v5);
            \draw[line width=0.9pt, line cap=round] (v3) -- (v7);
            \draw[line width=0.9pt, line cap=round] (v4) -- (v6);
            \draw[line width=0.9pt, line cap=round] (v4) -- (v8);
            \draw[line width=0.9pt, line cap=round] (v5) -- (v6);
          \end{scope}
        \end{tikzpicture}%
      }\\[3pt]
      {\small\bfseries \href{https://jesus99222.github.io/MMI-Failure-For-Graph-States/minimalStarNonTFailSvecs/30.html}{(30)}}%
      \vspace*{2pt}%
    \end{minipage}%
    \hspace{0.02\linewidth}%
    \begin{minipage}[t]{0.48\linewidth}\centering
      \vspace*{2pt}%
      \resizebox{0.960\linewidth}{!}{%
        \begin{tikzpicture}[scale=1]
          \path[use as bounding box] (-0.100, -0.100) rectangle (3.600, 3.780);
          \node[anchor=south,  inner sep=0pt] at (1.750, 3.680) {\large (0, 966, 6804)};
          \node[Inode] (v1) at (3.290, 1.327) {\scriptsize 1};
          \node[Jnode] (v2) at (2.440, 0.252) {\scriptsize 2};
          \node[Knode] (v3) at (1.068, 0.247) {\scriptsize 3};
          \node[Cnode] (v4) at (0.210, 1.317) {\scriptsize 4};
          \node[Cnode] (v5) at (1.744, 3.253) {\scriptsize 5};
          \node[Cnode] (v6) at (0.511, 2.655) {\scriptsize 6};
          \node[Cnode] (v7) at (2.981, 2.662) {\scriptsize 7};
          \node[Cnode] (v8) at (1.748, 1.672) {\scriptsize 8};
          \begin{scope}[on background layer]
            \draw[line width=0.9pt, line cap=round] (v1) -- (v8);
            \draw[line width=0.9pt, line cap=round] (v2) -- (v8);
            \draw[line width=0.9pt, line cap=round] (v3) -- (v8);
            \draw[line width=0.9pt, line cap=round] (v4) -- (v8);
            \draw[line width=0.9pt, line cap=round] (v5) -- (v8);
            \draw[line width=0.9pt, line cap=round] (v6) -- (v8);
            \draw[line width=0.9pt, line cap=round] (v7) -- (v8);
          \end{scope}
        \end{tikzpicture}%
      }\\[3pt]
      {\small\bfseries \href{https://jesus99222.github.io/MMI-Failure-For-Graph-States/minimalStarNonTFailSvecs/128.html}{(128)}}%
      \vspace*{2pt}%
    \end{minipage}%
  }
  \panelcaption{b}{Graphs failing $MMI_{CIJ}$ (Table~\ref{tab:MMI-Failing-EightQ-Graphs-NontrivialInt})}
\end{minipage}
&
\begin{minipage}[t]{0.30\textwidth}\centering
  \graphpairbox{%
    \begin{minipage}[t]{0.48\linewidth}\centering
      \vspace*{2pt}%
      \resizebox{0.960\linewidth}{!}{%
        \begin{tikzpicture}[scale=1]
          \path[use as bounding box] (-0.100, -0.100) rectangle (3.600, 3.780);
          \node[anchor=south,  inner sep=0pt] at (1.750, 3.680) {\large (4920, 2846, 4)};
          \node[Anode,BorderTriangle] (v1) at (3.290,1.751) {\scriptsize 1};
          \node[Bnode,BorderSquare]  (v2) at (1.230,1.184) {\scriptsize 2};
          \node[CnodeABC,BorderStar] (v3) at (1.230,2.316) {\scriptsize 3};
          \node[CnodeABC,BorderStar] (v4) at (0.210,1.304) {\scriptsize 4};
          \node[Bnode]               (v5) at (0.210,2.197) {\scriptsize 5};
          \node[Anode,BorderStar]    (v6) at (2.193,1.750) {\scriptsize 6};
          \node[Knode,BorderStar]    (v7) at (1.556,1.750) {\scriptsize 7};
          \node[Knode,BorderStar]    (v8) at (0.743,1.751) {\scriptsize 8};
          \begin{scope}[on background layer]
            \draw[line width=0.9pt,line cap=round] (v1) -- (v6);
            \draw[line width=0.9pt,line cap=round] (v2) -- (v4);
            \draw[line width=0.9pt,line cap=round] (v2) -- (v6);
            \draw[line width=0.9pt,line cap=round] (v2) -- (v7);
            \draw[line width=0.9pt,line cap=round] (v2) -- (v8);
            \draw[line width=0.9pt,line cap=round] (v3) -- (v5);
            \draw[line width=0.9pt,line cap=round] (v3) -- (v6);
            \draw[line width=0.9pt,line cap=round] (v3) -- (v7);
            \draw[line width=0.9pt,line cap=round] (v3) -- (v8);
            \draw[line width=0.9pt,line cap=round] (v4) -- (v5);
            \draw[line width=0.9pt,line cap=round] (v4) -- (v8);
            \draw[line width=0.9pt,line cap=round] (v5) -- (v8);
            \draw[line width=0.9pt,line cap=round] (v6) -- (v7);
            \draw[line width=0.9pt,line cap=round] (v7) -- (v8);
          \end{scope}
        \end{tikzpicture}%
      }\\[3pt]
      {\small\bfseries \href{https://jesus99222.github.io/MMI-Failure-For-Graph-States/minimalStarTFailnoCIJSvecs/1.html}{(1)}}%
      \vspace*{2pt}%
    \end{minipage}%
    \hspace{0.02\linewidth}%
    \begin{minipage}[t]{0.48\linewidth}\centering
      \vspace*{2pt}%
      \resizebox{0.960\linewidth}{!}{%
        \begin{tikzpicture}[scale=1]
          \path[use as bounding box] (-0.100, -0.100) rectangle (3.600, 3.780);
          \node[anchor=south,  inner sep=0pt] at (1.750, 3.680) {\large (5244, 2510, 16)};
          \node[Anode, BorderTriangle] (v1) at (0.241, 0.210) {\scriptsize 1};
          \node[Bnode, BorderSquare]  (v2) at (3.290, 0.242) {\scriptsize 2};
          \node[CnodeABC]             (v3) at (0.210, 3.258) {\scriptsize 3};
          \node[CnodeABC, BorderStar] (v4) at (3.257, 3.290) {\scriptsize 4};
          \node[Anode, BorderStar]    (v5) at (0.692, 1.740) {\scriptsize 5};
          \node[Bnode, BorderStar]    (v6) at (2.806, 1.761) {\scriptsize 6};
          \node[Knode, BorderStar]    (v7) at (1.761, 0.693) {\scriptsize 7};
          \node[CnodeABC, BorderStar] (v8) at (1.738, 2.807) {\scriptsize 8};
          \begin{scope}[on background layer]
            \draw[line width=0.9pt, line cap=round] (v1) -- (v5);
            \draw[line width=0.9pt, line cap=round] (v1) -- (v7);
            \draw[line width=0.9pt, line cap=round] (v2) -- (v6);
            \draw[line width=0.9pt, line cap=round] (v2) -- (v7);
            \draw[line width=0.9pt, line cap=round] (v3) -- (v5);
            \draw[line width=0.9pt, line cap=round] (v3) -- (v8);
            \draw[line width=0.9pt, line cap=round] (v4) -- (v6);
            \draw[line width=0.9pt, line cap=round] (v4) -- (v8);
            \draw[line width=0.9pt, line cap=round] (v5) -- (v6);
            \draw[line width=0.9pt, line cap=round] (v7) -- (v8);
          \end{scope}
        \end{tikzpicture}%
      }\\[3pt]
      {\small\bfseries \href{https://jesus99222.github.io/MMI-Failure-For-Graph-States/minimalStarTFailnoCIJSvecs/5.html}{(5)}}%
      \vspace*{2pt}%
    \end{minipage}%
  }
  \panelcaption{c}{Graphs failing MMI but not $MMI_{CIJ}$ (Table~\ref{tab:MMI-Failing-EightQ-Graphs-Missing-Nontrivial-Int})}
\end{minipage}

\end{tabular}

\caption{Graph examples from the classification of star graphs $\mathcal{G}$ realizing all unique eight-qubit entropy vectors, organized by MMI failure and the presence of a nontrivial intersection (see Tables~\ref{tab:MMI-Satisfying-EightQ-Graphs}–\ref{tab:MMI-Failing-EightQ-Graphs-Missing-Nontrivial-Int} for the full clasification). Vertex shape encodes a partition defining each $\mathcal{G}$, and vertex color marks a failing MMI instance; in Table~\ref{tab:grpah-representatives-examples} these two partitions coincide, so we indicate it only by coloring. MMI evaluation counts are given as the ordered triple (Satisfies, Saturates, Fails), with each entry equal to the number of instances yielding that outcome. Clicking each graph’s label redirects to a website that provides easy access to the entropy vector of its state.} 
\label{tab:grpah-representatives-examples}
\end{table}
We present the following observations from the data%
\footnote{All data is available and reproducible from \textit{Mathematica} notebooks in \cite{fuentes_munizzi2025mmi}. We give code for graph and tableau representations of stabilizer states, functions to generate/transform states, compute entropy vectors and symmetries, and evaluate entropy inequalities. The repository includes data for Appendix \ref{B1:MMI-Failure-Among-Stabilizer-States} \label{footnote:Mathematica-Notebook} on MMI failure for stabilizer states through seven qubits,}, 
which found no counterexamples in an exhaustive search through eight qubits; limited exploration at nine and ten qubits has also yielded no counterexamples. Tables~\ref{tab:MMI-Satisfying-EightQ-Graphs}-\ref{tab:MMI-Failing-EightQ-Graphs-Missing-Nontrivial-Int} list MMI-satisfying and MMI-failing graphs, respectively, for each of the $182$ entropy vectors of stabilizer states at $n=8$, unique up to symmetry~\cite{Danielsen_2006}. For each representative graph $G$, we generated and searched its full $LC$ orbit to identify a minimal-edge given in
Tables~\ref{tab:MMI-Satisfying-EightQ-Graphs}-\ref{tab:MMI-Failing-EightQ-Graphs-Missing-Nontrivial-Int}.
\begin{observation}
    Every graph state, up to $n=8$, is $LC$-equivalent to a $\ket{\mathcal{G}}$ state for some partition.
\end{observation}

\begin{observation}
    The property that $\ket{\mathcal{G}}$ has a nontrivial intersection among the column spaces $CI$, $CJ$, and $CK$, for some valid choice of $\{C,I,J\}$, is not $LC$-invariant.
\end{observation}

\begin{observation}
    Although every MMI-violating graph, up to $n=6$, is $LC$-equivalent to a $\mathcal{G}$ graph with nontrivial intersection on the column spaces $CI$, $CJ$, and $CK$, this condition is not necessary for generic MMI failure on $\ket{\mathcal{G}}$.
\end{observation}

\begin{observation}
    Possessing a nontrivial intersection of the column spaces $CI$, $CJ$, and $CK$ is a sufficient condition for the failure of \textbf{some} MMI instance, though not that specific $MMI_{CIJ}$, in $\ket{\mathcal{G}}$ states. Table~\ref{tab:MMI-Satisfying-EightQ-Graphs} confirms this observation for $\ket{\mathcal{G}}$ up through $n=8$.
\end{observation}

\begin{observation}
    \label{obs:path-cycle-satisfy-MMIl}
    Path and cycle graphs, which never realize a nontrivial intersection, as well as specific combinations of path and cycle graphs, form families of graphs that never fail MMI. Table~\ref{tab:MMI-Satisfying-EightQ-Graphs} confirms this observation for $\ket{\mathcal{G}}$ up through $n=8$.
\end{observation}

We now conclude the series of observations from Tables~\ref{tab:MMI-Satisfying-EightQ-Graphs}-\ref{tab:MMI-Failing-EightQ-Graphs-Missing-Nontrivial-Int} with a forbidden-subgraph conjecture. In Section~\ref{section:Physical-Interpretation-Nontrivial-Intersection}, we showed that the necessary condition~\ref{obs:necessary-condition-generalized-graph} for $MMI_{CIJ}$ violation in a graph state $\ket{\mathcal{G}}$ implies that $\mathcal{G}$ contains an induced four-star. Therefore, we expect, and indeed observe, every graph in Table~\ref{tab:MMI-Failing-EightQ-Graphs-NontrivialInt} to have an induced four-star. Additionally, we can explicitly observe the graphs in Table~\ref{tab:MMI-Failing-EightQ-Graphs-Missing-Nontrivial-Int}, which are of type $\mathcal{G}$ and fail MMI not of type $MMI_{CIJ}$\eqref{eq:MMI-Fail-G_M}, to have an induced four-star as well. Moreover, we confirmed Observation~\ref{obs:path-cycle-satisfy-MMIl} exhaustively up to $n=10$, providing further evidence that path and cycle graphs—whose maximal vertex connectivity is $2$ and which therefore contain no induced four-star—appear to be graph families that never fail MMI~\eqref{eq:MMI}.

Ultimately, these computational and analytic observations provide evidence that an induced four-star may be necessary for violating \textit{any} MMI~\eqref{eq:MMI} inequality, motivating the conjecture that any MMI-violating graph state has some graph representation containing a $K_4$ graph. Then, by the argument in the paragraph following Eq.\ \eqref{K4Inclusion}, there is a partition under which that graph representation is nontrivially type $\mathcal{G}$. Formally, we conjecture the following: 
\begin{conjecture}[Forbidden-Subgraph]
    \label{conjecture:forbidden-subgraph}
    Any graph state that violates MMI \eqref{eq:MMI} has a graph representation that is $LC$-equivalent to a graph $H$ containing a $K_4$ subgraph. $H$ is $\mathcal{G}$ with respect to \emph{some} partition.
\end{conjecture}%
\noindent

In this section, we developed a framework for identifying and characterizing specific MMI violations in graph states. We proved that graphs of type $G$, a single central vertex $c$ connected to mutually disjoint and disconnected subgraphs, necessarily violate $MMI_{cIJ}~\eqref{eq:MMI-Fail-G_m}$. We then generalized this construction to consider graphs of type $\mathcal{G}$, where the vertex $c$ is promoted to a multi-vertex subgraph $C$. In this generalization, we derived conditions to determine whether states $\ket{\mathcal{G}}$ satisfy, saturate, or fail $MMI_{CIJ}$, based on the intersection and distributive properties adjacency submatrix column spaces in $\mathcal{G}$. Summarized in Table~\ref{tab:MMIcases}, our classification revealed that a nontrivial intersection among these column spaces is necessary, and when combined with distributivity is sufficient to produce $MMI_{CIJ}$ violation.

Physically, the algebraic relations derived in this section correspond to the existence of shared entangling circuits between subsystems, with $MMI_{CIJ}$ failure linked to the presence of a $K_4$ subgraph. We conjectured that every MMI-violating stabilizer state is $LC$-equivalent to a generalized star graph $\mathcal{G}$ under \emph{some} partition. We exhaustively verify this conjecture up through $n=8$ qubits, though the conjecture may hold at higher qubit number. We now conclude by discussing broader implications and extensions of these results.

\section{Discussion: Toward a Complete Characterization of MMI Violation in Stabilizer States}\label{Discussion}

In this work, we progress towards a complete characterization of monogamy of mutual information (MMI) violation for graph states by identifying algebraic conditions on adjacency matrices that determine the evaluation of specific MMI instances. We first present a family of graphs $G$, composed of disjoint subgraphs connected only through a single-vertex center $c$ (shown in Figure~\ref{fig:G_m-Graph}), and prove that the associated graph states $\ket{G}$ will always fail the set of MMI instances $MMI_{cIJ}$, defined in Eq.\ \eqref{eq:MMI-Fail-G_m}. The set of MMI instances $MMI_{cIJ}$ is defined by a pairwise disjoint tripartite subsystem $(\{c\},I,J)$, with a nonempty complement $K$, such that all paths between $I$ and $J$ must include only $c$.

We next extend our analysis to the broader class of generalized star graphs $\mathcal{G}$, in which the central node $c$ is promoted from a single vertex to an arbitrary subgraph $C$. Unlike $\ket{G}$, the corresponding graph states $\ket{\mathcal{G}}$ are not guaranteed to violate $MMI_{CIJ}$ based solely on their star-like graph topology. We therefore provide a complete characterization of $MMI_{CIJ}$ evaluation for any $\ket{\mathcal{G}}$, expressed in terms of the algebraic structure of adjacency submatrices that encode nontrivial adjacencies between elements in the partition $\{C,\ I,\ J,\ K\}$. The essential features are captured by matrix subblocks $CI$, $CJ$, and $CK$, which we systematically classify in Table~\ref{tab:MMIcases} and summarize by the following.
\begin{enumerate}
\item A nontrivial intersection of the column spaces of $CI$, $CJ$, and $CK$ is a \emph{necessary} condition to violate $MMI_{CIJ}$.
\item Requiring that intersections distribute over subspace sums, for the column spaces of $CI$, $CJ$, and $CK$, gives a \emph{sufficient} condition for $MMI_{CIJ}$ violation.
\item Configurations which are both trivially intersecting and distributive guarantee saturation of $MMI_{CIJ}$, whereas trivially intersecting and non-distributive configurations strictly satisfy $MMI_{CIJ}$.
\item Configurations with a non-trivial intersection and non-distributive column spaces are unconstrained in their $MMI_{CIJ}$ evaluation.
\end{enumerate}
For each case in the above categorization, we provide an explicit graph construction in Tables~\ref{tab:MMI-characterization-examples-1-3} and \ref{tab:case-4-graph-examples}.

Our analysis focused on MMI violations in graph states, providing a complete characterization of $MMI_{CIJ}$ evaluation for generalized star graphs $\mathcal{G}$. Nevertheless, computational evidence identifies graph states that fail MMI while possessing no single nontrivial intersection, nor any such intersection in their full $LC$ orbit, for any valid choice of $\{C,I,J\}$ (see Table \ref{tab:MMI-Failing-EightQ-Graphs-Missing-Nontrivial-Int}). These examples therefore demonstrate MMI violations not captured by the presence of a star graph structure. 

The analysis techniques developed in this work can, in principle, be extended to explore generic MMI instances, enabling a fully general study of MMI failure in graph states. The primary challenge for this extension is that, for general MMI instances, the relevant submatrix ranks involve decomposing the rank of a $2 \times 2$ block matrix, preventing a simplification analogous to Eq.\ \eqref{eq:appendix-MMI-Fail-ColSpaces-G_M}. Further work might seek to alternatively express the ranks of block matrices in terms of their blocks, thereby avoiding recursive applications of the subspace addition formula in Eq.\ \eqref{eq:DimSubspaceAddition}. Additional approaches include leveraging rank-exposing decompositions, or related matrix identities that enable the simplification of block matrix rank into an amenable form, ultimately facilitating a general framework for studying MMI violation.

This work establishes a correspondence between the algebraic features that determine $MMI$ evaluation for graph states, and the operators used to prepare such states. A nontrivial intersection among the column spaces of $CI$, $CJ$, and $CK$ implies the existence of an operator that can be constructed independently of the operators used to entangle each pair of subsystems. Distributivity, meanwhile, reveals a higher degree of organization: there exists a minimal generating set of operators for all entanglement patterns between qubits and the center $C$, whose subsets reproduce the entangling operators for $CI$, $CJ$, and $CK$ individually.

One natural follow-up to this work is to characterize MMI failure for the graph states shown in Table~\ref{tab:MMI-Failing-EightQ-Graphs-Missing-Nontrivial-Int}, which satisfy all $MMI_{CIJ}$ instances under which they are $\mathcal{G}$ graphs. However, these states still violate MMI instances, just not ones under which there is an LC orbit member that is $\mathcal{G}$. For example, graph $(1)$ in Table~\ref{tab:MMI-Failing-EightQ-Graphs-Missing-Nontrivial-Int} fails instances $MMI_{\{16,\ 25,\ 34\}}$, $MMI_{\{16,\ 34,\ 78\}}$, and $MMI_{\{25,\ 34,\ 78\}}$, with $MMI_{\{16,\ 25,\ 34\}}$ colored in the figure. Up to symmetry, the unique instances failed are $MMI_{\{16,\ 25,\ 34\}}$ and $MMI_{\{25,\ 34,\ 78\}}$, out of the $7770$ distinct MMI instances that exist at $n=8$ qubits. Comparing these MMI failures, e.g. $MMI_{\{16,\ 25,\ 34\}}$ and $MMI_{\{25,\ 34,\ 78\}}$, to those observed in the graphs of Table~\ref{tab:MMI-Failing-EightQ-Graphs-NontrivialInt}, which violate the chosen $MMI_{CIJ}$, may provide insights towards a complete characterization of MMI violation in graph states.

The correspondence between entropy inequality evaluation and the algebraic structure of adjacency matrix components extends beyond MMI. Higher-party holographic inequalities are especially appealing targets given their physical relevance and inherent symmetry. All holographic entropy inequalities are are balanced, meaning each subsystem appears the same number of times on each side of the inequality~\cite{Bao_2015,Czech:2022fzb,He:2020xuo}. The methods introduced in this work are particularly effective for evaluating such inequalities on generalized star graphs, as the subgraph disjointness condition~\ref{condition: genDisconnection} ensures that each entropy reduces to the rank of a single row of block matrices for inequalities analogous to $MMI_{CIJ}$. Rewriting these inequalities using the adjacency matrix rank expression in Eq.\ \eqref{eq:EEadj}, and applying the subspace dimension addition relation, allows the inequality to be completely expressed in terms of intersection dimensions of specific subspaces. For balanced inequalities on $N$ regions, such as those constraining the entropy vectors for holographic states, all individual submatrix ranks cancel. This reformulation should enable the systematic exploration of a broader set of entropy inequalities using the same techniques presented herein.

The discovery of the graphs $G$, shown in Figure~\ref{fig:G_m-Graph}, was inspired by the repeated appearance of a four-star subgraph within the LC orbit of graph states that fail MMI. We formalized why generalized star graphs fail MMI, by establishing the necessary condition of a nontrivial intersection in Observation~\ref{obs:necessary-condition-generalized-graph}. Interestingly, even certain graph states that do not violate $MMI_{CIJ}$ contain the same ``forbidden'' four-star in their LC orbits (see Table~\ref{tab:MMI-Failing-EightQ-Graphs-Missing-Nontrivial-Int}). This observation motivates the forbidden-subgraph conjecture, stated in Conjecture \ref{conjecture:forbidden-subgraph}, which posits the \textit{necessity} of an induced four-star in the LC orbit of any graph whose associated graph state violates $MMI$. This conjecture realizes every MMI-violating stabilizer entropy vector as a graph state $\ket{\mathcal{G}}$, since containment of a $K_4$ subgraph is sufficient for a graph to be of type $\mathcal{G}$ (see Section \ref{section:Physical-Interpretation-Nontrivial-Intersection}).

If an induced four-star in the LC orbit is indeed necessary for any MMI failure, it becomes compelling to ask whether the empirical observation that MMI violations precede other holographic entropy inequality violations for $n\leq 6$ \cite{Keeler_2025} is coincidental or structural. If higher-party violations decompose into their minimal violators, refined versions of our methods, which would account for MMI failure in general, could be used to test whether MMI failure is necessary for \emph{all} higher-party violations. 

For $n\leq 8$ qubits, as shown in Figures~\ref{tab:MMI-Satisfying-EightQ-Graphs} and \ref{tab:MMI-Failing-EightQ-Graphs-NontrivialInt}, a nontrivial intersection of the column spaces of $CI$, $CJ$, and $CK$ (Observation~\ref{obs:necessary-condition-generalized-graph}) appears not only necessary for failure of $MMI_{CIJ}$ but also sufficient to guarantee failure of \emph{some} MMI instance. The challenge in proving the sufficiency of a nontrivial intersection for MMI failure lies in managing the edge cases across all MMI instances whose outcome can be determined only by this algebraic requirement.

The methods established in this work apply directly to graph states and, by local Clifford equivalence, to stabilizer states~\cite{VanDenNestDehaeneDeMoor2004}. The algebraic conditions derived herein could be extended from qubits to qudit systems, potentially offering insights into state preparation where unitary operators that prepare states and entropies are expressed using a directed graph representation~\cite{Salton:2016qpp,Munizzi:2025suf}. Our calculations in this work rely on an adjacency-matrix computation for entanglement entropy, as in Eq.\ \eqref{eq:EEadj}. If such a formula holds for directed graph representations, then our methods apply equivalently. Moreover, this framework provides a bridge to studying non-stabilizer resources based in entropy calculations~\cite{Pollack:2024bnj}, i.e. such as stabilizer Renyi entropy~\cite{Leone:2021rzd} or non-local quantum magic~\cite{Bao:2022mkc,Dallas:2025zxs}, relating their structure to entanglement-based diagnostics, e.g. capacity of entanglement~\cite{DeBoer:2018kvc,Cao:2024nrx,Khumalo:2025xfv} or operator entanglement~\cite{Andreadakis:2025mfw}, to capture deviations from stabilizerness.

Figures~\ref{tab:MMI-Satisfying-EightQ-Graphs}--\ref{tab:MMI-Failing-EightQ-Graphs-Missing-Nontrivial-Int} demonstrate that every graph state, for $n\leq 8$, is LC-equivalent to a graph that is a generalized star under some partition as indicated for each graph.
Nevertheless, the property of being $\mathcal{G}$ under a given partition is not invariant under LC transformations.
Moreover, at $n=9$ qubits and above, two graph states which are not LC-equivalent can share the same entropy vector \cite{Fon-Der-Flaass1996,KIM202454}. Accordingly, it remains an open question whether \emph{every} graph is LC-equivalent to a nontrivial generalized star. As we note in Section~\ref{section:generalizedStarGraphs}, this question reduces to proving that every graph state on $n\geq 4$ vertices is LC-equivalent to a graph with some $4$-vertex subset inducing an empty graph. If true, then generalized star graphs would enable an efficient canonical form for exploring the stabilizer entropy vector space.

The property of having a nontrivial intersection of the column spaces of $CI$, $CJ$, and $CK$, as in Observation~\ref{obs:necessary-condition-generalized-graph}, is not an LC invariant statement, as we demonstrate in the construction of Figures~\ref{tab:MMI-Satisfying-EightQ-Graphs}--\ref{tab:MMI-Failing-EightQ-Graphs-Missing-Nontrivial-Int}. A natural extension of this work is to formulate an LC-invariant version of the necessary condition for $MMI_{CIJ}$ failure, and ultimately to characterize the failure of all MMI instances.

A deeper understanding of MMI violation in graph states provides a tool for advancing both the theoretical foundations and practical implementation of future quantum networks. Many contemporary approaches to near-term quantum networks adopt a graph state model~\cite{Hahn_2019,Mannalath:2022tos,Sen:2023nmf,Azuma:2022nuy}, where network nodes correspond to graph vertices and entanglement is shared through physical channels, e.g. optical fibers. Once initialized, these  long-range entangled states enable a wide range distributed computing tasks, from entanglement sharing~\cite{Coffman2000, Hein_2004, HeinDurEisertBriegel2006,Negrin:2024tyj,Hughes:2025jwa} to quantum key distribution~\cite{Shor_2000,Gottesman:2002gg,Bennett:2014rmv,Mariani:2025qqu}. Characterizing how MMI evaluation is impacted by the underlying entanglement structure of the graph, and equivalently the associated network topology, informs how entanglement can be shared, localized, and redistributed among different network nodes. These insights inform the development of more robust and resource-efficient network protocols, providing direct technological applications for the constraints imposed by entropy inequalities.

\acknowledgments

The authors thank Adam Burchardt, Bartek Czech, William Kretschmer, and Monica Kang for helpful conversations. C. K. is supported by the U.S. Department of Energy under grant number DESC0019470 and by the Heising-Simons Foundation “Observational Signatures of Quantum Gravity” collaboration grants 2021-2818 and 2024-5305, and W.M. and J.F. have also been supported by 2021-2818. W.M. is supported by the Department of Energy (DOE) Office of Science (SC) Grant No DOE DE-FOA-0003432 and by Grant No GBMF12976 of the Gordon and Betty Moore Foundation. This work was performed in part at the Aspen Center for Physics, which is supported by National Science Foundation grant PHY-2210452.

\newpage

\appendix
\section{Additional Derivations and Proofs}

\subsection{Expressing $MMI_{CIJ}$ as the Dimension of Column Spaces of Adjacency Submatrices}
\label{A1:DimColEquations}

Recall that for any matrix $M$,
\begin{equation}
    \ra{M} = \dimension{\col{M}},
\end{equation}
where $\col{M}$ denotes the column space of $M$, which is the subspace spanned by the columns of $M$. Moreover, for any linear subspaces $W_1$ and $W_2$,
\begin{equation}
    \label{eq:DimSubspaceAddition}
    \dimension{W_1+W_2} = \dimension{W_1} + \dimension{W_2} - \dimension{W_1\cap W_2}.
\end{equation}
Thus, every entropy involved in $MMI_{CIJ}~\eqref{eq:MMI-Fail-G_M}$ can be expressed as follows:
\begin{equation}
\label{eq-block:MMI-G_M-Entropies}
    \begin{aligned}
      S_C &= \ra{\begin{bmatrix} C I & C J & C K \end{bmatrix}} \\
          &= \ra{\begin{bmatrix} C I & C J \end{bmatrix}} + \ra{C K}
             \\ & \quad- \dimension{(\col{C I} + \col{C J}) \cap \col{C K}} \\
          &= \ra{C I} + \ra{C J} + \ra{C K} 
             \\ & \quad- \dimension{\col{C I} \cap \col{C J}}
             - \dimension{(\col{C I} + \col{C J}) \cap \col{C K}} ,\\[0.9em]
      S_I   &= \ra{C I}, \\
      S_J   &= \ra{C J},  \\ 
      S_{CIJ} &= \ra{C K}, \\
      S_{CI} &= \ra{C J} + \ra{C K} - \dimension{\col{C J} \cap \col{C K}}, \\
      S_{CJ} &= \ra{C I} + \ra{C K} - \dimension{\col{C I} \cap \col{C K}}, \\
      S_{IJ} &= \ra{C I} + \ra{C J} - \dimension{\col{C I} \cap \col{C J}}
    \end{aligned}
\end{equation}
The entropies above allow us to express $MMI_{CIJ}$ as:
\begin{equation}
\label{eq:appendix-MMI-Fail-ColSpaces-G_M}
\begin{aligned}
&\dimension{\col{C  I} \cap \col{C  K}}+ \dimension{\col{C  J} \cap \col{C  K}}\\[1ex]
\leq \; & \dimension{\left( \col{C  I} + \col{C  J}\right) \cap \col{C  K}}.
\end{aligned}
\end{equation}

\subsection{A Classification of $\ket{\mathcal{G}}$ According to MMI Evaluation Based on the Properties of the Block Elements of $\Gamma^{\mathcal{L}}$}
\label{A2:sufficient-Condition-General-Graph}

Let $W_1, W_2$, and $W_3$ denote arbitrary subspaces of a vector space $V$. From the dimension formula for the sum of two subspaces \eqref{eq:DimSubspaceAddition}, it follows:
\begin{equation}\label{subspaceDimension}
    \resizebox{0.97\linewidth}{!}{$
    \dimension{W_1 \cap W_3 + W_2 \cap W_3}
    = \dimension{W_1 \cap W_3 }
    + \dimension{W_2 \cap W_3 }
    - \dimension{W_1\cap W_2 \cap W_3}.
    $}
\end{equation}

Additionally, the subspace inclusion
\begin{equation}\label{subspaceInclusion}
   W_1 \cap W_3 + W_2 \cap W_3 \subseteq (W_1 + W_2) \cap W_3,
\end{equation}
implies the following general inequality:
\begin{equation}
    \label{eq:Appendix-Gen-Intersection-Over-Sum-Inequality}
    \dimension{W_1 \cap W_3 } + \dimension{W_2 \cap W_3 } - \dimension{W_1\cap W_2 \cap W_3} \leq \dimension{(W_1 + W_2) \cap W_3}.
\end{equation}
Equality in \eqref{eq:Appendix-Gen-Intersection-Over-Sum-Inequality} holds precisely under the condition
\begin{equation}
    W_1 \cap W_3 + W_2 \cap W_3 = (W_1 + W_2)\cap W_3,
\end{equation}
which is known as distributivity. If distributivity holds, it remains valid under any permutation of the subspaces. As established in \cite[Section~1.7]{polishchuk2005quadratic}, a collection of subspaces $W_1, \cdots, W_N$ of $V$ is distributive if and only if there exists a common basis $\mathcal{B}$ of $V$ such that
\begin{equation}
\label{eq:DistributiveSubspacesCond}
    W_i = \spanv{B_i}, \quad B_i \subseteq \mathcal{B}, \text{ for } i\in \{1,\cdots, N \},
\end{equation}
where $\spanv{B_i}$ denotes the linear span of $B_i$. 

We now examine exhaustively all evaluations of the inequality \eqref{eq:Appendix-Gen-Intersection-Over-Sum-Inequality} according to the algebraic structures of the subspaces $W_1, W_2,$ and $W_3$.

\vspace{1em}

\textbf{Case 1.} \emph{$W_1, W_2, W_3$ are non-distributive} with $\dimension{W_1 \cap W_2 \cap W_3}=0$: In this case, Eq. \eqref{eq:Gen-Intersection-Over-Sum-Inequality} simplifies to
\begin{equation}
\label{eq:satisfiesSubspaces}
\dimension{W_1 \cap W_3} + \dimension{W_2 \cap W_3} < \dimension{(W_1 + W_2) \cap W_3}.
\end{equation}

\textbf{Case 2.} $W_1, W_2, W_3$ are \emph{distributive} with $\dimension{W_1 \cap W_2 \cap W_3}=0$: Then, equality holds exactly as:
\begin{equation}
\label{eq:saturatesSubspaces}
\dimension{W_1 \cap W_3} + \dimension{W_2 \cap W_3} = \dimension{(W_1 + W_2) \cap W_3}.
\end{equation}

\textbf{Case 3.} $W_1, W_2, W_3$ are \emph{distributive} with $\dimension{W_1 \cap W_2 \cap W_3} > 0$: Here, dimensional subadditivity ensures the strict inequality
\begin{align}
\dimension{W_1 \cap W_3} + \dimension{W_2 \cap W_3} 
&> \dimension{W_1 \cap W_3+W_2 \cap W_3}.
\end{align}
Since the subspaces are distributive $W_1 \cap W_3+W_2 \cap W_3 = (W_1+W_2)\cap W_3$, then
\begin{align}
\dimension{W_1 \cap W_3} + \dimension{W_2 \cap W_3} 
&> \dimension{(W_1+W_2)\cap W_3}.
\label{eq:failsSubspaces}
\end{align}

\textbf{Case 4.} $W_1, W_2, W_3$ are \emph{non-distributive} with $\dimension{W_1 \cap W_2 \cap W_3} > 0$: Under these assumptions, we examine the inequality
\begin{equation}
\label{eq:subspaceMMI}
\dimension{W_1 \cap W_3 } + \dimension{W_2 \cap W_3 } \leq \dimension{(W_1 + W_2) \cap W_3}.
\end{equation}
All previously considered conditions, namely, strict satisfaction \eqref{eq:satisfiesSubspaces}, saturation \eqref{eq:saturatesSubspaces}, and failure \eqref{eq:failsSubspaces}, remain possible for inequality \eqref{eq:subspaceMMI}. To provide explicit constructions that realize each possibility, we set $V=\mathbb{Z}_2^6$ as the ambient vector space containing subspaces $W_1$, $W_2$, and $W_3$. Let $\mathcal{B} = \{ e_1,e_2,e_3,e_4,e_5,e_6\}$ be a fixed basis of $\mathbb{Z}_2^6$. Table \ref{tab:Case4Examples} provides an example of a Case 4 construction for $W_1$, $W_2$, and $W_3$ that illustrates these distinct possibilities concerning inequality \eqref{eq:subspaceMMI}.

\begin{table}[H]
\centering
\footnotesize
\setlength{\tabcolsep}{5pt}
\renewcommand\arraystretch{1.15}

\begin{subtable}[t]{\textwidth}
  \centering
  \caption{Subspace Construction}
  \begin{tabularx}{\textwidth}{YYYY}
  \toprule
  \textbf{Evaluation} & $W_1$ & $W_2$ & $W_3$\\
  \midrule
  ST & $\spanv{e_1,e_2}$ & $\spanv{e_1,e_3}$ &
        $\spanv{e_1,e_2+e_3}$\\
  F & $\spanv{e_1,e_3,e_4}$ & $\spanv{e_2,e_3,e_4}$ &
        $\spanv{e_1+e_2,e_3,e_4}$\\
  S & $\spanv{e_1,e_2,e_3,e_5}$ & $\spanv{e_1,e_2,e_4,e_6}$ &
        $\spanv{e_1+e_2,e_3+e_4,e_5+e_6}$\\
  \bottomrule
  \end{tabularx}
\end{subtable}

\vspace{6pt}

\begin{subtable}[t]{\textwidth}
  \centering
  \caption{Intersections}
  \begin{tabularx}{\textwidth}{YYYYc}
  \toprule
  \textbf{Evaluation} & $W_1\!\cap\!W_3$ & $W_2\!\cap\!W_3$ &
  $(W_1{+}W_2)\!\cap\!W_3$ \\
  \midrule
  ST & $\spanv{e_1}$ & $\spanv{e_1}$ &
        $\spanv{e_1,e_2+e_3}$\\
  F & $\spanv{e_3,e_4}$ & $\spanv{e_3,e_4}$ &
        $\spanv{e_1+e_2,e_3,e_4}$ \\
  S & $\spanv{e_1+e_2}$ & $\spanv{e_1+e_2}$ &
        $\spanv{e_1+e_2,e_3+e_4,e_5+e_6}$ \\
  \bottomrule
  \end{tabularx}
\end{subtable}

\caption{Example constructions of subspaces $W_1,W_2, W_3 \subseteq \mathbb{Z}_2^6$ are provided, illustrating saturation (ST), failure (F), and satisfaction (S) with respect to inequality \eqref{eq:subspaceMMI}. The vectors $\{e_i\}_{i=1}^{i=6}$ form a fixed basis of the ambient space $\mathbb{Z}_2^6$. The notation $\langle e_1,\cdots, e_k\rangle$ denotes the linear span of $\{e_1,\cdots, e_k\}.$}
\label{tab:Case4Examples}
\end{table}

Specializing the preceding discussion by setting $W_1 = W_I$, $W_2=W_J$, and $W_3=W_K$ as subspaces of $V=Z_{2}^{|C|}$, one directly obtains the MMI evaluation corresponding to each subclass of $\mathcal{G}$ detailed in Table \ref{tab:MMIcases}. Table \ref{tab:MMI-characterization-examples-1-3} presents example graphs corresponding to the first three cases in Table \ref{tab:MMIcases}, while Table \ref{tab:case-4-graph-examples} presents graphs that realize each Case 4 evaluation example in Table~\ref{tab:Case4Examples} (a).

\subsection{Deducing the Claims in Section \ref{subsection:physical-interpretation}}
\label{A3:deducing-physical-interp-claims}

\begin{proposition}
    Operations on $\ket{\mathcal{G}}$ that act invariantly on $\col{CI}, \col{CJ}$ and $\col{CK}$:
    \begin{enumerate}
        \item Leave every entropy of $\ket{\mathcal{G}}$ in $MMI_{CIJ}$ invariant.
        \item Preserve the $MMI_{CIJ}$ evaluation and classification of $\ket{\mathcal{G}}$.
        \item Generally alter the entropy vector of $\ket{\mathcal{G}}$.
    \end{enumerate}
\end{proposition}
\begin{proof}
    The first statement of the proposition follows immediately from the equation block \eqref{eq-block:MMI-G_M-Entropies}. This block contains every entropy involved in $MMI_{CIJ}~\eqref{eq:MMI-Fail-G_M}$, and shows that they are expressed only in terms of $\col{CI}$, $\col{CJ}$, and $\col{CK}$. Consequently, operations that preserve $\col{CI}$, $\col{CJ}$, and $\col{CK}$ simultaneously leave every entropy involved in $MMI_{CIJ}$ invariant. 

    The second statement follows immediately from the first, as if every entropy in $MMI_{CIJ}$ is preserved, its evaluation will remain invariant. Furthermore, $MMI_{CIJ}$ classification per Table \ref{tab:MMIcases} is exclusively made in terms of $\col{CI}$, $\col{CJ}$, and $\col{CK}$, meaning that their algebraic properties such as distributivity and intersections will be unchanged if these subspaces are also unchanged.

    The third statement is a consequence of column operations on $CI$, $CJ$, and $CK$, which preserve their column spaces, being realized on $\ket{\mathcal{G}}$ as a sequence of $CZ$ operations, which are non-local. In a given partition of $\ket{\mathcal{G}}$ specified by $\{C,I,J\}$, we can apply focalized transformations that correspond to column operations on $\ket{\mathcal{G}}$ that preserve, $\col{CI}$, $\col{CJ}$, and $\col{CK}$, separately, since $I, J$ and $K$ are pairwise disjoint subsystems. Nevertheless, these transformations could alter the column space of the complete adjacency matrix of $\mathcal{G}$, possibly changing the entropy vector of $\ket{\mathcal{G}}$.
\end{proof}

\begin{proposition}
    Every unique failure of $MMI_{CIJ}$ for states of type $\ket{\mathcal{G}}$ is realized by a graph state whose representation only admits star graphs.
\end{proposition}
\begin{proof}
    From Eq.~\eqref{eq-block:MMI-G_M-Entropies}, all entropies needed to evaluate $MMI_{CIJ}$\eqref{eq:MMI-Fail-G_M} are independent of the self adjacency blocks $CC, II, \dots$ of each element in the partition $\{ C, I, J, K\}$. Thus, we may fix these regions having no internal adjacencies, and still be able to obtain any portion of any graph state entropy vector that produces any outcome for $MMI_{CIJ}$ instances. Under such an assumption, the construction of $\mathcal{G}$ (conditions \ref{condition: genDisjointness} and \ref{condition: genDisconnection}) ensures we obtain a graph composed exclusively of the disjoint union of $n$-stars, $K_n$, which completes the proof. 
\end{proof}

\subsection{Graph Representatives for Every Stabilizer State Entropy Vector at $n=8$ Qubits, up to Symmetry, and Their MMI Evaluations}

This section presents a graph representative for each of the $182$ eight-qubit stabilizer-state entropy vectors \cite{Danielsen_2006}, unique up to qubit exchange. We chose each representative to be $\mathcal{G}$ under the specified partition $\{C, I,J,K\}$. We classify these graphs into three categories, according to the MMI behavior of their associated graph states:

\begin{enumerate}
\item Those that satisfy all MMI instances (Table \ref{tab:MMI-Satisfying-EightQ-Graphs});
\item Those that fail $MMI_{CIJ}$ and are  $\mathcal{G}$, necessarily with a nontrivial intersection for the partition $\{C,I,J,K\}$ (Table \ref{tab:MMI-Failing-EightQ-Graphs-NontrivialInt});
\item Those that fail MMI but do not exhibit a nontrivial intersection for any choice of $\{C,I,J\}$ under which the graph is $\mathcal{G}$ (Table \ref{tab:MMI-Failing-EightQ-Graphs-Missing-Nontrivial-Int}).
\end{enumerate}

No eight-qubit graph state that satisfies MMI possesses a nontrivial intersection for any $\{C,I,J\}$ under which it is nontrivially $\mathcal{G}$. Each graph representative was selected as the minimal-edge representative of its local Clifford ($LC$) orbit within its category. The full dataset, along with a detailed tutorial explaining how to reproduce the results and use the accompanying \textit{Mathematica} notebook, is available in the GitHub repository \cite{fuentes_munizzi2025mmi}.

\begin{center}
\\[1em]
  \scriptsize \textit{Clicking each graph’s label redirects to a website that provides easy access to the entropy vector of its state.}
\end{center}
%
 \\[1em]
  \scriptsize \textit{Clicking each graph’s label redirects to a website that provides easy access to the entropy vector of its state.}
\end{center}

%
 \\[1em]
    \scriptsize \textit{Clicking each graph’s label redirects to a website that provides easy access to the entropy vector of its state.}
\end{center}
%

\section{MMI Failure in Stabilizer States}
\label{B1:MMI-Failure-Among-Stabilizer-States}

For stabilizer state entropy vectors on $n \leq 5$ qubits, we also have a complete description for which entropy vectors saturate, satisfy, and fail MMI. In the following subsections we discuss the various instances of MMI that arise at each qubit number, and how each relates to the corresponding set of stabilizer entropy vectors. 

\underline{Note:} Instances of MMI which include all single-qubit subregions in the inequality are trivially satisfied by stabilizer states. There is therefore no need to check such instances.
\subsubsection{Three Qubits}

For $n=3$ qubits there are $5$ total stabilizer state entropy vectors, which reduce to $3$ entropy vectors up to symmetry by qubit exchange. Table \ref{tab:StabilizerEntropyThreeQubits} gives a representative from each symmetry class, along with the number and percentage of overall states which have each entropy vector.
\begin{table}[h]
\centering
\begin{tabular}{c||c|c|c||c}
& A & B & C & \text{Num. States}\\
\hline
$s_{1,1,1}$ &  $0$ & $0$ & $0$  & $216 \quad (20\%)$\\
$s_{2,1}$ &  $1$ & $1$ & $0$  & $432 \quad (40\%)$\\
$s_3$ &  $1$ & $1$ & $1$  & $432 \quad (40\%)$\\
\end{tabular}
\caption{Stabilizer entropy vectors on $3$ qubits, all trivially saturate MMI.}
\label{tab:StabilizerEntropyThreeQubits}
\end{table}

Vectors $s_{1,1,1}$ and $s_{2,1}$ in Table \ref{tab:StabilizerEntropyThreeQubits} correspond to product states of the form $A \otimes B \otimes C$ and $AB \otimes C$, respectively. Vector $s_3$ represents GHZ-type entanglement among $3$ qubits. All vectors in Table \ref{tab:StabilizerEntropyThreeQubits} trivially saturate or satisfy the single $3$-qubit instance of MMI.

\subsubsection{Four Qubits}

At $4$ qubits, there are $18$ stabilizer state entropy vectors which arrange into $6$ equivalence classes up to qubit-exchange symmetry. Table \ref{tab:StabilizerEntropyFourQubits} gives a single representative from each equivalence class, as well as the number of overall states which admit an entropy vector in each equivalence class.\\
\begin{table}[h]
\centering
\begin{tabular}{c||c|c|c|c|c|c|c||c}
& A & B & C & D & AB & AC & AD & \text{Num. States}\\
\hline
$s_{1,1,1,1}$ &  $0$ & $0$ & $0$ & $0$ & $0$ & $0$ & $0$ &  1296 \quad (3.53\%)\\
$s_{2,1,1}$ &  $1$ & $1$ & $0$ & $0$ & $0$ & $1$ & $1$ &  5184 \quad (14.11\%)\\
$s_{3,1}$ &  $1$ & $1$ & $1$ & $0$ & $1$ & $1$ & $1$ &  10368 \quad (28.24\%)\\
$s_{2,2}$ &  $1$ & $1$ & $1$ & $1$ & $0$ & $2$ & $2$ &  1728 \quad (4.71\%)\\
\hline
$s_{4_1}$ &  $1$ & $1$ & $1$ & $1$ & $1$ & $1$ & $1$ &  2592 \quad (7.1\%)\\
$s_{4_2}$ &  $1$ & $1$ & $1$ & $1$ & $1$ & $2$ & $2$ &  15552 \quad (42.35\%)\\
\end{tabular}
\caption{Representatives of the $6$ equivalence classes, up to qubit exchange symmetry, for all $4$-qubit stabilizer state entropy vectors.}
\label{tab:StabilizerEntropyFourQubits}
\end{table}

The first $4$ vectors in Table \ref{tab:StabilizerEntropyFourQubits} are products of lower-qubit entropy vectors%
\footnote{These vectors lie at the base of the MMI cone in some sense, since several of their components are $0$.}%
, and thereby will trivially saturate all instances of MMI. Only $s_{4_1}$ and $s_{4_2}$ represent non-trivial entanglement among $4$ qubits.

At $4$ qubits there are $10$ instances of MMI. We express each instance in the table below, where the sum of the left-hand columns must be greater than or equal to the sum of the right

\begin{center}
    $
\begin{array}{c|||ccc||cccc|}
    \hline
    \multicolumn{7}{c}{\qquad IJ + IK + JK \geq I + J + K + IJK} \\
    \hline
      \hline
1& \{\text{A},\text{B}\} & \{\text{A},\text{C}\} & \{\text{B},\text{C}\} & \{\text{A}\} & \{\text{B}\} & \{\text{C}\} & \{\text{A},\text{B},\text{C}\} \\
      \hline
2& \{\text{A},\text{B}\} & \{\text{A},\text{D}\} & \{\text{B},\text{D}\} & \{\text{A}\} & \{\text{B}\} & \{\text{D}\} & \{\text{A},\text{B},\text{D}\} \\
      \hline
3& \{\text{A},\text{C}\} & \{\text{A},\text{D}\} & \{\text{C},\text{D}\} & \{\text{A}\} & \{\text{C}\} & \{\text{D}\} & \{\text{A},\text{C},\text{D}\} \\
      \hline
4& \{\text{B},\text{C}\} & \{\text{B},\text{D}\} & \{\text{C},\text{D}\} & \{\text{B}\} & \{\text{C}\} & \{\text{D}\} & \{\text{B},\text{C},\text{D}\} \\
      \hline
\end{array}$\\
\end{center}

Table \ref{tab:FourQMMITable} below gives the data for all $4$-qubit entropy vectors, stating whether each saturates, satisfies or fails MMI. Each column of Table \ref{tab:FourQMMITable} represents exactly one non-trivial instance of MMI, read from left to right (following top to bottom in the MMI instances table).

\begin{table}[h]
\centering
\begin{tabular}{|c||c|c|c|c|c|}
\hline
   &  1 & 2 & 3 & 4\\
     \hline
      \hline
  $s_{1,1,1,1}$ &  \text{S} & \text{S} & \text{S} &  \text{S}\\
     \hline
  $s_{2,1,1}$ &   \text{S} & \text{S} & \text{S} & \text{S}\\
     \hline
  $s_{3,1}$ &   \text{S} & \text{S} &  \text{S} & \text{S}\\
     \hline
   $s_{2,2}$ &  \text{S} & \text{S} &  \text{S} & \text{S}\\
     \hline
   $s_{4_1}$ &  \text{F} & \text{F} &  \text{F} & \text{F}\\
     \hline
   $s_{4_2}$ &  \text{Y} & \text{Y} &  \text{Y} &  \text{Y}\\
     \hline
\end{tabular}
    \caption{$S = \text{Saturates},\,Y = \text{Satisfies},\,F = \text{Fails}$.}
    \label{tab:FourQMMITable}
\end{table}

We observe that vector $s_{4_1}$ fails $4$ instances of MMI, and so lives outside the MMI cone. Vector $s_{4_2}$ lies inside the MMI cone, but does not saturate any non-trivial instance of MMI.

\subsubsection{Five Qubits}

At $5$ qubits, the $93$ stabilizer state entropy vectors reduce to the following $11$ representatives
\begin{table}[h]
\centering
\setlength{\tabcolsep}{3pt}      
\renewcommand{\arraystretch}{1.05}
\small                      
\begin{tabular}{c||c||c|c|c|c|c|c|c|c|c|c|c|c|c|c|c||c}
& & A & & & & E & AB &  &  &  &  &  &  &  &  & DE & \text{Num. States}\\
\hline
$s_{1,1,1,1,1}$ & $1$ & $0$ & $0$ & $0$ & $0$ & $0$ & $0$ & $0$ & $0$ & $0$ & $0$ & $0$ & $0$ & $0$ & $0$ & $0$ & $15552 \,(0.64\%)$\\
$s_{2,1,1,1}$ & $10$ & $1$ & $1$ & $0$ & $0$ & $0$ & $0$ & $1$ & $1$ & $1$ & $1$ & $1$ & $1$ & $0$ & $0$ & $0$ & $160704 \,(6.63\%)$\\
$s_{3,1,1}$ & $10$ &  $1$ & $1$ & $1$ & $0$ & $0$ & $1$ & $1$ & $1$ & $1$ & $1$ & $1$ & $1$ & $1$ & $1$ & $0$ & $51840 \,(2.14\%)$\\
$s_{2,2,1}$ & $15$ & $1$ & $1$ & $1$ & $1$ & $0$ & $0$ & $2$ & $2$ & $1$ & $2$ & $2$ & $1$ & $0$ & $1$ & $1$ & $400896 \,(16.54\%)$\\
$s_{3,2}$ & $10$ &  $1$ & $1$ & $1$ & $1$ & $1$ & $1$ & $1$ & $2$ & $2$ & $1$ & $2$ & $2$ & $2$ & $2$ & $0$ & $556416 \,(22.96\%)$\\
$s_{4_1,1}$ & $5$ & $1$ & $1$ & $1$ & $1$ & $0$ & $1$ & $1$ & $1$ & $1$ & $1$ & $1$ & $1$ & $1$ & $1$ & $1$ & $63072 \,(2.60\%)$\\
$s_{4_2,1}$ & $15$ & $1$ & $1$ & $1$ & $1$ & $0$ & $1$ & $2$ & $2$ & $1$ & $2$ & $2$ & $1$ & $1$ & $1$ & $1$ & $273024 \,(11.27\%)$\\
\hline
$s_{5_1}$ & $1$ & $1$ & $1$ & $1$ & $1$ & $1$ & $1$ & $1$ & $1$ & $1$ & $1$ & $1$ & $1$ & $1$ & $1$ & $1$ & $15552 \,(0.64\%)$\\
$s_{5_2}$ & $10$ & $1$ & $1$ & $1$ & $1$ & $1$ & $1$ & $2$ & $2$ & $2$ & $2$ & $2$ & $2$ & $1$ & $1$ & $1$ & $362880 \,(14.97\%)$\\
$s_{5_3}$ & $15$ &  $1$ & $1$ & $1$ & $1$ & $1$ & $1$ & $2$ & $2$ & $2$ & $2$ & $2$ & $2$ & $1$ & $2$ & $2$ & $513216 \,(21.18\%)$\\
$s_{5_4}$ & $1$ & $1$ & $1$ & $1$ & $1$ & $1$ & $2$ & $2$ & $2$ & $2$ & $2$ & $2$ & $2$ & $2$ & $2$ & $2$ & $10368 \,(0.43\%)$\\
\end{tabular}
\caption{Representatives of the $11$ equivalence classes, up to qubit exchange symmetry, for all $5$-qubit stabilizer state entropy vectors.}
\label{tab:StabilizerEntropyFiveQubits}
\end{table}

The first $7$ entropy vectors of Table \ref{tab:StabilizerEntropyFiveQubits} are lifts of lower-qubit entanglement structures. The first $5$ entropy vectors in particular are lifts of $3$-qubit entanglement, and will trivially satisfy all $5$-qubit instances of MMI.

Vectors $s_{4_1, 1}$ and $s_{4_2, 1}$ are lifts of entropy vectors $s_{4_1}$ and $s_{4_2}$ from Table \ref{tab:StabilizerEntropyFourQubits}. As with their $4$-qubit analogs, the vectors $s_{4_1, 1}$ and $s_{4_2, 1}$ fail and satisfy, respectively, $5$-qubit lifts of $4$-qubit MMI.

The final four entropy vectors of Table \ref{tab:StabilizerEntropyFiveQubits} are genuine instances of $5$-qubit entanglement. We discuss their relation to MMI below. At $5$ qubits there are $40$ unique instances of MMI, they are\\

$$
\left(
\begin{array}{ccc||cccc}
 \{\text{A},\text{B}\} & \{\text{A},\text{O}\} & \{\text{B},\text{O}\} & \{\text{A}\} & \{\text{B}\} & \{\text{O}\} & \{\text{A},\text{B},\text{O}\} \\
 \{\text{A},\text{B}\} & \{\text{A},\text{D}\} & \{\text{B},\text{D}\} & \{\text{A}\} & \{\text{B}\} & \{\text{D}\} & \{\text{A},\text{B},\text{D}\} \\
 \{\text{A},\text{B}\} & \{\text{A},\text{C}\} & \{\text{B},\text{C}\} & \{\text{A}\} & \{\text{B}\} & \{\text{C}\} & \{\text{A},\text{B},\text{C}\} \\
 \{\text{A},\text{C}\} & \{\text{A},\text{O}\} & \{\text{C},\text{O}\} & \{\text{A}\} & \{\text{C}\} & \{\text{O}\} & \{\text{A},\text{C},\text{O}\} \\
 \{\text{A},\text{C}\} & \{\text{A},\text{D}\} & \{\text{C},\text{D}\} & \{\text{A}\} & \{\text{C}\} & \{\text{D}\} & \{\text{A},\text{C},\text{D}\} \\
 \{\text{A},\text{D}\} & \{\text{A},\text{O}\} & \{\text{D},\text{O}\} & \{\text{A}\} & \{\text{D}\} & \{\text{O}\} & \{\text{A},\text{D},\text{O}\} \\
 \{\text{B},\text{C}\} & \{\text{B},\text{O}\} & \{\text{C},\text{O}\} & \{\text{B}\} & \{\text{C}\} & \{\text{O}\} & \{\text{B},\text{C},\text{O}\} \\
 \{\text{B},\text{C}\} & \{\text{B},\text{D}\} & \{\text{C},\text{D}\} & \{\text{B}\} & \{\text{C}\} & \{\text{D}\} & \{\text{B},\text{C},\text{D}\} \\
 \{\text{B},\text{D}\} & \{\text{B},\text{O}\} & \{\text{D},\text{O}\} & \{\text{B}\} & \{\text{D}\} & \{\text{O}\} & \{\text{B},\text{D},\text{O}\} \\
 \{\text{C},\text{D}\} & \{\text{C},\text{O}\} & \{\text{D},\text{O}\} & \{\text{C}\} & \{\text{D}\} & \{\text{O}\} & \{\text{C},\text{D},\text{O}\} \\
\end{array}
\right)
$$

$$
\left(
\begin{array}{ccc||cccc}
 \{\text{A},\text{O}\} & \{\text{A},\text{B},\text{D}\} & \{\text{B},\text{D},\text{O}\} & \{\text{A}\} & \{\text{O}\} & \{\text{B},\text{D}\} & \{\text{A},\text{B},\text{D},\text{O}\} \\
 \{\text{A},\text{D}\} & \{\text{A},\text{B},\text{O}\} & \{\text{B},\text{D},\text{O}\} & \{\text{A}\} & \{\text{D}\} & \{\text{B},\text{O}\} & \{\text{A},\text{B},\text{D},\text{O}\} \\
 \{\text{A},\text{O}\} & \{\text{A},\text{B},\text{C}\} & \{\text{B},\text{C},\text{O}\} & \{\text{A}\} & \{\text{O}\} & \{\text{B},\text{C}\} & \{\text{A},\text{B},\text{C},\text{O}\} \\
 \{\text{A},\text{C}\} & \{\text{A},\text{B},\text{O}\} & \{\text{B},\text{C},\text{O}\} & \{\text{A}\} & \{\text{C}\} & \{\text{B},\text{O}\} & \{\text{A},\text{B},\text{C},\text{O}\} \\
 \{\text{A},\text{D}\} & \{\text{A},\text{B},\text{C}\} & \{\text{B},\text{C},\text{D}\} & \{\text{A}\} & \{\text{D}\} & \{\text{B},\text{C}\} & \{\text{A},\text{B},\text{C},\text{D}\} \\
 \{\text{A},\text{C}\} & \{\text{A},\text{B},\text{D}\} & \{\text{B},\text{C},\text{D}\} & \{\text{A}\} & \{\text{C}\} & \{\text{B},\text{D}\} & \{\text{A},\text{B},\text{C},\text{D}\} \\
 \{\text{A},\text{O}\} & \{\text{A},\text{C},\text{D}\} & \{\text{C},\text{D},\text{O}\} & \{\text{A}\} & \{\text{O}\} & \{\text{C},\text{D}\} & \{\text{A},\text{C},\text{D},\text{O}\} \\
 \{\text{A},\text{D}\} & \{\text{A},\text{C},\text{O}\} & \{\text{C},\text{D},\text{O}\} & \{\text{A}\} & \{\text{D}\} & \{\text{C},\text{O}\} & \{\text{A},\text{C},\text{D},\text{O}\} \\
 \{\text{A},\text{B}\} & \{\text{A},\text{C},\text{O}\} & \{\text{B},\text{C},\text{O}\} & \{\text{A}\} & \{\text{B}\} & \{\text{C},\text{O}\} & \{\text{A},\text{B},\text{C},\text{O}\} \\
 \{\text{A},\text{B}\} & \{\text{A},\text{C},\text{D}\} & \{\text{B},\text{C},\text{D}\} & \{\text{A}\} & \{\text{B}\} & \{\text{C},\text{D}\} & \{\text{A},\text{B},\text{C},\text{D}\} \\
 \{\text{A},\text{C}\} & \{\text{A},\text{D},\text{O}\} & \{\text{C},\text{D},\text{O}\} & \{\text{A}\} & \{\text{C}\} & \{\text{D},\text{O}\} & \{\text{A},\text{C},\text{D},\text{O}\} \\
 \{\text{A},\text{B}\} & \{\text{A},\text{D},\text{O}\} & \{\text{B},\text{D},\text{O}\} & \{\text{A}\} & \{\text{B}\} & \{\text{D},\text{O}\} & \{\text{A},\text{B},\text{D},\text{O}\} \\
 \{\text{B},\text{O}\} & \{\text{A},\text{B},\text{D}\} & \{\text{A},\text{D},\text{O}\} & \{\text{B}\} & \{\text{O}\} & \{\text{A},\text{D}\} & \{\text{A},\text{B},\text{D},\text{O}\} \\
 \{\text{B},\text{D}\} & \{\text{A},\text{B},\text{O}\} & \{\text{A},\text{D},\text{O}\} & \{\text{B}\} & \{\text{D}\} & \{\text{A},\text{O}\} & \{\text{A},\text{B},\text{D},\text{O}\} \\
 \{\text{B},\text{O}\} & \{\text{A},\text{B},\text{C}\} & \{\text{A},\text{C},\text{O}\} & \{\text{B}\} & \{\text{O}\} & \{\text{A},\text{C}\} & \{\text{A},\text{B},\text{C},\text{O}\} \\
 \{\text{B},\text{C}\} & \{\text{A},\text{B},\text{O}\} & \{\text{A},\text{C},\text{O}\} & \{\text{B}\} & \{\text{C}\} & \{\text{A},\text{O}\} & \{\text{A},\text{B},\text{C},\text{O}\} \\
 \{\text{B},\text{D}\} & \{\text{A},\text{B},\text{C}\} & \{\text{A},\text{C},\text{D}\} & \{\text{B}\} & \{\text{D}\} & \{\text{A},\text{C}\} & \{\text{A},\text{B},\text{C},\text{D}\} \\
 \{\text{B},\text{C}\} & \{\text{A},\text{B},\text{D}\} & \{\text{A},\text{C},\text{D}\} & \{\text{B}\} & \{\text{C}\} & \{\text{A},\text{D}\} & \{\text{A},\text{B},\text{C},\text{D}\} \\
 \{\text{B},\text{O}\} & \{\text{B},\text{C},\text{D}\} & \{\text{C},\text{D},\text{O}\} & \{\text{B}\} & \{\text{O}\} & \{\text{C},\text{D}\} & \{\text{B},\text{C},\text{D},\text{O}\} \\
 \{\text{B},\text{D}\} & \{\text{B},\text{C},\text{O}\} & \{\text{C},\text{D},\text{O}\} & \{\text{B}\} & \{\text{D}\} & \{\text{C},\text{O}\} & \{\text{B},\text{C},\text{D},\text{O}\} \\
 \{\text{B},\text{C}\} & \{\text{B},\text{D},\text{O}\} & \{\text{C},\text{D},\text{O}\} & \{\text{B}\} & \{\text{C}\} & \{\text{D},\text{O}\} & \{\text{B},\text{C},\text{D},\text{O}\} \\
 \{\text{C},\text{O}\} & \{\text{A},\text{C},\text{D}\} & \{\text{A},\text{D},\text{O}\} & \{\text{C}\} & \{\text{O}\} & \{\text{A},\text{D}\} & \{\text{A},\text{C},\text{D},\text{O}\} \\
 \{\text{C},\text{D}\} & \{\text{A},\text{C},\text{O}\} & \{\text{A},\text{D},\text{O}\} & \{\text{C}\} & \{\text{D}\} & \{\text{A},\text{O}\} & \{\text{A},\text{C},\text{D},\text{O}\} \\
 \{\text{C},\text{O}\} & \{\text{A},\text{B},\text{C}\} & \{\text{A},\text{B},\text{O}\} & \{\text{C}\} & \{\text{O}\} & \{\text{A},\text{B}\} & \{\text{A},\text{B},\text{C},\text{O}\} \\
 \{\text{C},\text{D}\} & \{\text{A},\text{B},\text{C}\} & \{\text{A},\text{B},\text{D}\} & \{\text{C}\} & \{\text{D}\} & \{\text{A},\text{B}\} & \{\text{A},\text{B},\text{C},\text{D}\} \\
 \{\text{C},\text{O}\} & \{\text{B},\text{C},\text{D}\} & \{\text{B},\text{D},\text{O}\} & \{\text{C}\} & \{\text{O}\} & \{\text{B},\text{D}\} & \{\text{B},\text{C},\text{D},\text{O}\} \\
 \{\text{C},\text{D}\} & \{\text{B},\text{C},\text{O}\} & \{\text{B},\text{D},\text{O}\} & \{\text{C}\} & \{\text{D}\} & \{\text{B},\text{O}\} & \{\text{B},\text{C},\text{D},\text{O}\} \\
 \{\text{D},\text{O}\} & \{\text{A},\text{C},\text{D}\} & \{\text{A},\text{C},\text{O}\} & \{\text{D}\} & \{\text{O}\} & \{\text{A},\text{C}\} & \{\text{A},\text{C},\text{D},\text{O}\} \\
 \{\text{D},\text{O}\} & \{\text{A},\text{B},\text{D}\} & \{\text{A},\text{B},\text{O}\} & \{\text{D}\} & \{\text{O}\} & \{\text{A},\text{B}\} & \{\text{A},\text{B},\text{D},\text{O}\} \\
 \{\text{D},\text{O}\} & \{\text{B},\text{C},\text{D}\} & \{\text{B},\text{C},\text{O}\} & \{\text{D}\} & \{\text{O}\} & \{\text{B},\text{C}\} & \{\text{B},\text{C},\text{D},\text{O}\} \\
\end{array}
\right)
$$


Three-region MMI\\
$
\begin{array}{c||cccccccccc}
\hline
s_{1,1,1,1,1} & \text{S} & \text{S} & \text{S} & \text{S} & \text{S} & \text{S} & \text{S} & \text{S} & \text{S} & \text{S} \\
s_{2,1,1,1} & \text{S} & \text{S} & \text{S} & \text{S} & \text{S} & \text{S} & \text{S} & \text{S} & \text{S} & \text{S} \\
s_{3,1,1} & \text{S} & \text{S} & \text{S} & \text{S} & \text{S} & \text{S} & \text{S} & \text{S} & \text{S} & \text{S} \\
s_{2,2,1} & \text{S} & \text{S} & \text{S} & \text{S} & \text{S} & \text{S} & \text{S} & \text{S} & \text{S} & \text{S} \\
s_{3,2} & \text{S} & \text{S} & \text{S} & \text{S} & \text{S} & \text{S} & \text{S} & \text{S} & \text{S} & \text{S} \\
s_{4_1,1} & \text{S} & \text{F} & \text{F} & \text{S} & \text{F} & \text{S} & \text{S} & \text{F} & \text{S} & \text{S} \\
s_{4_2,1} & \text{S} & \text{Y} & \text{Y} & \text{S} & \text{Y} & \text{S} & \text{S} & \text{Y} & \text{S} & \text{S} \\
s_{5_1} & \text{F} & \text{F} & \text{F} & \text{F} & \text{F} & \text{F} & \text{F} & \text{F} & \text{F} & \text{F} \\
s_{5_2} & \text{Y} & \text{Y} & \text{Y} & \text{S} & \text{S} & \text{S} & \text{S} & \text{S} & \text{S} & \text{F} \\
s_{5_3} & \text{Y} & \text{S} & \text{S} & \text{Y} & \text{S} & \text{Y} & \text{Y} & \text{S} & \text{Y} & \text{Y} \\
s_{5_4} & \text{Y} & \text{Y} & \text{Y} & \text{Y} & \text{Y} & \text{Y} & \text{Y} & \text{Y} & \text{Y} & \text{Y} \\
\hline
\end{array}$\\

Four region MMI\\
$\begin{array}{c||cccccccccccccccccccccccccccccc}
\hline
s_{1,1,1,1,1} & \text{S} & \text{S} & \text{S} & \text{S} & \text{S} & \text{S} & \text{S} & \text{S} & \text{S} & \text{S} & \text{S} & \text{S} & \text{S} & \text{S} & \text{S} & \text{S} & \text{S} & \text{S} & \text{S} & \text{S} & \text{S} & \text{S} & \text{S} & \text{S} & \text{S} & \text{S} & \text{S} & \text{S} & \text{S} & \text{S} \\
s_{2,1,1,1} & \text{S} & \text{S} & \text{S} & \text{S} & \text{S} & \text{S} & \text{S} & \text{S} & \text{S} & \text{S} & \text{S} & \text{S} & \text{S} & \text{S} & \text{S} & \text{S} & \text{S} & \text{S} & \text{S} & \text{S} & \text{S} & \text{S} & \text{S} & \text{S} & \text{S} & \text{S} & \text{S} & \text{S} & \text{S} & \text{S} \\
s_{3,1,1} & \text{S} & \text{S} & \text{S} & \text{S} & \text{S} & \text{S} & \text{S} & \text{S} & \text{S} & \text{S} & \text{S} & \text{S} & \text{S} & \text{S} & \text{S} & \text{S} & \text{S} & \text{S} & \text{S} & \text{S} & \text{S} & \text{S} & \text{S} & \text{S} & \text{S} & \text{S} & \text{S} & \text{S} & \text{S} & \text{S} \\
s_{2,2,1} & \text{S} & \text{S} & \text{S} & \text{S} & \text{S} & \text{S} & \text{S} & \text{S} & \text{S} & \text{S} & \text{S} & \text{S} & \text{S} & \text{S} & \text{S} & \text{S} & \text{S} & \text{S} & \text{S} & \text{S} & \text{S} & \text{S} & \text{S} & \text{S} & \text{S} & \text{S} & \text{S} & \text{S} & \text{S} & \text{S} \\
s_{3,2} & \text{S} & \text{S} & \text{S} & \text{S} & \text{S} & \text{S} & \text{S} & \text{S} & \text{S} & \text{S} & \text{S} & \text{S} & \text{S} & \text{S} & \text{S} & \text{S} & \text{S} & \text{S} & \text{S} & \text{S} & \text{S} & \text{S} & \text{S} & \text{S} & \text{S} & \text{S} & \text{S} & \text{S} & \text{S} & \text{S} \\
s_{4_1,1} & \text{S} & \text{F} & \text{S} & \text{F} & \text{S} & \text{S} & \text{S} & \text{F} & \text{F} & \text{S} & \text{F} & \text{F} & \text{S} & \text{F} & \text{S} & \text{F} & \text{S} & \text{S} & \text{S} & \text{F} & \text{F} & \text{S} & \text{F} & \text{S} & \text{S} & \text{S} & \text{F} & \text{S} & \text{S} & \text{S} \\
s_{4_2,1} & \text{S} & \text{Y} & \text{S} & \text{Y} & \text{S} & \text{S} & \text{S} & \text{Y} & \text{Y} & \text{S} & \text{Y} & \text{Y} & \text{S} & \text{Y} & \text{S} & \text{Y} & \text{S} & \text{S} & \text{S} & \text{Y} & \text{Y} & \text{S} & \text{Y} & \text{S} & \text{S} & \text{S} & \text{Y} & \text{S} & \text{S} & \text{S} \\
s_{5_1} & \text{F} & \text{F} & \text{F} & \text{F} & \text{F} & \text{F} & \text{F} & \text{F} & \text{F} & \text{F} & \text{F} & \text{F} & \text{F} & \text{F} & \text{F} & \text{F} & \text{F} & \text{F} & \text{F} & \text{F} & \text{F} & \text{F} & \text{F} & \text{F} & \text{F} & \text{F} & \text{F} & \text{F} & \text{F} & \text{F} \\
s_{5_2} & \text{S} & \text{S} & \text{S} & \text{S} & \text{S} & \text{S} & \text{Y} & \text{Y} & \text{Y} & \text{Y} & \text{Y} & \text{Y} & \text{S} & \text{S} & \text{S} & \text{S} & \text{S} & \text{S} & \text{Y} & \text{Y} & \text{Y} & \text{S} & \text{S} & \text{F} & \text{F} & \text{S} & \text{S} & \text{S} & \text{F} & \text{S} \\
s_{5_3} & \text{Y} & \text{S} & \text{Y} & \text{S} & \text{Y} & \text{Y} & \text{Y} & \text{S} & \text{S} & \text{Y} & \text{S} & \text{S} & \text{Y} & \text{S} & \text{Y} & \text{S} & \text{Y} & \text{Y} & \text{Y} & \text{S} & \text{S} & \text{Y} & \text{S} & \text{Y} & \text{Y} & \text{Y} & \text{S} & \text{Y} & \text{Y} & \text{Y} \\
s_{5_4} & \text{Y} & \text{Y} & \text{Y} & \text{Y} & \text{Y} & \text{Y} & \text{Y} & \text{Y} & \text{Y} & \text{Y} & \text{Y} & \text{Y} & \text{Y} & \text{Y} & \text{Y} & \text{Y} & \text{Y} & \text{Y} & \text{Y} & \text{Y} & \text{Y} & \text{Y} & \text{Y} & \text{Y} & \text{Y} & \text{Y} & \text{Y} & \text{Y} & \text{Y} & \text{Y} \\
\hline
\end{array}$\\
$S = \text{Saturates},\,Y = \text{Satisfies},\,F = \text{Fails}$.\\

The vectors $s_{5_1}$ and $s_{5_2}$ lie outside the MMI cone since each fails at least one instance of MMI. Vector $s_{5_3}$ lies inside the MMI cone, though residing on one of its faces. Vector $s_{5_4}$ lies clearly inside the MMI cone, saturating only the trivial instances of MMI which are saturated by all $5$-qubit stabilizer states.

\newpage

\subsubsection{Six Qubits}

At $6$ qubits there are $760$ stabilizer state entropy vectors, which arrange into $26$ equivalence symmetry classes up to qubit exchange. There are $15$ lifts of lower-qubit entanglement and $11$ new instances of $6$-qubit entanglement.
\begin{figure}[H]
\centering
\small
\begin{overpic}[width=0.97\textwidth,trim=0 0 0 0,clip]{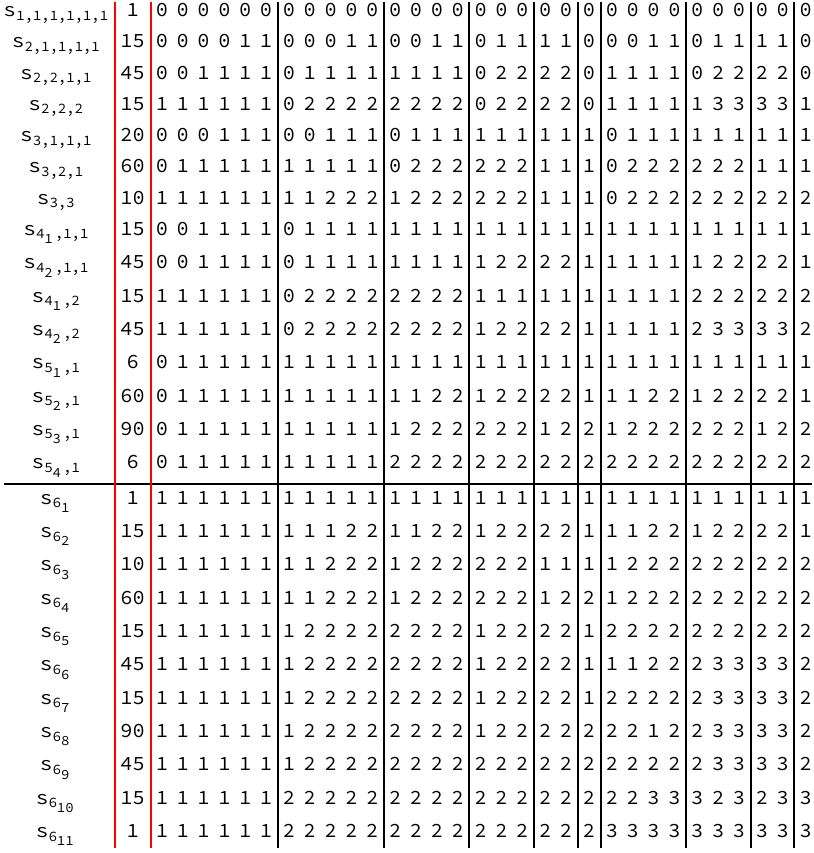}
    \put (20,100.8) {\scriptsize{A\,B\,C\,D\,E\,F}}
    \put (34,100.8) {\scriptsize{A\_}}
    \put (46,100.8) {\scriptsize{B\_}}
    \put (56,100.8) {\scriptsize{C\_}}
    \put (63,100.8) {\scriptsize{D\_}}
    \put (68,100.8) {\scriptsize{EF}}
    \put (71.5,100.8) {\scriptsize{AB\_}}
    \put (81,100.8) {\scriptsize{AC\_}}
    \put (88,100.8) {\scriptsize{AD\_}}
    \put (94,100.8) {\scriptsize{AEF}}
\end{overpic}
\caption{Representatives of the $26$ equivalence classes, up to qubit exchange symmetry, for all $6$-qubit stabilizer-state entropy vectors.}
\label{fig:SixQEntanglementTable}
\end{figure}
\newpage
\begin{figure}
\begin{center}
		\begin{overpic}[width = 1.\textwidth]{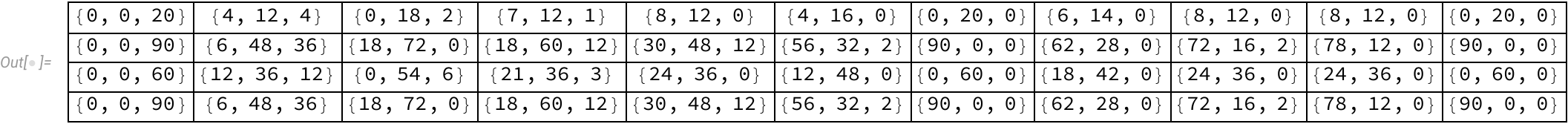}
        \put (.7,-2) {\scriptsize{$\{Y,S,F\}$}}
        \put (3,10) {\scriptsize{$S_{6_1}$}}
        \put (12,10) {\scriptsize{$S_{6_2}$}}
        \put (21,10) {\scriptsize{$S_{6_3}$}}
        \put (30.5,10) {\scriptsize{$S_{6_4}$}}
        \put (40,10) {\scriptsize{$S_{6_5}$}}
        \put (50,10) {\scriptsize{$S_{6_6}$}}
        \put (59,10) {\scriptsize{$S_{6_7}$}}
        \put (67,10) {\scriptsize{$S_{6_8}$}}
        \put (76,10) {\scriptsize{$S_{6_9}$}}
        \put (85,10) {\scriptsize{$S_{6_{10}}$}}
        \put (94,10) {\scriptsize{$S_{6_{11}}$}}
        \put (-8,12) {\scriptsize{MMI \textbackslash Vec.}}
        \put (-6,6.5) {\scriptsize{(1,1,1)}}
        \put (-6,4.5) {\scriptsize{(2,1,1)}}
        \put (-6,2.5) {\scriptsize{(3,1,1)}}
        \put (-6,.5) {\scriptsize{(2,2,1)}}
        \end{overpic}
    \caption{Table showing MMI satisfaction, saturation, and failure counts for all $6$-qubit entanglement structures given in Table \ref{fig:SixQEntanglementTable}. Each row is determined by the partition that defines the MMI lift, e.g. (3,1,1). Each cell contains the triple $\{Y,S,F\}$, where $Y$ counts the number of instances satisfied, $S$ the number saturated, and $F$ the number failed.}
    \label{SixQMMITable}
\end{center}
\end{figure}

\subsubsection{Seven Qubits}

At $7$ qubits there are $10773$ stabilizer state entropy vectors, which arrange into $59$ equivalence symmetry classes up to qubit exchange. There are $33$ lifts of lower-qubit entanglement and $26$ instances of $7$-qubit entanglement. The sizes of each $7$-qubit symmetry orbit are
\begin{equation}
\begin{split}
1, &1, 7, 7, 21, 21, 21, 21, 21, 21, 21, 30, 35, 35, 35, 35, 70, 70, \\
&70, 105, 105, 105, 105, 105, 105, 105, 105, 105, 105, 105, 105, 105, \\
&105, 105, 105, 105, 105, 210, 210, 210, 210, 315, 315, 315, 315, 315, \\
&315, 315, 315, 315, 315, 360, 420, 420, 630, 630, 630, 630, 630.\\
\end{split}
\end{equation}

There are $26$ new entanglement structures among the $7$-qubit stabilizer states. Each is shown in Figure \ref{SevenQNewEntanglement}, listed alongside the size of their respective symmetry orbits under qubit exchange.
\begin{figure}
\begin{center}
		\begin{overpic}[width = 1.\textwidth]{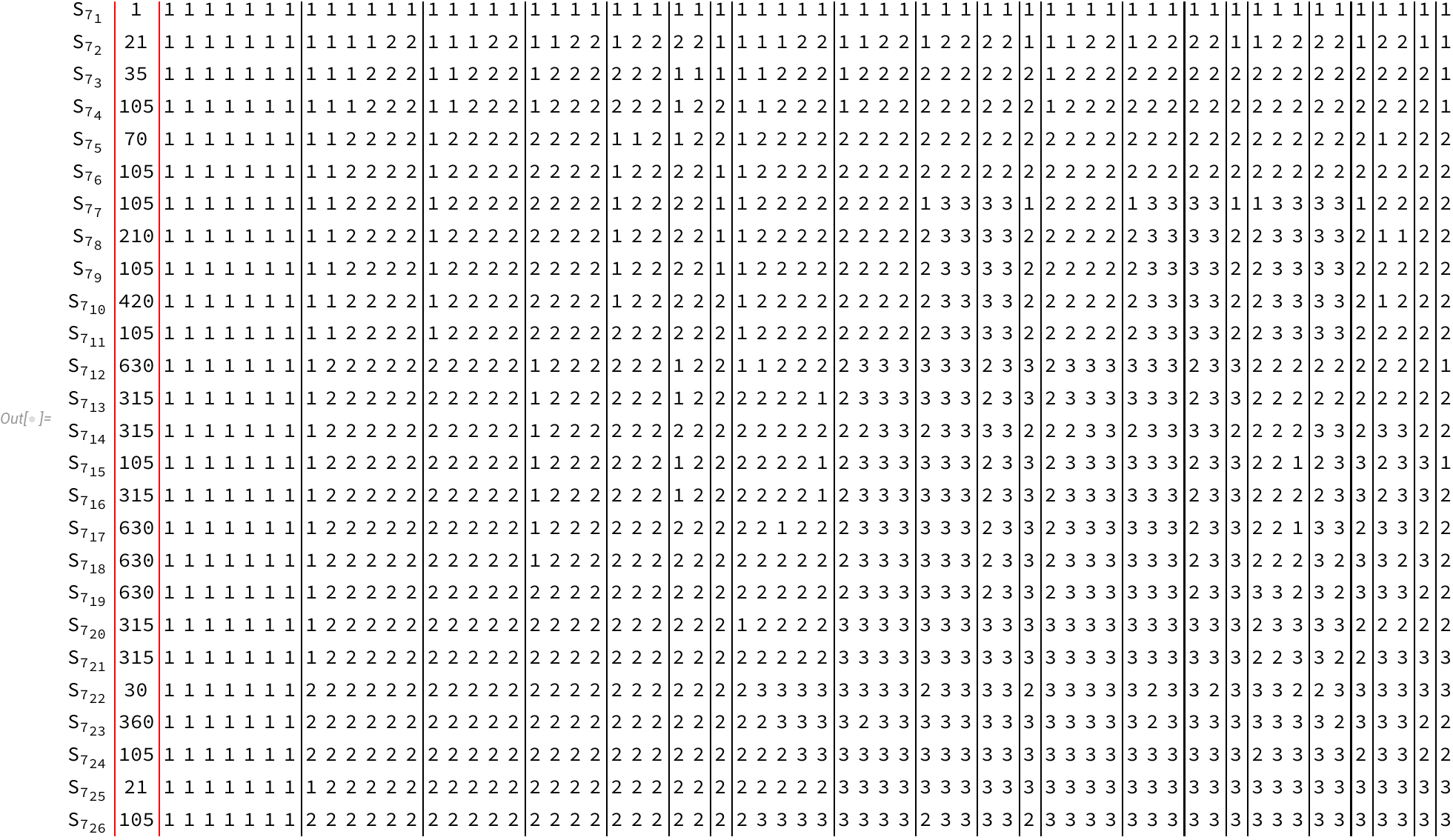}
        \put (7,61) {\tiny{A\,B\,C\,D\,E\,F\,G}}
		\put (19,61) {\tiny{A\_}}
        \put (28,61) {\tiny{B\_}}
        \put (34,61) {\tiny{C\_}}
        \put (40,61) {\tiny{D\_}}
        \put (44,61) {\tiny{E\_}}
        \put (49,61) {\tiny{AB\_}}
        \put (56,61) {\tiny{AC\_}}
        \put (61.5,61) {\tiny{AD\_}}
        \put (66,61) {\tiny{AE\_}}
        \put (71,61) {\tiny{BC\_}}
        \put (76.5,61) {\tiny{BD\_}}
        \put (80.5,61) {\tiny{BE\_}}
        \put (85.5,61) {\tiny{CD\_}}
        \put (89.5,61) {\tiny{CE\_}}
        \put (94,61) {\tiny{DE\_}}
        \put (99,61) {\tiny{EFG}}
        \put (-2,12.5) {\tiny{$\rightarrow$}}
        \end{overpic}
    \caption{Novel entanglement structures first appearing at $7$ qubits. Each entropy vector represents a symmetry orbit, under qubit exchange, the size of which is given in the second column.}
    \label{SevenQNewEntanglement}
\end{center}
\end{figure}

\begin{figure}
\begin{center}
		\begin{overpic}[width = 1.\textwidth]{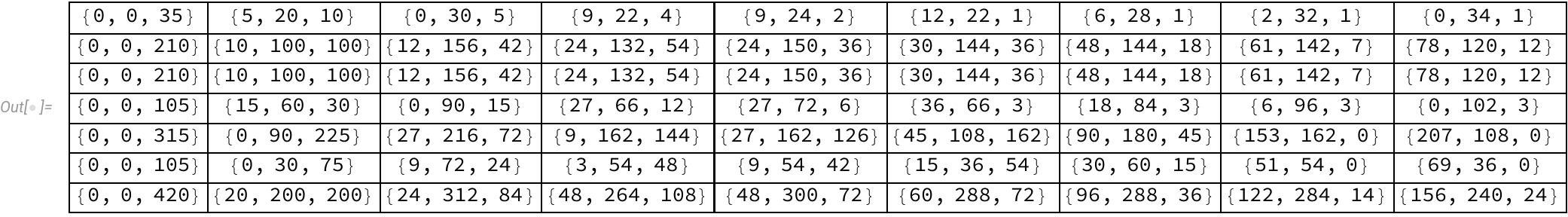}
        \put (.7,-2) {\scriptsize{$\{Y,S,F\}$}}
        \put (3,16) {\scriptsize{$S_{7_1}$}}
        \put (14,16) {\scriptsize{$S_{7_2}$}}
        \put (25,16) {\scriptsize{$S_{7_3}$}}
        \put (35.5,16) {\scriptsize{$S_{7_4}$}}
        \put (47,16) {\scriptsize{$S_{7_5}$}}
        \put (59,16) {\scriptsize{$S_{7_6}$}}
        \put (70,16) {\scriptsize{$S_{7_7}$}}
        \put (81,16) {\scriptsize{$S_{7_8}$}}
        \put (93,16) {\scriptsize{$S_{7_9}$}}
        \put (-9,18) {\scriptsize{MMI \textbackslash Vec.}}
        \put (-6,12.7) {\scriptsize{(1,1,1)}}
        \put (-6,10.7) {\scriptsize{(2,1,1)}}
        \put (-6,8.7) {\scriptsize{(3,1,1)}}
        \put (-6,6.7) {\scriptsize{(4,1,1)}}
        \put (-6,4.7) {\scriptsize{(2,2,1)}}
        \put (-6,2.7) {\scriptsize{(2,2,2)}}
        \put (-6,.7) {\scriptsize{(3,2,1)}}
        \end{overpic}    
    \label{SevenQMMITable1}
\end{center}
\end{figure}

\begin{figure}
\begin{center}
		\begin{overpic}[width = 1.\textwidth]{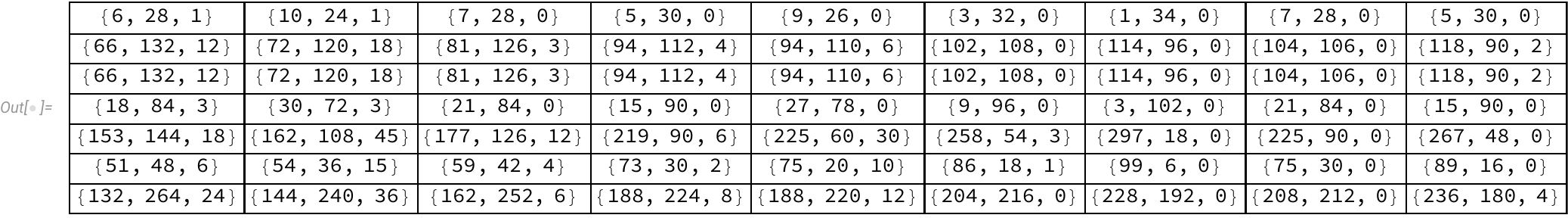}
        \put (2,-2) {\scriptsize{$\{Y,S,F\}$}}
        \put (4,16) {\scriptsize{$S_{7_{10}}$}}
        \put (15,16) {\scriptsize{$S_{7_{11}}$}}
        \put (27,16) {\scriptsize{$S_{7_{12}}$}}
        \put (37.5,16) {\scriptsize{$S_{7_{13}}$}}
        \put (49,16) {\scriptsize{$S_{7_{14}}$}}
        \put (60,16) {\scriptsize{$S_{7_{15}}$}}
        \put (71,16) {\scriptsize{$S_{7_{16}}$}}
        \put (81,16) {\scriptsize{$S_{7_{17}}$}}
        \put (92.5,16) {\scriptsize{$S_{7_{18}}$}}
        \put (-9,18) {\scriptsize{MMI \textbackslash Vec.}}
        \put (-6,12.7) {\scriptsize{(1,1,1)}}
        \put (-6,10.7) {\scriptsize{(2,1,1)}}
        \put (-6,8.7) {\scriptsize{(3,1,1)}}
        \put (-6,6.7) {\scriptsize{(4,1,1)}}
        \put (-6,4.7) {\scriptsize{(2,2,1)}}
        \put (-6,2.7) {\scriptsize{(2,2,2)}}
        \put (-6,.7) {\scriptsize{(3,2,1)}}
        \end{overpic}
    \label{SevenQMMITable2}
\end{center}
\end{figure}

\begin{figure}
\begin{center}
		\begin{overpic}[width = 1.\textwidth]{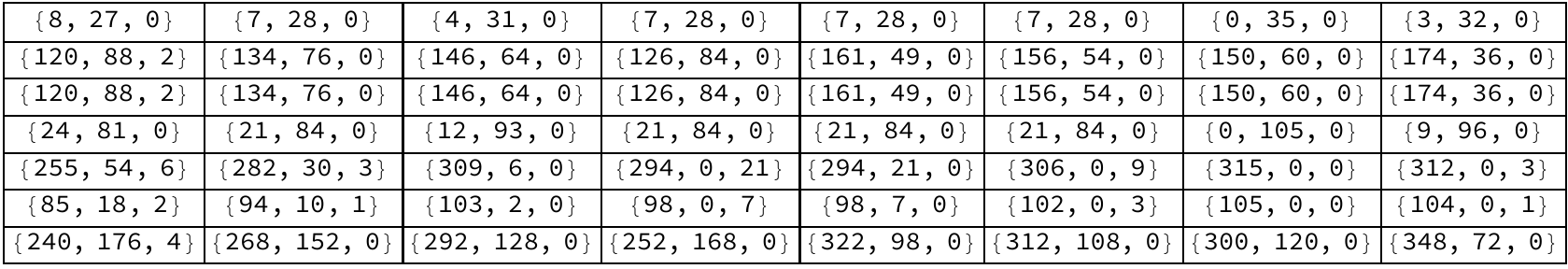}
        \put (.7,-2) {\scriptsize{$\{Y,S,F\}$}}
        \put (5,18.5) {\scriptsize{$S_{7_{19}}$}}
        \put (17,18.5) {\scriptsize{$S_{7_{20}}$}}
        \put (30,18.5) {\scriptsize{$S_{7_{21}}$}}
        \put (43,18.5) {\scriptsize{$S_{7_{22}}$}}
        \put (55,18.5) {\scriptsize{$S_{7_{23}}$}}
        \put (67,18.5) {\scriptsize{$S_{7_{24}}$}}
        \put (80,18.5) {\scriptsize{$S_{7_{25}}$}}
        \put (92,18.5) {\scriptsize{$S_{7_{26}}$}}
        \put (-9,20) {\scriptsize{MMI \textbackslash Vec.}}
        \put (-6,15.4) {\scriptsize{(1,1,1)}}
        \put (-6,12.8) {\scriptsize{(2,1,1)}}
        \put (-6,10.5) {\scriptsize{(3,1,1)}}
        \put (-6,8.2) {\scriptsize{(4,1,1)}}
        \put (-6,5.8) {\scriptsize{(2,2,1)}}
        \put (-6,3.4) {\scriptsize{(2,2,2)}}
        \put (-6,1.) {\scriptsize{(3,2,1)}}
        \end{overpic}
    \caption{MMI table for $7$-qubit entanglement structures given from Table \ref{SevenQNewEntanglement}. Each row is labeled by the partition that defines that MMI lift, e.g. (3,2,1). The value $Y$ counts the number of satisfied instances, $S$ the number saturated, and $F$ the number failed.}
    \label{SevenQMMITable3}
\end{center}
\end{figure}
%


\subsection{MMI Failure}

We present data on MMI failure in stabilizer states through $6$ qubits. We obtain a tableau representative for each stabilizer state, ignoring global phases, and compute the entropy vectors. The first three columns of Table~\ref{tab:MMIStateCountTable} give the number of tableaux that saturate every MMI instance, those that do not fail any instance and satisfy at least one, and those that fail at least one MMI instance, respectively. The last column gives the number of entropy vectors at each qubit number that correspond to MMI-failing states. We account for global phases by multiplying the obtained counts by $2^{n}$. All data is available and reproducible from \textit{Mathematica} notebooks in \cite{fuentes_munizzi2025mmi,githubStab}

\begin{table}[h]
    \centering
    \resizebox{\textwidth}{!}{
    \begin{tabular}{|c||c|c|c||c|}
    \hline
    & \multicolumn{3}{|c|}{MMI} & \\
    \hline
    Qubit \#& Satur. & Satis. & Fail & Ent. Vec. Fail\\
    \hline
    \hline
    3  & 1,080 (100\%) & 0 (0\%) & 0 (0\%) & (0\%)\\
    \hline
    4  & 18,576 (50.6\%) & 15,552 (42.36\%) & 2,592 (7.06\%) & 1 (5.55\%)\\
    \hline
    5  & 370,656 (15.29\%) & 1,648,512 (68\%) & 404,352 (16.70\%) & 16 (17.20\%)\\
    \hline
    6  & 9,118,656 (2.89\%) & 175,115,520 (55.59\%) & 130,823,424 (41.52\%) & 287 (37.76\%)\\
    \hline
    7  & - & - & - & 6436 (59.74\%) \\
    \hline
    \end{tabular}}
\caption{State and vector count for totally saturate, at least satisfy, and (at least $1$ instance) fail MMI.}
\label{tab:MMIStateCountTable}
\end{table}

\newpage
\bibliographystyle{JHEP}
\bibliography{bib}

\end{document}